\documentclass[aps,prx,twocolumn,notitlepage,nofootinbib,nobalancelastpage,unsortedaddress, superscriptaddress]{revtex4-2}
\pdfoutput=1
\usepackage[latin9]{inputenc}
\usepackage[T1]{fontenc}
\usepackage{amsmath,amsfonts,amssymb,amsthm,bbm,bm, braket}
\usepackage{wasysym}
\usepackage{MnSymbol}
\usepackage{comment}
\usepackage{scalerel}
\usepackage{simpler-wick}
\usepackage{times}
\usepackage{enumerate}
\usepackage{simplewick}
\usepackage{amssymb}
\usepackage{graphicx}
\usepackage{mathtools}
\usepackage{ifthen}
\usepackage{tensor}
\usepackage{tikz}
\usepackage{tikz-network}
\usetikzlibrary{patterns,decorations.pathreplacing}

\usepackage{color}
\usepackage[colorlinks=true,linkcolor=blue, citecolor=blue, urlcolor=blue, bookmarks]{hyperref}
\usepackage{longtable}
\usepackage[normalem]{ulem} 

\theoremstyle{break}        
\newtheorem{Lemma}{Lemma}
\usepackage{color}
\newtheorem{theorem}{Theorem}
\theoremstyle{plain}

\theoremstyle{plain}

\newcommand{\papertitle}{Optimizing digital quantum simulation of open quantum lattice models}
\newcommand{\vecket}[1]{|{#1} \rrangle}
\newcommand{\vecbra}[1]{\llangle {#1}|}
\newcommand{\vecbraket}[1]{\llangle {#1}\rrangle}
\theoremstyle{plain}

\newcommand{\abs}[1]{\left | {#1}\right |}
\newcommand{\norm}[1]{\left \Vert {#1}\right \Vert}

\begin{document}
\title{\papertitle}
\author{Xie-Hang Yu}
\affiliation{
Max-Planck-Institut f{\"{u}}r Quantenoptik, Hans-Kopfermann-Str. 1, 85748 Garching, Germany
}
\affiliation{
Munich Center for Quantum Science and Technology (MCQST), Schellingstr. 4, 80799 M{\"{u}}nchen, Germany
}

\author{Hongchao Li}
\affiliation{Department of Physics, The University of Tokyo, 7-3-1 Hongo, Tokyo 113-0033, Japan}


\author{J. Ignacio Cirac}
\affiliation{
Max-Planck-Institut f{\"{u}}r Quantenoptik, Hans-Kopfermann-Str. 1, 85748 Garching, Germany
}
\affiliation{
Munich Center for Quantum Science and Technology (MCQST), Schellingstr. 4, 80799 M{\"{u}}nchen, Germany
}

\author{Rahul Trivedi}
\email{rahul.trivedi@mpq.mpg.de}
\affiliation{
Max-Planck-Institut f{\"{u}}r Quantenoptik, Hans-Kopfermann-Str. 1, 85748 Garching, Germany
}
\affiliation{
Munich Center for Quantum Science and Technology (MCQST), Schellingstr. 4, 80799 M{\"{u}}nchen, Germany
}

\begin{abstract}
Many-body systems arising in condensed matter physics and quantum optics inevitably couple to the environment and need to be modelled as open quantum systems. While near-optimal algorithms have been developed for simulating many-body quantum dynamics, 
algorithms for their open system counterparts remain less well investigated. We address the problem of simulating geometrically local many-body open quantum systems interacting with a stationary Gaussian environment. Under a smoothness assumption on the system-environment interaction, we develop near-optimal algorithms that, for a model with $N$ spins and evolution time $t$, attain a simulation error $\delta$ in the system-state with $\mathcal{O}(Nt(Nt/\delta)^{o(1)})$ gates, $\mathcal{O}(t(Nt/\delta)^{o(1)})$ parallelized circuit depth and $\tilde{\mathcal{O}}(N(Nt/\delta)^{o(1)})$  ancillas. We additionally show that, if only simulating local observables is of interest, then the circuit depth of the digital algorithm can be chosen to be independent of the system size $N$. This provides theoretical evidence for the utility of these algorithms for simulating physically relevant models, where typically local observables are of interest, on pre-fault tolerant devices. Finally, for the limiting case of Markovian dynamics with commuting jump operators, we propose two algorithms based on sampling a Wiener process and on a locally dilated Hamiltonian construction, respectively. These algorithms reduce the asymptotic gate complexity on $N$ compared to currently available algorithms in terms of the required number of geometrically local gates.
\end{abstract}
\maketitle
\section{Introduction}
Quantum computers promise exponential-to-polynomial speed-ups for certain tasks that are otherwise considered to be hard on classical computers \citep{nielsen2002quantum}. Simulating many-body quantum systems, which is central to high and low-energy physics, is one of the tasks that quantum computers are naturally suited to solving \citep{Feynman_1982}. 
These systems are often described by geometrically local lattice models and display a rich variety
of collective behaviors \citep{Zhou2021Hightemperature,Stormer1999fractional,Broholm2020quantumspin,Qin2022Hubbard,Bulla2008numerical,Mitra2006nonequlibrium, spalek2007tj,mi2022time,Chen2022Errormitigated,Pal2010manybodylocalization,Vosk2015Theory,Finsterh2020nonequlibrium,Begg2024Quantum, Rothe_2012Lattice,Kogut1983lattice}.
Because such systems typically evolve into highly entangled states, simulating their time evolution with known classical algorithms requires resources that grow exponentially in both time and memory. Consequently, substantial effort has been put into designing quantum algorithms for simulating many-body systems.--- The state-of-the-art techniques, including the higher-order product Trotterization formulas \citep{childs2019nearly,Hatano2005Finding,Childs2021trotter}
or circuit approximation using Lieb-Robinson bounds together with quantum signal processing \citep{haah2021quantum}, now achieve provably nearly optimal scaling of both geometrically-local gate count and parallelized circuit depths with respect to system size and target precision.

However, in practice, most physical systems inevitably interact with their surrounding environment, and need to be described as an open quantum system \citep{Breuer2007theory}. The environment is often described by an infinite-dimensional Hilbert space, including the environment degrees of freedom substantially complicates both their classical and quantum simulation. The environment can even have a correlation time comparable to the time-scales of the many-body system, and introduce memory in the dynamics of the reduced state of the system \citep{Ferialdi2014general,Feruadku2016exact,trivedi2022descriptioncomplexitynonmarkovianopen}. Such memory effects have been observed to be of importance in solid-state physics\citep{Hadfield_Johansson_2016_superconducting, Chin2011generalized,Devega2008matter,groeblacher2015observation,Finsterh2020nonequlibrium, Leggett1987dynamics,Vojta2005quantum,Bulla2008numerical,Winter2009quantum}, quantum optics \citep{Calaj2019exciting,andersson2019non,gonzalez2019engineering,Leonforte2021vacancy} as well as quantum biology and chemistry \citep{Chin2012coherence,Ivanov2015extension,caycedo2022exact}. If the environment can be modelled as Markovian, i.e., its correlation time is significantly shorter than the intrinsic timescales of the system, the reduced dynamics of the system can be captured by a Lindbladian master equation \citep{Gorini1976completely,Gardiner2010quantumnoise} without tracking the environment degrees of freedom.

While it has been established that geometrically local open quantum systems, both in the Markovian and non-Markovian regime, can be simulated in time scaling at-most polynomially with the system size on quantum computers \citep{DissipativeTuring2011Kliesch,cleve2019efficientquantumalgorithmssimulating,Andrew2017sparseLindbladian}, there still remains a gap between the simulation time needed for open systems as compared to closed systems. In part, this is due to the fact that open-system dynamics are fundamentally irreversible i.e., while the superoperator describing the evolution of an open system forward in time is a valid quantum channel, its inverse (if it exists) is not necessarily a valid quantum channel. This creates a significant difficulty in adapting the existing algorithmic tools \cite{Efficient2006Osborne,childs2019nearly, Childs2021trotter, haah2021quantum} that yield near-optimal run-times for lattice Hamiltonians to the open-system setting since these tools often rely on carefully designing forward and backward evolution corresponding to parts of the lattice Hamiltonian.

Despite these difficulties, there has been recent progress on optimizing the run-time for simulating open systems. Markovian open systems, that can be modelled by a Lindblad master equation, have been shown to be quantum simulable in time that scales \emph{linearly} with the physical evolution time \citep{cleve2019efficientquantumalgorithmssimulating,li2023simulatingmarkovianopenquantum}, which is optimal due to no-fast forwarding theorems for quantum dynamics \cite{Berry2015nearlylinear}. 
However, these approaches rely on the access to oracles for both the system Hamiltonian and the
jump operators in order to use the quantum signal processing toolbox. For geometrically local models, this results in a large polynomial overhead in the required gate count with respect to the system size, especially if the gates are also restricted to be geometrically local \citep{li2023simulatingmarkovianopenquantum, childs2018toward}. Furthermore, for Markovian open systems, approaches such as Trotterization
with higher order product formula \cite{Han2021experimental, borras2024quantumalgorithmsimulatelindblad}, which for geometrically local lattice models implicitly yield algorithms with geometrically local gates, are limited to only the second order in time since at higher-order they necessarily require a backward time evolution \citep{Suzuki:1991jtk,Hatano2005Finding}. For non-Markovian models, where a closed form dynamical equation for the system evolution is generally not available beyond the perturbative regime \citep{Ivanov2015extension, Pereverzev2006Time}, we need to explicitly keep track of the full environment. Classical algorithms for simulating such systems have used either influence functional methods \citep{FEYNMAN1963theorygeneral,Breuer2007theory,strathearn2018efficient,Marten2022,Valentin2024,Gribben2022,zhang2024timeevolvingmatrixproduct}
which formally integrate out the environment but obtain a time non-local evolution, or the use of star-to-chain and pseudo-mode approximations \citep{Chin2010exact,trivedi2022descriptioncomplexitynonmarkovianopen,Garraway1997nonperturbative,Dalton2001pseudomodes,Trivedi2021convergence}
where only the ``important" environment modes are taken into account. While these methods have also been used to derive quantum algorithms for non-Markovian many-body systems \cite{trivedi2022descriptioncomplexitynonmarkovianopen, li2023succinct}, there is still a large gap between the resulting run-times and the run-times for algorithms available for simulating geometrically local closed systems.

In this paper, we present several results on simulating both geometrically local non-Markovian and Markovian open systems which obtain near-optimal run-times i.e., for geometrically local open systems on $N$ qubits evolved for time $t$, we provide algorithms to simulate them with close to $\mathcal{O}(Nt)$ geometrically local gates (corresponding to close to $\mathcal{O}(t)$ circuit depth) and with close to $\mathcal{O}(N)$ ancillary qubits. For non-Markovian models with Gaussian and stationary environments, our results are based on establishing a set of higher order Trotter formulae on the full system-environment Hamiltonian. We also specifically consider the task of simulating local observables instead of the full many-body state, and using Lieb-Robinson bounds for such models \cite{trivedi2024liebrobinsonboundopenquantum}, show that the circuit depth can even be chosen to be independent of the system size $N$ and dependent only on the evolution time $t$. For Markovian models, we consider the case of commuting jump operators (but with a possibly non-commuting local Hamiltonian) and develop two quantum algorithms --- the first one achieving the nearly optimal scaling of $\mathcal{O}(Nt)$ geometrically local gates when the jump operators are assumed to be Hermitian; and the second algorithm with non-Hermitian jump operators attaining the same
scaling as the third-order Trotterization and thus going beyond Trotter formulae that only require forward time evolution.

This paper is structured as follows: In section \ref{sec:Models},
we review the key concepts of open-system dynamics and
summarize our main results.
In section \ref{sec:Proof-outline-of}, we outline the proof strategy --- subsections \ref{subsec:Proof-of-Lemma}-\ref{subsec:Proof-of-Theorem-non-dissipative} address the algorithms for the non-Markovian models, in subsection \ref{sub:proof_of_error_robustness} we show how to optimize the algorithm runtime when only considering local observables, and in subsections \ref{subsec:Proof-of-Observation}-\ref{subsection:proof_markovian_non}
we describe the algorithms for the Markovian model.

\section{Summary of Results\label{sec:Models}}
\subsection{Model and preliminaries}

\begin{figure}
\includegraphics[width=0.8\columnwidth]{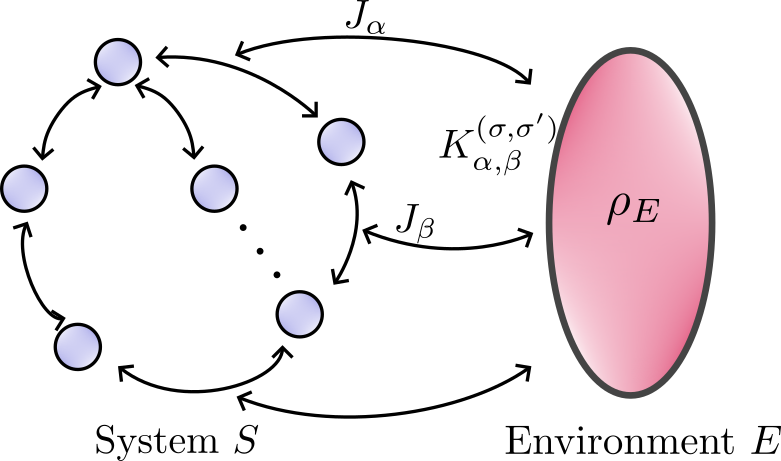}

\caption{Schematic of the general non-Markovian quantum dynamics model. The circle represents the physical sites. $J_\alpha, J_\beta$ are the jump operators and $K_{\alpha,\beta}^{(\sigma,\sigma')}$ is the two-point correlation function.}

\label{Fig:general_structure}
\end{figure}

A general model of an open system comprises of a system $S$ interacting with an environment $E$ as shown in Fig.~\ref{Fig:general_structure}. The joint system-environment dynamics can be modeled by a possibly time-dependent Hamiltonian expressed in the interaction picture with respect to the environment Hamiltonian:
\begin{equation}
H_{SE}(t)=H_{S}+V_{SE}(t),\label{eq:non_Markovian_general_begin}
\end{equation}
where $H_S$ is the Hamiltonian describing the internal dynamics of the system and $V_{SE}(t)$ is the Hamiltonian describing the system-environment interaction. The system-environment interaction Hamiltonian can generally be expressed as:
\begin{equation}
V_{SE}(t)=\sum_{\alpha}J^\dagger_{\alpha} A_{\alpha}(t) + \text{h.c.},
\end{equation}
where $J_{\alpha}$ and $A_{\alpha}(t)$ are system and environment
operators respectively. We will also make the physically reasonable assumption that at $t=0$, the system and environment are in a product state i.e., $\rho_{SE}(0)=\rho_S(0)\otimes\rho_{E}(0)$. In most physically-relevant scenarios, it is of interest to only monitor the system dynamics which are described by the channel $\mathcal{E}_S(t)$
\begin{align}\label{eq:channel}
    \mathcal{E}_S(t)(\cdot) = \text{Tr}_E\big(U_{SE}(t, 0)((\cdot) \otimes \rho_E(0))U_{SE}(0,t) \big),
\end{align}
where $U_{SE}(t, s) = \mathcal{T}\exp(-i\int_s^t H_{SE}(\tau) d\tau)$. It follows from the Dyson expansion for $U_{SE}(t, s)$ that $\mathcal{E}_S(t)$ is entirely determined by the multi-point correlation
functions of the environment \citep{huang2024unifiedanalysisnonmarkovianopen}\footnote{Throughout this paper, we will use the following convention for products:
\[
\prod_{i = 1}^n O_i = O_n O_{n - 1} \dots O_1 \text{ and } \prod_{i = n}^1 O_i = O_1 O_2 \dots O_n .
\]
}:
\begin{equation}
\begin{aligned}
\mathcal{C}^{(\sigma_1, \sigma_2 \dots \sigma_n)}_{\alpha_{1},\alpha_2\cdots \alpha_{n}}(t_{1}, t_2 \cdots t_{n})=\mathrm{Tr}_E\left(\rho_{E}(0) \prod_{i = n}^1 A_{\alpha_{i}}^{(\sigma_i)}(t_{i})\right)\label{eq:general_model_multi_correlators},
\end{aligned}
\end{equation}
where $\sigma_i \in \{+, -\}$ and, for an operator $O$, $O^{(-)} = O$, $O^{(+)} = O^\dagger$. For general environments, we would in principle need to know all the correlators in Eq. (\ref{eq:general_model_multi_correlators}) in order to describe and compute the system dynamics. However, for most physical systems, it is reasonable to make two simplifying assumptions on the environment:  it is \emph{stationary} and \emph{Gaussian}. Stationarity of the environment is equivalent to requiring that the correlators in Eq.~\ref{eq:general_model_multi_correlators} are time-translation independent i.e., $\forall \tau \in \mathbb{R}$:
\begin{align}
&\mathcal{C}^{(\sigma_1, \sigma_2 \dots \sigma_n)}_{\alpha_{1}, \alpha_2 \cdots \alpha_{n}}(t_{1},t_2 \cdots t_{n})\nonumber\\
&\qquad  =\mathcal{C}^{(\sigma_1, \sigma_2 \dots \sigma_n)}_{\alpha_{1},\alpha_2 \cdots \alpha_{n}}(t_{1}+\tau,t_2+\tau \cdots t_{n}+\tau).
\end{align}
Furthermore, Gaussianity of the environment requires that the correlators $\mathcal{C}^{(\sigma_1,\sigma_2 \dots \sigma_n)}_{\alpha_1, \alpha_2\dots \alpha_n}(t_1, t_2 \dots t_n)$ satisfy the Wick's theorem:
\begin{align}
&\mathcal{C}_{\alpha_{1}, \alpha_2\cdots\alpha_{m}}^{(\sigma_1, \sigma_2 \dots \sigma_m)}(t_{1}, t_2, \cdots t_{m}) =  \nonumber \\
&\qquad \qquad \sum_{\mathcal{S}} \prod_{(i, j) \in \mathcal{S}} \mathcal{C}^{(\sigma_i, \sigma_j)}_{\alpha_i, \alpha_j}(t_i, t_j) \prod_{k \in \mathcal{S}^c} \mathcal{C}^{(\sigma_k)}_{\alpha_k}(t_k), 
\label{eq:wick_contraction_expression}
\end{align}
where the sum is taken over all $\mathcal{S}\subseteq\{1, 2 \dots n\}$ with even number of elements that are divided into pairs $(i, j)$ with $i < j$. Furthermore, for stationary environments, $\mathcal{C}^{(\sigma)}_{\alpha}(t) = \mathcal{C}^{(\sigma)}_\alpha(0)$.--- Since we can always transform $H_S \to H_S + \sum_\alpha \mathcal{C}^{(\sigma)}_\alpha(0) J_\alpha^{(1- \sigma)}$ and $A^{(\sigma)}_\alpha(t) \to A^{(\sigma)}_\alpha(t) - \mathcal{C}^{(\sigma)}_\alpha(0) I$, without loss of generality, we can assume that $\mathcal{C}_\alpha(0) = 0$. Thus, an open system with a Stationary and Gaussian environment is specified entirely by the system Hamiltonian $H_S$, the system operators $J_\alpha$ which we will call the ``jump operators", and the two-point correlation function $\mathcal{C}^{(\sigma_1, \sigma_2)}_{\alpha_1,\alpha_2}(t_1, t_2) \equiv K_{\alpha_1, \alpha_2}^{(\sigma_1, \sigma_2)}(t_1-t_2)$ which we call the ``memory kernel". We remark that if we assume that $K_{\alpha_1, \alpha_2}^{(\sigma_1, \sigma_2)}(t_1-t_2)\propto \delta(t_1-t_2)$, then we recover the Markovian model of open quantum systems where the system dynamics is described by a Lindblad master equation \cite{nielsen2002quantum}.


In this work, we will focus on geometrically local open quantum systems on a $D$-dimensional lattice. For simplicity, we consider the $D=1$ model as schematically depicted in Fig.~\ref{Fig:fig1}, although our results hold for higher-dimensional lattice models. We consider a model with $N$ qubits in $1$D where the system Hamiltonian $H_S$ is a nearest neighbour Hamiltonian:
\begin{subequations}\label{eq:lattice_model}
\begin{align}
    H_S = \sum_{i=1}^{N-1} H_{i, i+1},\label{eq:sec2HSlattice}
\end{align}
and the system-environment interaction also involves at-most nearest neighbour system  operators,
\begin{align}\label{eq:lattice_mode_VSE}
    V_{SE}(t) = \sum_{i = 1}^{N-1} J^\dagger_{i, i+1} A_i(t) +\text{h.c.},
\end{align}
\end{subequations}
and we also assume that the environment is also locally independent i.e., 
\begin{equation}
K_{i, j}^{(\sigma, \sigma')}(\tau) = K_{i}^{(\sigma, \sigma')}(\tau)\delta_{i, j}. \label{eq:independent_bath_formula_maintext}
\end{equation}

To develop a quantum algorithm for simulating such models, it will be more convenient to work with an explicit description of the environment. Indeed, for given $K_i^{(\sigma,\sigma')}$, it is always possible to choose an explicit description of the environment which gives rises to the same memory kernels, and thus, the same reduced dynamics on the system. More specifically, suppose that the environment operators $A_i(t)$ were given by
\begin{align}\label{eq:Ai_exp_rep}
    A_i(t)= \int_{-\infty}^\infty v_i(t - s) a_{i, s}ds,
\end{align}
where we will choose $v_i$ later and $a_{i, s}$ are bosonic annihilation operators satisfy the commutation relations $[a_{i, s}, a_{i', s'}] = 0, [a_{i, s}, a^\dagger_{i', s'}] = \delta_{i, i'}\delta(s -s')$. The operators $A_i$, in this representation, are determined by the functions $v_i$ which we call ``coupling functions" throughout this paper. Furthermore, suppose that the initial environment state was given by $\rho_E(0) = \otimes_{i = 1}^{N - 1} \rho_{i, E}(0)$, where $\rho_{i, E}(0)$ are Gaussian states with $\text{Tr}(a_{i, s}\rho_{i, E}(0)) = 0$ and specified by $B_i, G_i$ via
\begin{align}\label{eq:initial_state_rep}
&B_i(s - s') = \text{Tr}(a_{i, s} a_{i, s'}^\dagger \rho_E(0)) \text{ and }\nonumber\\
&G_i(s - s') = \text{Tr}(a_{i, s} a_{i, s'}\rho_E(0)).
\end{align}
If $v_i, B_i, G_i$ are chosen such that
\[
\begin{aligned}
&\hat{K}_i^{(-, +)}(\omega) = (2\pi)^2|\hat{v}_i(\omega)|^2\hat{B}_i(\omega),\\ &\hat{K}^{(-,-)}_i(\omega) = \hat{K}_i^{(+, +)*}(\omega) =(2\pi)^2\hat{v}_i(\omega)\hat{v}_i(-\omega)\hat{G}_i(\omega),\\
&\hat{K}_i^{(-,+)}(\omega)-\hat{K}_i^{(+,-)}(-\omega)=(2\pi)^2|\hat{v}_i(\omega)|^2,
\end{aligned}
\]
where $\hat{f}(\omega) = \int_{-\infty}^{\infty} f(\tau) e^{i\omega \tau} d\tau /2\pi$, then $A_i(t)$ in Eq.~(\ref{eq:Ai_exp_rep}) together with the initial states in Eq.~\eqref{eq:initial_state_rep} provide an explicit representation of the environment with memory kernels Eq. (\ref{eq:independent_bath_formula_maintext}). In the main text, for simplicity, we will assume that $\rho_{E}(0) = \ket{\text{Vac}}\!\bra{\text{Vac}}$ or, equivalently, $ G_i(s, s') = 0$ and $B_i (s, s') = \delta(s-s')$ i.e., the environment is initially in the vacuum state. In the Supplemental Material \citep{SM}, we provide extensions of our results to any easily preparable $\rho_E(0)$ with integrable two-point correlators.

\begin{figure}
\includegraphics[width=1\columnwidth]{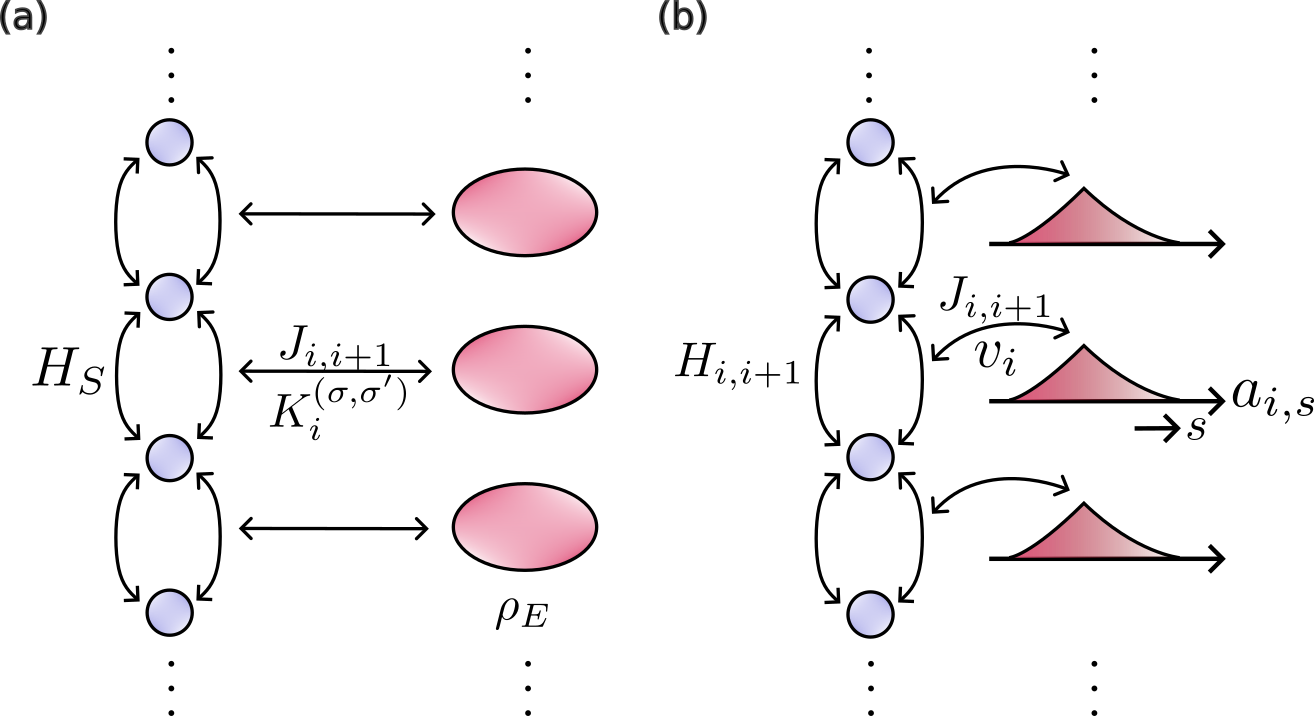}

\caption{(a)The non-Markovian dynamics on a lattice model with geometrically local interactions and independent environments. (b)An explicit description of the environment. Here, the thick solid line
depicts the continuum of the bosonic modes and the shaded envelope illustrates
the temporal coupling function $v_i(t-s)$.}

\label{Fig:fig1}
\end{figure}


In this work, our goal would be to approximate the channel describing the system dynamics [Eq. (\ref{eq:channel})] on a quantum computer for the lattice model [Eq.~\eqref{eq:lattice_model}].
More specifically, we will design a quantum circuit, with as few gates as possible, that implements a channel $\mathcal{E}$ such that for any initial state of the system $\rho_S(0)$
\begin{equation}
\begin{aligned}  \left\Vert{\mathcal{E}_S(t) (\rho_S(0)) - \mathcal{E}(\rho_S(0))}\right \Vert_\text{tr}
\leq  \delta,
\end{aligned}
\label{eq:error_requirement}
\end{equation}
where $\delta$ is the desired accuracy we want to achieve and $\lVert\cdot\rVert_{\mathrm{tr}}$ is the trace norm \citep{Watrous2018theory}. We will provide results for both cases: (i) the kernel $K_i^{\sigma,\sigma'}(\tau)$ is a smooth function of $\tau$ which $\to 0$ as $\tau \to \infty$: in this case, the dynamics of the system would be non-Markovian and not describable by a Lindblad master equation. (ii) $K_i^{\sigma,\sigma'}(\tau) \propto \delta(\tau)$: in this case, the system dynamics is Markovian and described by a Lindblad master equation.

\subsection{Summary of results\label{sec:Summary-of-results}}

\subsubsection{Simulation of Non-Markovian dynamics\label{subsec:Simulation-of-Non-Markovian}}

\emph{Higher-order Trotterization}. Trotterization is a powerful tool for simulating dynamics of lattice Hamiltonians,--- given a Hamiltonian $H$ expressed as a sum of two operators, $H = A + B$, the $P^\text{th}$ order Trotter formula expresses the exponential $\exp(-i\Delta t H)$ as a product of exponentials of $A$ and $B$, with a residual error of $\mathcal{O}(\Delta t^{P + 1})$. More specifically, a $P^\text{th}$ order Trotter formula for $e^{-i\Delta t H}$ with $s_P$ stages is of the form
\begin{align}\label{eq:trotter_formula}
   \mathcal{S}_P(\Delta t) = \prod_{i = 1}^{s_P}e^{-i \mu_{i} B \Delta t} e^{-i  \varepsilon_{i} A \Delta t},
\end{align}
where $\varepsilon_1, \varepsilon_2 \dots \varepsilon_{s_P}, \mu_1, \mu_2 \dots \mu_{s_P} \in \mathbb{R}$ are chosen such that $\exp(-i H \Delta t) =\mathcal{S}_P(\Delta t) + \mathcal{O}(\Delta t^{P + 1})$. To use Trotter formulae to simulate geometrically local spin Hamiltonians, e.g., a 1D lattice model $H = \sum_{i = 1}^{N - 1} H_{i, i+1}$, we can choose $A := H_\mathrm{o} = H_{1,2} + H_{3, 4} + \dots $ and $B := H_\mathrm{e} = H_{2, 3} + H_{4, 5} + \dots$ in which case $\mathcal{S}_P(\Delta t)$ directly yields a quantum circuit approximating $\exp(-iH \Delta t)$. In Ref.~\cite{childs2019nearly}, by explicitly exploiting the geometrical locality of Hamiltonian, it was shown that $\left\Vert{\exp(-iH \Delta t) - \mathcal{S}_P(\Delta t)}\right \Vert \leq \mathcal{O}(N \Delta t^{P + 1})$. An immediate implication of this error bound was that higher order Trotter formulae could be used to approximate $\exp(-iH t)$ to an error $\delta$ with $\mathcal{O}(Nt(Nt/\delta)^{1/P})$ geometrically local gates.

In our first lemma, we provide a Trotterization procedure for the open-system model in Eq.~\eqref{eq:lattice_model} and derive an explicit upper bound on the Trotter error. While the system-environment model that we consider is a geometrically local Hamiltonian, there are two reasons why the previous Trotter analysis in Refs.~\cite{childs2019nearly, Childs2021trotter} is inapplicable.--- \emph{First}, the system-environment Hamiltonian is unbounded, while the previous Trotter analysis is restricted to bounded Hamiltonians. \emph{Second}, the system-environment Hamiltonian as presented in Eq.~\eqref{eq:lattice_model} is time-dependent,--- directly using Trotter formulae for time-independent Hamiltonians for time-dependent models typically doesn't yield the correct error scaling. Indeed, additional smoothness assumptions on the time-dependent terms in the Hamiltonian are required to systematically develop higher order Trotter formulae \citep{Hatano2005Finding}. Nevertheless, we show that for stationary and Gaussian environments, as long as the kernels $K_{i}^{(\sigma, \sigma')}(\tau)$ have bounded $L^1$ and $L^\infty$ norms i.e., $\exists m, M > 0$ such that $\forall \sigma, \sigma', i$,
\begin{align}\label{eq:memory_kernel_bounded_trotter}
\sup_{\tau \in \mathbb{R}} \vert K_i^{(\sigma, \sigma')}(\tau) \vert < m, \int_{-\infty}^\infty \vert K_i^{(\sigma, \sigma')}(\tau) \vert d\tau < M,
\end{align}
then the open system lattice model can be Trotterized with a similar error scaling as that of Hamiltonian models. 

More specifically, we will decompose the total system-environment Hamiltonian $H_{SE}(t) = H_\mathrm{o}(t) + H_\mathrm{e}(t)$ where
\begin{align*}
&H_\mathrm{o}(t) = \sum_{i \in \text{odd}} H_{i, i + 1} + (J_{i,i+1}^{\dagger}A_{i}(t)+\text{h.c.}), \\
&H_\mathrm{e}(t) = \sum_{i \in \text{even}} H_{i, i + 1} + (J_{i,i+1}^{\dagger}A_{i}(t)+\text{h.c.}),
\end{align*}
and define the corresponding system-environment unitaries: $V_\mathrm{e}(t, s) = \mathcal{T}\exp(-i \int_s^t H_\mathrm{e}(\tau) d\tau)$ and $V_\mathrm{o}(t, s) = \mathcal{T}\exp(-i \int_s^t H_\mathrm{o}(\tau) d\tau)$. We divide the total evolution time $t$ into $T$ time steps and construct the Trotterized system-environment unitary $U_\text{tro}(t, T)$ via
\begin{subequations}\label{eq:trotterization_expression}
\begin{align}
U_\text{tro}(t, T) = \prod_{j =0}^{T-1} V\bigg((j + 1) \frac{t}{T}, j\frac{t}{T}\bigg),
\end{align}
where the unitary $V(j t/T, (j - 1) t/T)$ corresponding to the $j^\text{th}$ time-step is constructed from the $P^\text{th}$ order Trotter formula (Eq.~\eqref{eq:trotter_formula}) as follows:
\begin{align}
  &V\bigg((j + 1) \frac{t}{T}, j\frac{t}{T}\bigg)  = \prod_{i = 1}^{s_P} \bigg[V_\mathrm{e}\bigg((f_{i}+j) \frac{t}{T}, (f_{i - 1}+j) \frac{t}{T} \bigg) \times \nonumber \\
  &\qquad \qquad \qquad V_\mathrm{o}\bigg((e_{i}+j) \frac{t}{T}, (e_{i - 1}+j) \frac{t}{T}\bigg)\bigg],\label{eq:Sec2trotterexpression}
\end{align}
\end{subequations}
with $e_{i} = \sum_{k = 1}^i \varepsilon_k, f_{i} =\sum_{k = 1}^i \mu_k$ where $\varepsilon_k, \mu_k$ are determined the $P^\text{th}$ order Trotter formula used in the Trotterization (Eq.~\ref{eq:trotter_formula}). To characterize the Trotterization error, we consider the channel on the system $\mathcal{E}_\text{tro}(t, T)$ given by
\[
\mathcal{E}_\text{tro}(t, T)(\cdot) = \text{Tr}_E\big[U_\text{tro}(t, T)((\cdot) \otimes \rho_E(0))U_\text{tro}^\dagger(t, T)\big].
\]
The lemma below provides an upper bound on the error between the Trotterized channel $\mathcal{E}_\text{tro}(t, T)$ and the channel $\mathcal{E}_S(t)$ defined in Eq.~\eqref{eq:channel} corresponding to the exact system evolution.


\begin{Lemma}[Trotterization of open system dynamics with bounded Kernels]

If the kernels $K_i^{(\sigma, \sigma')}(\tau)$ satisfy Eq. (\ref{eq:memory_kernel_bounded_trotter}) with $m, M$ being independent of the number of spins $N$, then for any initial system state $\rho_S(0)$
\[
\left \Vert \mathcal{E}_\textnormal{tro}(t, T)(\rho_S(0)) - \mathcal{E}_S(t)(\rho_S(0))\right \Vert_\mathrm{tr} \leq\mathcal{O}\bigg((Ps_P)^P  Nt\bigg(\frac{t}{T}\bigg)^P\bigg)
\]
for any Trotter formula (Eq.~\eqref{eq:trotter_formula}) of order $P$ and with $s_P$ stages. \label{lemma1}
\end{Lemma}
\noindent We reemphasize that Eq.~\eqref{eq:trotterization_expression} is \emph{not} the Trotter formula for a generic time-dependent Hamiltonian and if applied to the generic case would only be accurate to $\mathcal{O}((t/T)^2)$. Nevertheless, by exploiting the stationarity of the Gaussian environment, we can show that Eq.~\eqref{eq:trotterization_expression} correctly Trotterizes the system-environment Hamiltonian to any desired order. We also remark that while Lemma \ref{lemma1} achieves the same Trotter error scaling in $N, t, T$ as the one for closed lattice models, its dependence on $P$ is worse than that for closed systems Ref. \citep{childs2019nearly}. Indeed, the constant prefactor in Lemma \ref{lemma1} scales as $\mathcal{O}(P^Ps_P^P)$ while that in Ref. \citep{childs2019nearly} grows as $\exp(\mathcal{O}(P))\times \mathcal{O}(s_P^P)$. This difference can be physically attributed to the the infinite-dimensionality of the bosonic environment, and the unbounded nature of the system-environment interaction.


We next apply this Trotter formula to develop an explicit algorithm for simulating the open lattice model in Eq.~\eqref{eq:lattice_model}. Our first result requires an assumption on the memory kernels of the open system model which is easiest to state in terms of the coupling function $v_i$ introduced in Eq.~\eqref{eq:Ai_exp_rep}. We assume that $v_i(t)$ is a smooth function of $t$ which decays superpolynomially as $\abs{t}\to \infty$ and has bounded derivatives i.e., 
\begin{align}\label{eq:smoothness_assumption}
    \exists C_{\nu}, D_\mu > 0: \sup_{t\in \mathbb{R}}\left \vert  t^\nu  v_i(t)\right \vert \leq C_{\nu} \text{ and }\sup_{t\in \mathbb{R}} \abs{\frac{\partial^\mu}{\partial t^\mu} v_i(t)} \leq D_\mu.
\end{align}
We remark that this assumption also implies that Eq.~\eqref{eq:memory_kernel_bounded_trotter} holds, and consequently we can use lemma \ref{lemma1} for models satisfying this assumption. Physically, the decay of $v_i(t)$ with $t$ implies that the environment should not retain information for a very long time, a condition which we expect to be satisfied in most physically relevant open system models. Furthermore, the smoothness of $v_i(t)$ has a physical consequence that the system-environment coupling has a smooth high-frequency cutoff. On physical grounds, we would expect the dynamics of an open system to be largely governed by environment modes that have frequencies around the dominant energies in the system. Consequently it is reasonable to assume a smooth high-frequency cutoff on the system-environment coupling in most physically relevant models.

\begin{theorem}[Simulation of general non-Markovian dynamics]\label{theoremdi}
If the coupling functions $v_i$ satisfy Eq.~\eqref{eq:smoothness_assumption}, then the channel on the system qubits  $\mathcal{E}_S(t)$ [Eq.~\eqref{eq:channel}] can be approximated within an error $\delta$ with $\mathcal{O}(Nt(Nt/\delta)^{1/p})$ geometrically local gates and $\tilde{\mathcal{O}}(N(Nt/\delta)^{1/(p+1)})$ \footnote{The notation $\tilde{\mathcal{O}}$ omits polylogarithmic factors.}
ancillary qubits for any user-specified $p > 0$.
\end{theorem}
\noindent This result establishes that, under the assumption in Eq.~\eqref{eq:smoothness_assumption}, the open system model can be simulated with almost $O(Nt)$ gates. Beyond Trotterization, the primary obstacle in establishing theorem \ref{theoremdi} is the environment which is described by a continuum of bosonic modes labelled by the annihilation operators $a_{i, s}$ in Eq.~\eqref{eq:Ai_exp_rep} for $i \in \{1, 2 \dots N - 1\}$ and $s \in \mathbb{R}$. To simulate this model with a quantum circuit, we need to both  discretize the environment to a finite number of discrete bosonic modes \emph{and} truncate the local Hilbert space of the bosonic modes.--- Furthermore, we aim to execute both of these steps while retaining close to $\mathcal{O}(Nt)$ scaling in the number of gates and close to $\mathcal{O}(N)$ scaling in the number of ancillas. To this end, we introduce a scheme that uses a discretization of $a_{i, t}$ with respect to $t$ together with a complete orthonormal set of basis functions. Furthermore, we establish that for each of the discretized bosonic mode, the probability of high particle number states being occupied in each bosonic mode descreases superpolynomially with the number of particles and consequently its Hilbert space can be truncated to $\mathcal{O}(\text{polylog}(N, t, 1/\delta))$ levels. This allows us to control the overheads incurred by discretization and the Hilbert space truncation and establish theorem \ref{theoremdi}.

We remark that the problem of discretizing bosonic environment has been extensively studied in previous work \cite{trivedi2022descriptioncomplexitynonmarkovianopen,Chin2010exact, Woods2014mappings}. In particular, Ref.~\cite{Chin2010exact} developed a discretization procedure, called the star-to-chain transformation, based on the Lanczos iteration which is able to approximate the environment to an accuracy $\delta$ with $\mathcal{O}(Nt \text{polylog}(N, t, 1/\delta))$ discrete modes. However, for this specific discretization procedure, it is hard to rigorously guarantee that a truncation of the local Hilbert space dimension to near $\mathcal{O}(\text{polylog}(N, t, 1/\delta))$ levels that our scheme provably attains. Furthermore, if the star-to-chain transformation is directly implemented on a quantum computer, it would require at-least $\mathcal{O}(Nt)$ ancillas,--- however, our discretization scheme allows us to use the fact that the coupling functions $v_i(t)$ superpolynomially decay with $t$ to reduce the number of ancillas needed to close to $\mathcal{O}(N)$.

We also emphasize that the prefactor in the complexity scaling depends on the property of the coupling function $v_i$ and can grow with $p$. Nevertheless, we can always choose $p$ to grow slowly enough with $N, t, \delta$ such that the gate complexity is $\mathcal{O}(Nt(Nt/\delta)^{o(1)})$, where $o(1)$ denotes a variable which can be made smaller than any constant. 

In a special class of non-Markovian models, referred as ``non-dissipative" models in the open-system literature \cite{Ferialdi2014general}, we can relax the assumption of smoothness and superpolynomial decay [Eq.~\eqref{eq:smoothness_assumption}] and still attain a near-optimal gate and ancilla count. Non-dissipative non-Markovian models are models where the system-environment interaction Hamiltonian [Eq.~\eqref{eq:lattice_mode_VSE}] has two additional properties:  \emph{first}, the jump operators $J_{i, i + 1}$ are Hermitian i.e., $J_{i, i + 1} = J_{i, i + 1}^\dagger$ and \emph{second}, the kernel 
\begin{align*}
\bar{K}_{i}(t - t') = \text{Tr}((A_i(t) + \text{h.c.})(A_i^\dagger(t') + \text{h.c.})\rho_E(0)),
\end{align*}
is purely real. We emphasize that $\bar{K}_i(t,t')$ can be expressed as a linear combination of $K_i^{(\sigma,\sigma')}(t,t')$. In this case, the channel $\mathcal{E}_S(t)$ describing the system dynamics can be considered to be an ensemble average over a family of random time-dependent Hamiltonian. More specifically,
\[
\mathcal{E}_S(t)= \mathbb{E}_\xi[U_{\xi}(t, 0) (\cdot) U_{\xi}(0, t)],
\]
where $U_\xi(t, s) = \mathcal{T}\exp(-i\int_s^t H_\xi(\tau) d\tau)$ with $H_\xi(t)$ being 
\begin{align}\label{eq:hamiltonian_ensemble}
H_\xi(\tau) = H_S + \sum_{i=1}^{N - 1} \xi_i(t) J_{i,  i+1},
\end{align}
where $\xi_i(t)$ is a classical Gaussian random process with 
\begin{equation}\label{eq:Sec2NonMarkovianclassicalnoise}
\langle \xi_i(t)\rangle=0,\; \langle \xi_i(t) \xi_j(t')\rangle=\delta_{i, j} \bar{K}_{i}(t- t').
\end{equation}In our next result, we make no assumptions on the smoothness or decay of $\bar{K}_i(t)$ but only assume that $\exists m, M >0$ such that
\begin{align}\label{eq:non_dissipative_kernel_conditions}
\sup_{\tau \in \mathbb{R}} \vert \bar{K}_i(\tau) \vert < m, \int_{-\infty}^\infty \vert \bar{K}_i(\tau) \vert d\tau < M.
\end{align}

\begin{theorem}[Simulation of non-dissipative non-Markovian dynamics]
For the non-dissipative non-Markovian open system model satisfying Eq.~\eqref{eq:non_dissipative_kernel_conditions}, for any $p>0$, there is an ensemble of depth $\mathcal{O}(t(Nt/\delta)^{1/p})$ geometrically local unitary circuits $\{U_{z}\}_{z\sim Z}$, where $Z = \{Z_1, Z_2\dots Z_Q\}$ is a jointly Gaussian zero-mean random vector with dimensionality $Q = \mathcal{O}(Nt(Nt/\delta)^{1/p})$ and a classically efficiently computable covariance matrix, such that for all $\rho_S(0)$
\[
\norm{\mathbb{E}_z(U_z \rho_S(0) U_z^\dagger) - \mathcal{E}_S(t)(\rho_S(0))}_\textnormal{tr} \leq \delta.
\]
\label{theoremnondi}
\end{theorem}
\noindent A first approach to simulating $\mathcal{E}_S(t)$ in the non-dissipative case could be to simulate the time-dependent Hamiltonian $H_\xi$ in Eq.~\eqref{eq:hamiltonian_ensemble} using a standard Hamiltonian simulation algorithm \cite{Hatano2005Finding,Kieferov2019Simulating}. However, for time-dependent Hamiltonians, to be efficient such algorithms typically require a guaranteed upper bound on the time-dependent terms in the Hamiltonian as well as their derivatives. However, $\xi_i(t)$ in Eq.~\eqref{eq:hamiltonian_ensemble} does not have bounded derivatives at all times with high probability, which prohibits directly using a Hamiltonian simulation algorithm for this setting. Instead, in our approach we use the Higher-order Trotterization result from lemma \ref{lemma1} with the fact that the Hamiltonian is non-dissipative to map the Trotterized dynamics to a circuit ensemble. We also remark that the algorithm in Theorem \ref{theoremnondi} does not need any ancilla qubits and thus has fewer resource requirements than Theorem \ref{theoremdi}.

\subsubsection{Tighter results for local observables}

\begin{figure}
\includegraphics[width=0.9\columnwidth]{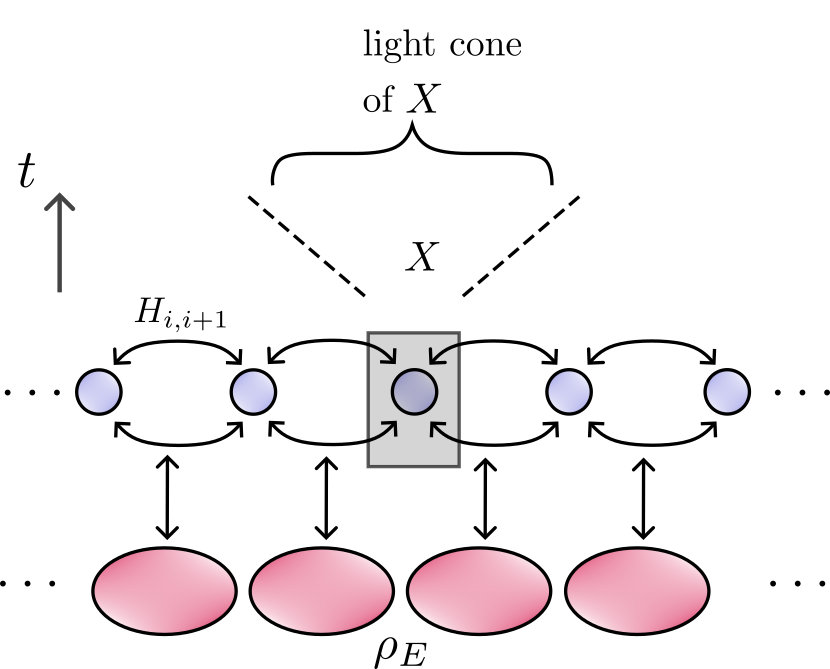}

\caption{Schematic illustration of the effective light cone in the non-Markovian dynamics model. An observable originally supported on $X$ can only propagate to a finite region whose size is linearly in time $t$ due to the locality of the Hamiltonian.}

\label{Fig:illustration_lightcone}
\end{figure}

In the algorithms described in the previous subsection, the parallelized circuit depth required to obtain a good approximation of the full many-body state of the spins increases, although very slowly, with the number of qubits $N$. This is of special concern when using near-term quantum devices without any error correction, where there is a possibility of errors accumulating during the run-time of the circuit --- if the circuit depth increases with the number of qubits, then the total accumulated error due to noise would also typically grow with the number of qubits and thus limit the size of the many-body problem that can be solved on such devices.

However, in many physically relevant scenarios, we are typically not interested in the full many-body state of the spins but only in local observables measured on the many-body state \cite{Trivide2024advantage}. For geometrically local many-body models, since local observables at any time $t$ are expected to depend only on the Hamiltonian terms within distance $t$ of the observables \cite{Anthony2023Speedlimit}, as also shown in Fig. \ref{Fig:illustration_lightcone}, the circuit depth needed to accurately simulate local observables can be made independent of the system size $N$. In our next result, we rigorously address this question --- more specifically, we will consider the problem of simulating observables whose support has a diameter bounded by a given constant $R$. We will aim to implement a channel $\mathcal{E}$ such that for any initial state $\rho_S(0)$ and any region $X$ such that $\text{diam}(X) \leq R$,
\begin{align}\label{eq:local_error_channel}
\norm{\text{Tr}_{\overline{X}}[\mathcal{E}_S(t)(\rho_S(0))] - \text{Tr}_{\overline{X}}[\mathcal{E}(\rho_S(0))]}_\text{tr}\leq \delta_\text{loc},
\end{align}
where $\overline{X}$ is the complement of $X$. We will assume that the diameter $R$ is an $N$ and $t$ independent $\mathcal{O}(1)$ constant. Our next result shows that the circuit depth in both the dissipative and non-dissipative geometrically local open models can be chosen to be independent of $N$ if we want to  accurately simulate only local observables. We first present the result for non-Markovian dynamics with coupling functions satisfying smoothness and polynomial decay assumption laid out in Eq.~\eqref{eq:smoothness_assumption}.

\begin{theorem}[Local observables in general non-Markovian dynamics.]\label{theorem_error_robustness}
For a $D$-dimensional lattice model, if the coupling functions $v_i$ satisfy Eq.~\eqref{eq:smoothness_assumption}, then for any $p>0$, the channel on the system qubits  $\mathcal{E}_S(t)$ [Eq.~\eqref{eq:channel}] can be approximated within a local error $\delta_\textnormal{loc}$ [Eq.~(\ref{eq:local_error_channel})] with a geometrically local circuit of depth $\mathcal{O}(t(t^{D+1}/\delta_\textnormal{loc})^{1/p})$ and with $\tilde{\mathcal{O}}(N(t^{D+1}/\delta_\textnormal{loc})^{1/(p+1)})$ ancillary qubits.
\end{theorem}
\noindent Similar to the previous subsection, for the case of non-dissipative models, we can relax the smoothness and superpolynomial decay assumption on the coupling function to just boundedness of memory kernels and their $L^1$ norms [Eq.~\eqref{eq:non_dissipative_kernel_conditions}].
\begin{theorem}[Local observables in non-dissipative non-Markovian dynamics]\label{theorem:error_robustness_non_dissipative}
For the non-dissipative non-Markovian open system model satisfying Eq.~\eqref{eq:non_dissipative_kernel_conditions}, for any $p>0$, there is an ensemble of depth $\mathcal{O}(t(t^{D +1}/\delta_\textnormal{loc})^{1/p})$ unitary circuits $\{U_{z}\}_{z\sim Z}$, where $Z = \{Z_1, Z_2\dots Z_Q\}$ is a jointly Gaussian zero-mean random vector with dimension $Q = \mathcal{O}(Nt(t^{D+1}/\delta_\mathrm{loc})^{1/p})$ and a classically efficiently computable covariance matrix, such that the channel $\mathbb{E}_z(U_z (\cdot) U_z^\dagger)$ approximates $\mathcal{E}_S(t)$ within a local error $\delta_\textnormal{loc}$ [Eq.~\eqref{eq:local_error_channel}].
\label{theoremnondi_local_obs}
\end{theorem}
\noindent We remark that in Theorem \ref{theorem:error_robustness_non_dissipative}, the classical sampling complexity of the circuit ensemble should scale as $1/\delta_\mathrm{loc}^2$ such that the sampling error is below $\delta_\mathrm{loc}$. If we also consider this classical sampling complexity, the total circuit depth is $\mathcal{O}(t/\delta_\mathrm{loc}^2(t^{D+1}/\delta_\mathrm{loc})^{1/p})$. To establish both of these theorems, we revisit the Trotterization error bound presented in lemma \ref{lemma1},--- using the recently established Lieb Robinson bounds for open lattice models with Gaussian environments \cite{trivedi2024liebrobinsonboundopenquantum}, we show that when considering only errors in local observables [Eq.~\eqref{eq:local_error_channel}], the Trotterization error bound can be significantly strengthened and made independent of the system size $N$. These improved Trotter error bounds then allow us to re-analyze the quantum algorithms developed in the analysis of theorems \ref{theoremdi} and \ref{theoremnondi} for local observables and provide upper bounds on the required circuit depth independent of $N$.

\subsubsection{Simulation of Markovian dynamics\label{subsec:Simulation-of-Markovian}}

We next turn to the case of Markovian dynamics i.e., when $K_i^{(\sigma, \sigma')}(s - s') \sim \delta(s - s')$.
In this case, the reduced density matrix of the system $\rho_S(t)$ evolves as per $\dot{\rho}_S(t) = \mathcal{L}\rho_S(t)$ where $\mathcal{L}$ is a geometrically local Lindbladian given by
\begin{align}\label{eq:Lindbladian_dynamics}
\mathcal{L} =-i[H_S, \cdot] + \sum_{i = 1}^{N - 1}\mathcal{D}_{J_{i, i + 1}},
\end{align}
with $\mathcal{D}_J = J(\cdot) J^\dagger - \{\cdot, J^\dagger J \}/2$ being the dissipator corresponding to the jump operator $J$. The results stated so far in the previous subsections do not directly apply to this case since the delta-function kernel diverges at $s = s'$. While, to the best of our knowledge, simulating a general geometrically local Lindbladian with the number of geometrically local gates scaling almost as $\mathcal{O}(Nt)$ is still an open question, in this subsection we provide several results for the setting where the Jump operators $J_{i, i +1}$ commute with each other and their Hermitian conjugates (i.e., for $i \neq j$, $[J_{i, i + 1}, J_{j, j + 1}] = [J_{i, i + 1}, J^\dagger_{j, j + 1}] = 0$) but do not necessarily commute with the Hamiltonian. We remark that the requirement that all jump operators commute is automatically satisfied for models with a geometrically local Hamiltonian and single site jump operators, and may serve as a useful intermediate step in simulating more general non-commuting Lindbladian evolutions \citep{Vikram2025Accuracy,Zanardi2016dissipative}.

We first consider the non-dissipative case where each jump operator is Hermitian (i.e.~$J_{i, i + 1} = J_{i, i + 1}^\dagger$). 
We will assume the availability of Quantum Random Access Memory (Q-RAM) $O(\bm{x})$, which given a $D_{\bm{x}}$-dimensional classical vector $\bm{x}$ builds a unitary $O(\bm{x})$ satisfying
\begin{equation}
O(\bm{x})\ket{i,z}=\ket{i,z\oplus x_i},\text{  for  }i\in\{1, 2\cdots, D_{\bm{x}}\}.
\end{equation}
We can then show that, in this case, the Lindbladian can be simulated near-optimally in $N$ and $t$.


\begin{theorem}[Simulation of commuting, non-dissipative Markovian dynamics]
\label{theorem_markovian_non}
If the jump operators $J_{i,i+1}$ are Hermitian and commuting then there is an ensemble of unitary circuits $\{U_z\}_{z \sim Z}$, where $z \sim Z$ is the trajectory of a Wiener process which can be classically efficiently sampled from, such that $\mathbb{E}_z(U_z(\cdot) U_z^\dagger)$ approximates the channel $\mathcal{E}_S(t)$ within error $\delta$ and each circuit in the ensemble requires $\mathcal{O}(Nt\mathrm{polylog}(Nt/\delta)$ queries to the Q-RAM $O$, with $D_{\bm{x}}= \mathcal{O}(N^4t^4/\delta^4)$, and $\mathcal{O}(Nt\mathrm{polylog}(Nt/\delta))$ additional quantum gates.
\end{theorem}
\noindent We remark that Q-RAM is required in most quantum algorithms that require loading input data onto quantum computers such as time-dependent Hamiltonian simulation algorithms \cite{Kieferov2019Simulating,low2019hamiltoniansimulationinteractionpicture,Berry2020Timedependent}, as well as other quantum algorithms for linear-algebra such as Harrow-Hassidim-Lloyd (HHL) algorithm \cite{Harrow2009Quantum}. However, this is indeed a stringent hardware reqiurement and thus we expect the algorithm proposed in theorem 5 to be suitable only for fault-tolerant quantum computers.

To establish this result, we begin similarly to the analysis of Theorem \ref{theoremnondi} by noting that in the non-dissipative setting, $\mathcal{E}_S(t) = \mathbb{E}_\xi[U_\xi(t, 0)(\cdot) U_\xi(0, t)]$ where $U_\xi(t, s) = \mathcal{T}\exp(-i\int_s^t H_\xi(\tau) d\tau)$ with \citep{Chenu2017classical}
\begin{equation}
H_\xi(t) = H_S + V_\xi(t).
\label{eq:stochastic_schrodinger_eq}
\end{equation}
Here $V_\xi(t)=\sum_{i=1}^{N-1}J_{i,i+1} \xi_i(t)$, with $\xi_i (t)$ is a zero-mean Gaussian white noise process satisfying $\langle \xi_i(t)\xi_{i'}(t')\rangle=\delta_{i,i'}\delta(t-t')$. However, Eq. (\ref{eq:stochastic_schrodinger_eq}) is still difficult to simulate, because any sample of the white-noise process $\xi_i(t)$ would be discontinuous with high probability. Furthermore, unlike for the non-Markovian models analyzed in the previous subsection, it remains unclear if the Hamiltonian in Eq.~\eqref{eq:stochastic_schrodinger_eq} can be Trotterized to an arbitrary order. However, since we assume that the jump operators $J_{i, i + 1}$ commute with each other, $V_\xi(t)$ can be seen to be a commuting Hamiltonian. Consequently, we first rotate to the interaction picture with respect to $V_\xi(t)$ and still obtain an effective geometrically local Hamiltonian. More specifically,
\begin{subequations}
\begin{align}
U_\xi(t, s) = W_\xi(t, 0) \bar{{U}}_\xi(t, s) W_\xi(0, s),
\end{align}
where
\begin{align}\label{eq:change_of_frame}
W_\xi(t, s)= \prod_{i= 1}^{N - 1} \mathcal{T}\exp\bigg(-i\int_s^t \xi_i(\tau) J_{i, i + 1}d\tau\bigg),
\end{align}
and $\bar{U}_{\xi}(t, s) = \mathcal{T}\exp(-i\int_s^t \bar{H}_\xi(\tau) d\tau)$ where
\begin{align}
\bar{H}_\xi(t) = \sum_{i= 1}^{N - 1} \bar{H}_{i, i + 1; \xi}(t),
\end{align}
\end{subequations}
where $\bar{H}_{i, i + 1; \xi}(t) = W_\xi(0, t) H_{i, i + 1} W_\xi(t, 0)$.
Since $W_\xi(t, s)$ is itself a product of commuting local unitaries [Eq.~\eqref{eq:change_of_frame}], $\bar{H}_{i,i + 1; \xi}(t)$ is geometrically local and bounded. Furthermore, since $W_\xi(t, s)$ involve \emph{integrals} of the white-noise process $\xi_i(t)$, $\bar{H}_{i,i + 1; \xi}(t)$ is also almost-surely continuous with respect to $t$ \cite{Peter2010Brownian}. However, its derivative can still be unbounded with high probability, which again makes it difficult to simulate $\bar{H}_\xi(t)$ with the existing quantum simulation toolbox. Nevertheless, we build upon the algorithm for time-dependent Hamiltonian simulation \citep{Kieferov2019Simulating} together with the Lieb-Robinson bounds \cite{haah2021quantum,Efficient2006Osborne} to obtain a near-optimal simulation algorithm for $\tilde{H}_\xi(t)$. 

In section \ref{subsubsec_commuting_dissipative_th5}, we also analyze a possible extension of this approach to the case where the jump operators $J_{i, i + 1}$ are not Hermitian which requires the appropriate quantum generalization of the white-noise process as a physical resource --- this, however, does not yield a digital quantum simulation algorithm with near-optimal scaling. Nevertheless, even when $J_{i, i + 1}$ are non-Hermitian but commuting (i.e., for $i \neq j$, $[J_{i, i + 1}, J_{j, j + 1}] = [J_{i, i + 1}, J^\dagger_{j, j + 1}] = 0$), we develop an approach that allows us to go beyond second-order Trotter error scalings. Unlike Theorem \ref{theorem_markovian_non}, this algorithm does not require a Q-RAM, and thus might be more suitable for pre-fault tolerant devices. We remind the reader that geometrically local Lindbladians can be Trotterized to the second order, similar to geometrically local Hamiltonians, since the second order Trotter formula only requires forward time evolution. However, any third (or higher) order Trotter formula necessarily requires backward time evolution \cite{Suzuki:1991jtk}, and thus cannot be directly applied to Lindbladians. Instead, we approach this problem by first trying to locally dilate the Lindbladian to third-order, and then Trotterize the resulting dilation. More specifically, as depcited in Fig.~\ref{fig:dilation}, given a time-step $\Delta t$, we construct a Hamiltonian $H_\text{dia}(\sqrt{\Delta t})$ on the $N$ system qubits and $N - 1$ ancillary qudits with $d = 5$ such that
\begin{align}\label{eq:dilated_Hamiltonian_markovian_third_order}
H_\text{dia}(\eta) = H_0  + \sum_{i = 1}^{N - 1} (S_{i, i+1}^\dagger O_i(\eta) + \text{h.c.}),
\end{align}
where $H_0$ is a Hermitian operator close to $H_S$, $S_{i,i+1}$ is a system local operator supported around the sites $i,i+1$, and $O_i(\eta)$ is an operator acting on the $i^\text{th}$ ancillary qudit such that for any fixed $N$,
\begin{align}\label{eq:definition_error_remainder_dilated}
\mathcal{R}_\text{dia}:= e^{\mathcal{L}\Delta t} - \mathcal{E}_\text{dia}(\Delta t) = \mathcal{O}(\Delta t ^4), 
\end{align}
where 
\[
\mathcal{E}_\text{dia}(\Delta t) = \text{Tr}_E(e^{-iH_\text{dia}(\sqrt{\Delta t})\Delta t} ((\cdot) \otimes \ket{\text{Vac}}\!\bra{\text{Vac}})e^{iH_\text{dia}(\sqrt{\Delta t})\Delta t}),
\]
with $\ket{\text{Vac}}$ denoting the $\ket{0}$ state of all the ancillary qudits, and $\text{Tr}_E$ denoting a trace over the ancillary qudits. Furthermore, we can consider the leading order contribution in $\mathcal{R}_\text{dia}$ with respect to $\Delta t$ i.e. the superoperator $\mathcal{G}_\text{dia}^{(4)}$ defined via
\[
\mathcal{R}_\text{dia} = \mathcal{G}_\text{dia}^{(4)}\Delta t^4 + \mathcal{O}(\Delta t^5).
\]
Our final result shows that the operators $O_i(\eta)$ in Eq.~\eqref{eq:dilated_Hamiltonian_markovian_third_order} can be chosen such that $\mathcal{G}_\text{dia}^{(4)}$ grows only linearly with $N$.

\begin{figure}
\includegraphics[width=1\columnwidth]{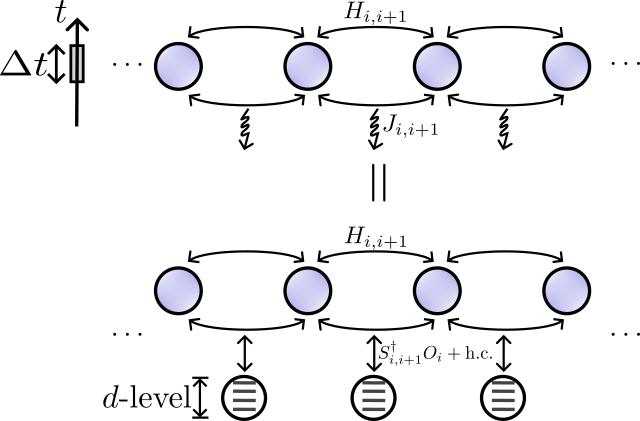}
\caption{The setting of the dilated Hamiltonian for simulating Lindbladian dynamics. The white circles represent the ancillas. $O$ is the local interaction between an ancilla and a pair of nearest neighbor system sites, which is a function of $\sqrt{\Delta t}$.}
\label{fig:dilation}
\end{figure}


\begin{theorem}[Third order dilation of Lindbladians]
If $[J_{i,i+1},J_{j, j + 1}] = 0$ and $[J_{i, i + 1}, J_{j, j + 1}^\dagger] = 0$ for $i \neq j$, then there is a choice of $O_i(\eta)$ which satisfies $\norm{O_i(\eta)} \leq \mathcal{O}(\eta^{-1})$ as $\eta \to 0$ such that $\norm{\mathcal{G}^{(4)}_\textnormal{dia}}_{\diamond} = \mathcal{O}(N)$.
\label{theorem3} 
\end{theorem}
\noindent This theorem indicates that, to the lowest order in $\Delta t$, our dilation procedure attains the same error scaling with respect to both $N$ and $\Delta t$ as a $3^\text{rd}$ Trotter formula. Since the resulting dilated Hamiltonian is geometrically local, we can Trotterize it as a lattice Hamiltonian using the results in Ref.~\cite{childs2019nearly}. However, since $\Vert O_i(\sqrt{\Delta t})\Vert = \mathcal{O}(\Delta t^{-1/2})$, we need to choose a $7^\text{th}$ order Trotter product formula to ensure that the Trotter error also scales as $\mathcal{O}(N \Delta t^4)$.  The above theorem then suggests that, when simulating the Lindbladian for a total time $t$, by dividing this evolution time into $T$ time-steps corresponding to $\Delta t = t / T$, the total simulation error would scale as $\mathcal{O}(T\times N(t/T)^4) = \mathcal{O}(Nt^4/T^3)$ and thus can be guaranteed to be $\leq \delta$ by choosing $T =\Theta(t(Nt/\delta)^{1/3})$. This simulation algorithm would thus yield a circuit with  $\mathcal{O}(t(Nt/\delta)^{1/3})$ depth and $\mathcal{O}(Nt(Nt/\delta)^{1/3})$ geometrically local gates that simulates the Markovian master equation.

While our analysis suggests that we can go beyond the error scaling obtained from second order Trotterization, we remark that we have only analyzed the error to the leading order in the discretization time step $\Delta t = t / T$ --- we leave a more rigorous analysis of the remainder that includes all orders as an open problem for future work. We also leave open the possibility to extend this dilation and Trotterization scheme to higher orders --- the main challenge here would be to systematically construct a geometrically local dilation scheme that achieves an error $\mathcal{O}(N\Delta t^{p + 1})$ per time-step $\Delta t$ for any desired order $p$. While recent work has provided a systematic way of constructing a higher-order dilation scheme for Lindbladians \cite{Linlin2024Simulating}, the resulting remainder scales as $\mathcal{O}(N^{p +1} \Delta t^{p + 1})$ instead of the desired scaling of $\mathcal{O}(N\Delta t^{p + 1})$. Furthermore, the resulting dilated Hamiltonian is not necessarily geometrically local, simulating which would result in a polynomial overhead with respect to the system size $N$ when accounting for geometrically local gates. 

\section{Proof outline of the results\label{sec:Proof-outline-of}}

In this section, we outline the proof strategy for our main results and describe the simulation algorithms. The details can be found in the supplemental material \cite{SM}.

\subsection{Trotterization error analysis (Lemma \ref{lemma1})\label{subsec:Proof-of-Lemma}}

Our starting
point is a time-independent description of the non-Markovian lattice model in Eqs. (\ref{eq:lattice_model}) and \eqref{eq:Ai_exp_rep}. Consider the  Hamiltonian $\mathcal{H}_{SE}$ given by
\begin{subequations}
\begin{equation}
\mathcal{H}_{SE}=H_S+V_{SE}+H_E,\label{eq:Schrodinger_picture_lattice_model}
\end{equation}
where $H_S$ is defined in Eq. (\ref{eq:sec2HSlattice}), $H_E$ is given by
\begin{align}\label{eq:environment_Hamiltonian}
H_E = \sum_{i = 1}^{N - 1} \int_{-\infty}^\infty \omega a_{i, \omega}^\dagger a_{i,\omega}d\omega,
\end{align}
$V_{SE}$ is given by,
\begin{align}
V_{SE} = \sum_{i = 1}^{N - 1} J_{i, i + 1}^\dagger A_i + \text{h.c.},
\end{align}
\end{subequations}
with $A_i$ being the following time independent operator acting on the environment:
\begin{align}
A_i=2\pi\int_{-\infty}^\infty \hat{v}_i(\omega) a_{i,\omega} d\omega.
\end{align}
Since the operators $a_{i, \omega}$ in Eq.~\eqref{eq:Ai_exp_rep} satisfy $[a_{i, \omega}, a^\dagger_{i', \omega'}] = \delta_{i, i'} \delta(\omega - \omega')$, it can be seen that the Hamiltonian in Eq. (\ref{eq:lattice_model}) is obtained by transforming $\mathcal{H}_{SE}$ into the interaction picture with respect to $H_E$. Therefore, we can express the unitary generated by $H_{SE}(t)$, $U_{SE}(t, 0) = \mathcal{T}\exp(-i\int_0^t H_{SE}(s) ds)$, as
\begin{align}\label{eq:interaction_picture}
U_{SE}(t, 0) = e^{iH_E t} e^{-i\mathcal{H}_{SE} t}.
\end{align}
We next decompose $\mathcal{H}_{SE}$ as $\mathcal{H}_{SE} =\mathcal{H}_\mathrm{o}+\mathcal{H}_\mathrm{e}$ where
\begin{equation}
\mathcal{H}_\mathrm{o}=\sum_{i\in\mathrm{odd}}\mathcal{H}_{i,i+1},
\mathcal{H}_\mathrm{e}=\sum_{i\in\mathrm{even}}\mathcal{H}_{i,i+1},
\end{equation}
with the definition
\begin{equation}
    \mathcal{H}_{i,i+1}=H_{i,i+1}+(J_{i,i+1}^\dagger A_i+\text{h.c.})+\int_{-\infty}^\infty \omega a_{i,\omega}^\dagger a_{i,\omega}d\omega.
\end{equation}
Applying the $P^\text{th}$ order time-independent Trotter formula in Eq.~\eqref{eq:trotter_formula} to the unitary evolution under $\mathcal{H}_{SE}$ in the time interval $[j \Delta t, (j + 1) \Delta t]$, where $\Delta t=t/T$, we obtain
\begin{equation}
\begin{aligned} & e^{-i\mathcal{H}_{SE}\Delta t}=\prod_{i=1}^{s_P}e^{-i\mathcal{H}_{\mathrm{e}}(f_{i}-f_{i-1})\Delta t}e^{-i\mathcal{H}_{\mathrm{o}}(e_{i}-e_{i-1})\Delta t}+\mathcal{R}_\mathrm{tro},
\end{aligned}
\label{eq:time_independent_trotter_formula}
\end{equation}
where
$\mathcal{R}_\mathrm{tro}$ denotes the error remainder and $e_{i},f_{i}$ are defined below Eq. (\ref{eq:Sec2trotterexpression}). Now, the Trotter decomposition $U_\text{tro}(t; T)$ provided in Eq.~\eqref{eq:trotterization_expression} is obtained by using Eq.~\eqref{eq:time_independent_trotter_formula} and transforming back into the interaction picture with respect to $H_E$ via Eq.~\eqref{eq:interaction_picture} i.e.,
\[
U_\text{tro}(t; T) = e^{iH_E T \Delta t}\bigg(\prod_{i=1}^{s_P}e^{-i\mathcal{H}_{\mathrm{e}}(f_{i}-f_{i-1})\Delta t}e^{-i\mathcal{H}_{\mathrm{o}}(e_{i}-e_{i-1})\Delta t} \bigg)^T.
\]


Analysis of the error between the Trotter formula $U_\text{tro}(t; T)$ and the exact evolution $U_{SE}(t, 0)$ requires a more careful analysis. While an explicit expression for the remainder $\mathcal{R}_\mathrm{tro}$ of the time dependent Trotter formula (Eq.~\eqref{eq:time_independent_trotter_formula}) has been derived in Ref. \citep{childs2019nearly}, applying it directly to $\mathcal{H}_{SE}$ would result in the remainder containing nested commutators with the environment Hamiltonian $H_E$. We remark that the Hamiltonian $H_E$ is a highly divergent unbounded Hamiltonian even when compared to $V_{SE}$. For instance, when restricted to the $m$-particle subspace of the bosonic environment, then $V_{SE}$ would be a bounded operator. However, even when restricting to a fixed number of particles in the bosonic environment, $H_E$ would still be unbounded --- physically, this is simply a consequence of $H_E$ [Eq.~\eqref{eq:environment_Hamiltonian}) containing arbitrarily high oscillation frequencies. Consequently, to obtain a finite bound on $U_\text{tro}(t; T) - U_{SE}(t, 0)$, we would preferably require a remainder expression which does not contain $H_E$.

In the supplemental material, we show that by carefully analyzing the transformation of the local error-sum representation for $\mathcal{R}_\text{tro}$ defined in Ref.~\cite{childs2019nearly} from the Schroedinger picture to the interaction picture, we can in fact obtain a remainder expression that is expressible entirely in terms of nested commutators of $H_{i,i+1},J_{i,i+1}$ and the operators $A_{i}(t)$
(as well as their Hermitian conjugate). 
Even with this expression, the remainder still contains the unbounded operators $A_i(t)$ and the strategy used in Refs.~\cite{childs2019nearly,Childs2021trotter} to  bound $\lVert \mathcal{R}_\mathrm{tro}\rVert$ would yield an infinite upper bound.
To obtain a finite error bound, we instead use the fact that the initial environment state is a Gaussian state and the environment is finally traced out. Consequently, we can use Wick's theorem to contract and remove the unbounded operators $A_i(t)$ from the remainder expression and effectively derive a local-error sum representation for the Trotter error that is entirely in terms of the bounded operators $H_{i, i + 1}$ and $J_{i, i +1}$. By exploiting the locality of the system Hamiltonian and jump operators, we can then obtain an upper on the Trotter error that scales as $\mathcal{O}(P! s_P^P NT(\Delta t)^{P  +1})$, where $s_P$ is the number of stages in the Trotter formula [Eq.~\eqref{eq:trotter_formula}]. We remark that the $P!$ prefactor in this upper-bound arises due to the Wick contractions, and can be physically interpreted to be a consequence of the bosonic nature of the environment. Finally, we remark that while we restrict our analysis to a 1D problem, this analysis can easily be extended to higher dimensional problems in which the Trotter formula used decomposes the Hamiltonian as a sum of more than 2 operators which are individually commuting Hamiltonians.

\subsection{Quantum algorithm for general non-Markovian lattice model (Theorem \ref{theoremdi})\label{subsec:Proof-of-Theorem-dissipative}}

Our algorithm comprises of three steps, as shown in Fig.~\ref{fig:figure2}. \emph{First}, we discretize the continuum of
bosonic modes describing the environment into a finite set of modes. \emph{Second}, we apply Lemma \ref{lemma1} to Trotterize the non-Markovian dynamics. Finally, we truncate each bosonic
mode to a $d$-level system (which we call quditization) so that the resulting Hamiltonian
can be implemented on a digital quantum computer. By ensuring that the incurred error in each of these steps is at most $\delta/3$, 
the overall simulation error remains below $\delta$. Throughout this subsection, we assume that the coupling function $v_i(t)$ satisfies Eq. (\ref{eq:smoothness_assumption}). We further assume that $v_i(t)$
is additionally compact and
\[
\text{supp}(v_i) \subseteq [-r, r] \ \forall i,
\]
where $\text{supp}(v_i)$ is the set of $t$ such that $v_i(t) \neq 0$.  We will also show that, since $v_i(t)$ decays superpolynomially with $t$ as $\abs{t} \to \infty$, we can always approximate it with such a function.

\subsubsection{Discretizing bosonic environment}
\begin{figure}
\includegraphics[width=1\columnwidth]{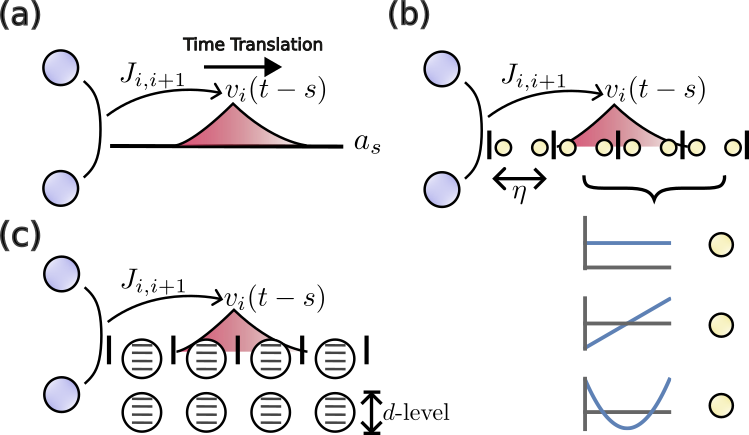}

\caption{Flowchart of the algorithm in proving Theorem \ref{theoremdi},
which comprises of three steps: discretization, Trotterization and quditization (the Trotterization step is omitted here). (a) Non-markovian
dynamics for a single jump operator in the interaction picture: the
system moves along a one-dimensional line hosting the continuous
bosonic modes $a_{\tau}$. (b) Discretization of the bosonic
environments: The entire line is divided into small segments of
width $\eta$. In each segment, we introduce a new set of bosonic modes
$b_{j}^{n}$ to replace $a_{\tau}$ as in Eq. (\ref{eq:definition_of_bosonicModeB}).
Yellow circles denote the newly introduced modes $b_{j}^{n}$, and the inset shows their structure for the first
few indices $j$. (c) Quditization of bosonic modes: Each
mode is truncated to dimension $d$ and encoded as a qudit,
here depicted by the white circle.}
\label{fig:figure2}
\end{figure}

A common strategy to discretize the bosonic environment to simulate the open system dynamics is to employ star-to-chain transformation \citep{Chin2010exact,trivedi2022descriptioncomplexitynonmarkovianopen}. In the star-to-chain transformation, the bosonic environment is effectively discretized in the frequency-domain. However, with this approach, it makes it difficult to rigorously exploit the fact that $v_i(t)$ has compact support in the time-domain, as well as optimally truncate the Hilbert space of the resulting bosonic modes. Instead, we provide a strategy to controllably discretize the bosonic environment directly in the time-domain.

Consider one term in the system-environment interaction Hamiltonian 
$J_{i,i+1}^{\dagger}\int_{-\infty}^\infty  v_i(t-s)a_{i,s}ds+\text{h.c.}$ --- a particularly convenient visual interpretation of this system-environment interaction, shown in Fig.~\ref{fig:figure2}a, is to consider the continuum of modes $a_{i,s}$
as a field theory on a one-dimensional line with position indexed by $s$. The system can be considered to move along the line at a constant velocity of 1 and at time $t$ couple to the points on the line within a temporal window $v_i(t - s)$ and via the system operator $J_{i,i+1}^{\dagger}$. Since we only consider a total evolution time
of $t$, only the bosonic modes $a_{i,s}$ with $s\in[- r , r +t]$ will participate in the dynamics.

To discretize the environment, we therefore divide the one-dimensional line into contiguous small segments of width $\eta$, as illustrated in Fig. \ref{fig:figure2}b. These
segments are labeled by integers $n\in[-\lceil r /\eta\rceil,\lceil( r +t)/\eta\rceil]$ --- the index $n$ corresponds to segment in $[n\eta, (n + 1)\eta)$.
We define a set of orthonormal polynomials $\{P_{j}^n(s)\}$ which are non-zero only in the segment corresponding to $n$:
\begin{equation}
P_{j}^{n}(s)=\begin{cases}
0, & s<n\eta;\\
\sqrt{\frac{2j+1}{\eta}}L_{j}(2\frac{s-n\eta}{\eta}-1), & n\eta\le s<(n+1)\eta;\\
0, & s\geq(n+1)\eta,
\end{cases}\label{eq:explicit_expression_of_Pjn}
\end{equation}
where $L_j$ is the standard $j^\text{th}$ order Legendre polynomial. The polynomials $P_j^n$ satisfy the orthonormality relations:
\begin{equation}
\int_{-\infty}^{\infty}P_{j}^{n}(s)P_{j'}^{n'}(s)ds=\delta_{n, n'}\delta_{j,j'}.\label{eq:orthonormality_of_Pjn_in_main_text}
\end{equation} 
Now, for every $t$, the coupling function $v_i(t-s)$ as a function of $s$ can then be expanded in the basis of $\{P_{j}^n(s)\}$ as
\begin{equation}\label{eq:expansion_v_i_basis}
v_i(t-s)=\sum_{j=0}^{\infty}C_{i,j}^{n}(t)P_{j}^{n}(s)\ \mathrm{for}\ s\in[n\eta,(n+1)\eta),
\end{equation}
where $C_{i,j}^{n}(t)$ are obtained by using the orthonormality relations in Eq.~\eqref{eq:orthonormality_of_Pjn_in_main_text}:
\begin{equation}
C_{i,j}^{n}(t)=\int_{n\eta}^{(n+1)\eta}v_i(t-s)P_{j}^{n}(s)ds.\label{eq:definition_of_cjn_main}
\end{equation}
We note that the norm of $C_{i,j}^{n}(t)$ can be bounded from Cauchy-Schwardz inequality. More explicitly,
\begin{equation}
\begin{aligned}
&|C_{i,j}^{n}(t)|\nonumber\\
&\leq\bigg[\int_{n\eta}^{(n+1)\eta}|v_i(t-s)|^{2}ds\int_{n\eta}^{(n+1)\eta}\big(P_{j}^{n}(s)\big)^{2}ds\bigg]^{\frac{1}{2}},\\&\leq C_0\sqrt{\eta}, \label{eq:upper_bound_on_cijn_from_Cauthy_Schwardz}
\end{aligned}
\end{equation}
where we have used the fact that $P^n_j$ has a $L^2$ norm $= 1$ and $C_0$ is defined in Eq. (\ref{eq:smoothness_assumption}).

Using Eq.~\eqref{eq:expansion_v_i_basis}, we can rewrite the bosonic operator $A_i(t)$ in the system-environment interaction Hamiltonian $V_{SE}(t)$ as
\begin{equation}
\begin{aligned} A_i(t)=&\int_{-\infty}^\infty  v_i(t-s)a_{i,s}ds,\\
= & \sum_{n}\sum_{j=0}^{\infty}C_{i,j}^{n}(t)\int_{n\eta}^{(n+1)\eta} P_{j}^{n}(s)a_{i,s}ds.
\end{aligned}
\end{equation}
The right hand suggests us to define a series of new bosonic modes
\begin{equation}
b_{i,j}^{n}=\int_{n\eta}^{(n+1)\eta}P_{j}^{n}(s)a_{i,s}ds.\label{eq:definition_of_bosonicModeB}
\end{equation}
The orthonormality of $\{P_{j}^{n}\}$ [Eq. (\ref{eq:orthonormality_of_Pjn_in_main_text})] imply that the annihilation operators $b_{i,j}^{n}$ satisfy the canonical commutation relation
\begin{equation}
[b_{i,j}^{n},b_{i',j'}^{n'\dagger }]=\delta_{n,n'}\delta_{i,i'}\delta_{j,j'}.
\end{equation}
With the newly introduced modes $b_{i,j}^{n}$, we obtain
\begin{equation}
A_i(t)=\sum_{n}\sum_{j=0}^{\infty}C_{i,j}^{n}(t)b_{i,j}^{n}.\label{eq:Hamiltonian_in_new_modes}
\end{equation}

We remark that uptil this point, we have simply reformulated the representation of the environment and not introduced any approximation. Consequently, Eq. (\ref{eq:Hamiltonian_in_new_modes}) still involves
infinitely many bosonic modes. Next, we cutoff each local polynomial basis
$\{P_{j}^{n}\}$ at a maximum degree $j_{\mathrm{max}}$. The resulting discretized bosonic modes $\tilde{A}_i(t)$ then takes the form
\begin{equation}
\tilde{A}_i(t)=\sum_{n}\sum_{j=0}^{j_{\mathrm{max}}}C_{i,j}^{n}(t)b_{i,j}^{n},\label{eq:discretized_Hamiltonian}
\end{equation}
where $C_{i,j}^{n}(t)$ is defined in Eq. (\ref{eq:definition_of_cjn_main}).
The parameter $j_\mathrm{max}$ will be chosen later to control the error of this approximation. This discretized environmental bosonic operator $\tilde{A}_i(t)$ corresponds to a new memory kernel
\begin{equation}
\begin{aligned}
\tilde{K_i}^{(-,+)}(s,s') & =\mathrm{Tr}(\tilde{A}_i(s)\tilde{A}_i^\dagger(s')\rho_E(0)),\\
 & =\sum_{n}\sum_{j=0}^{j_{\mathrm{max}}} C_{i,j}^{n}(s)C_{i,j}^{n*}(s').\label{eq:definition_of_new_memory_kernel_maintext}
\end{aligned}
\end{equation}
In the next lemma, we quantify the error between the approximate kernel $\tilde{K}_i^{(-, +)}(s, s')$ and the original kernel $K_i^{(-, +)}(s, s') = K_i^{(-, +)}(s - s')$.

\begin{Lemma}[Error bound on the Kernel difference]
For a fixed $j_{\mathrm{max}}$, the error between $K_i^{(-,+)}(s,s')$ and $\tilde{K}_i^{(-,+)}(s,s')$ can be bounded by 
\begin{equation}
\begin{aligned}
&|\tilde{K}_i^{(-,+)}(s,s')-K_i^{(-,+)}(s,s')|\leq\frac{2\eta^{j_\mathrm{max}+1}D_{j_\mathrm{max}+1}C_0 r }{(j_\mathrm{max}+1)!},
\end{aligned}
\end{equation}
where $C_\nu$, $D_\mu$ are defined in Eq.~(\ref{eq:smoothness_assumption}). Furthermore, 
\begin{equation}
\begin{aligned}
&\int_{0}^{t}ds\int_{0}^{t}ds'|\tilde{K}_i^{(-,+)}(s,s')-K_i^{(-,+)}(s,s')|\\&\quad\quad\leq4t\eta^{j_{\mathrm{max}}+1}D_{j_{\mathrm{max}}+1}\frac{2 r +\eta}{(j_{\mathrm{max}}+1)!}C_{0} r ,\label{eq:bounding_integral_kernel}
\end{aligned}
\end{equation}
and
\begin{equation}
\begin{aligned}
\int_{-\infty}^{\infty}|\tilde{K}_{i}^{(-,+)}(s,s')|ds'\leq2C_{0}^{2}(j_{\mathrm{max}}+1)(2 r +1)^{2}.
\end{aligned}\label{eq:Upperbound_on_L1_norm_of_tilde_K}
\end{equation}
\label{lemma:EB_Kernel}
\end{Lemma}

\noindent Since the reduced dynamics of the system
is completely determined by the memory kernels, we expect that if $\tilde{K}_i^{(-, +)}$ and $K_i^{(-, +)}$ are approximately equal, then so are the reduced states obtained on evolving the system with the two different memory kernels. More specifically, we consider the channel $\tilde{\mathcal{E}}_S(t)$ given by 
\begin{equation}
\tilde{\mathcal{E}}_S(t)=\mathrm{Tr}_E(\tilde{U}_{SE}(t,0)((\cdot)\otimes\rho_E(0))\tilde{U}_{SE}(0,t)),
\end{equation}
where $\tilde{U}_{SE}(t,0)=\mathcal{T}\exp(-i\int_0^t\tilde{H}_{SE}(\tau)d\tau)$ and
\[
\tilde{H}_{SE}(t) = \sum_{i=1}^{N-1}\left(H_{i,i+1}+(J_{i,i+1}^\dagger \tilde{A}_i(t)+\text{h.c.})\right).
\] 
Note that $\tilde{H}_{SE}(t)$ is the system-environment Hamiltonian corresponding to an environment with memory kernel $\tilde{K}_i^{(-, +)}$. Ref.~\citep{huang2024unifiedanalysisnonmarkovianopen} derived the following lemma providing a general upper bound on the error between the channels generated by different memory kernels in terms of the deviation between the two kernels.
\begin{Lemma}[From Ref.~\citep{huang2024unifiedanalysisnonmarkovianopen}, Eq. (2.43) in Sec. 2.3]
The difference between the reduced dynamics associated with different memory kernels can be upper bounded as 
\begin{equation}
\lVert \mathcal{E}_S(t)(\rho_S(0))-\tilde{\mathcal{E}}_S(t)(\rho_S(0))\rVert_\mathrm{tr}\leq e^{2\kappa}-1.
\end{equation} 
where
\[
\kappa=\sum_{i=1}^{N-1}\int_0^tds\int_0^tds'|\tilde{K}^{(-,+)}_i(s,s')-K^{(-,+)}_i(s,s')|.
\]
\label{Lemma: lemma_citation_from_functional_integral}
\end{Lemma}
\noindent Noting that $e^{2x}-1\leq 2e^2x \leq \mathcal{O}(x)$ if $x\leq 1$, we can combine lemmas \ref{lemma:EB_Kernel} and \ref{Lemma: lemma_citation_from_functional_integral} to obtain an upper bound on the error incurred in our scheme to discretize the environment.
\begin{Lemma}[Error bound on the evolution after discretization]
For any initial system state $\rho_S(0)$, the discretization error satisfies
\begin{equation}
\begin{aligned}  \lVert \mathcal{E}_S(t)(\rho_S(0)) -\tilde{\mathcal{E}}_S(t)(\rho_S(0))\rVert_{\mathrm{tr}}
\leq \mathcal{O}\left(Nt\eta^{j_{\mathrm{max}}+1}r^2\right).
\end{aligned}
\end{equation}
\label{Lemma:EB_discretization_total}
\end{Lemma}
\noindent From this lemma, it therefore follows that we can choose 
\begin{equation}
\eta=\Theta\bigg(\bigg(\frac{\delta}{Ntr^2}\bigg)^{{1}/{(j_{\mathrm{max}}+1)}}\bigg)\label{eq:expression_for_eta_scaling}
\end{equation}
to ensure that the discretization error in the reduced dynamics of the system is at most $\delta/3$. This completes the analysis of the first step of the simulation algorithm.

\subsubsection{Trotterization}
From Lemma \ref{lemma:EB_Kernel}, the new memory kernel $\tilde{K}_i$ has a finite $L^1$ and $L^\infty $ norm, therefore it satisfies the condition in Eq. (\ref{eq:memory_kernel_bounded_trotter}). We can exploit Lemma \ref{lemma1} to divide the total evolution time $t$ into $T$ time steps and Trotterize each time step with a $P^\text{th}$ order product formula. The Trotterization error is bounded by
\[
\lVert \mathcal{E}_\mathrm{tro}(t,T)(\rho_S(0))-\tilde{\mathcal{E}}_S(t)(\rho_S(0))\rVert_\mathrm{tr}\leq\mathcal{O}\left(Nt\left(\frac{t}{T}\right)^P\right).
\]
We can choose $T=\mathcal{O}(t(Nt/\delta)^{1/P})$ such that the Trotterization error is bounded by $\delta/3$.

\subsubsection{Quditization bosonic modes}

Until this point, we still have a finite set of bosonic modes with annihilation operators $b_{i, j}^n$ describing the environment. In order to simulate this model on a quantum computer, we must truncate the Hilbert space of each bosonic mode $b_{i,j}^{n}$ to a finite-dimensional one.
For this, we exploit the fact that the system-environment interaction is \emph{linear} in the annihilation operators describing the environment. A physical consequence of this is that the system-environment interaction more or less acts like a system-dependent displacement on the environment modes, and thus cannot create a very large number of particles in the environment. More specifically, we show explicitly in the supplement that the system-environment state at time $\tau$ in the Trotterized model, $\ket{\psi(\tau)}$, satisfies
\[
\bra{\psi(\tau)}e^{2 N^n_{i, j}}\ket{\psi(\tau)} \leq \exp(\mathcal{O}(C_0 s_P \sqrt{\eta}(\eta + 2r))),
\]
where $N^n_{i, j} = (b_{i, j}^n)^\dagger b_{i, j}^n$ and, importantly, the right-hand side in this upper bound is independent of $\tau$. From this upper-bound, it then follows that the probability that the mode $b^n_{i, j}$ is has $d$ particles decays exponentially with $d$.
More specifically, suppose $P_{i,j}^n(d)$ is the projector to $<d$ particles for the bosonic mode $b_{i,j}^n$: $P_{i,j}^n(d)=\sum_{l=0}^{d-1}\ket{l}_{i,j,n}\bra{l}$ then, we obtain the following lemma.

\begin{Lemma}[Probabilistic bound on occupation number]

\label{Lemma4}
For any intermediate time $\tau$, the state $\ket{\psi(\tau)}$ satisfies
\begin{equation}
\lVert (I-P_{i,j}^{n}(d))\ket{\psi(\tau)}\rVert \leq C_\mathrm{trun}e^{-d}\ \mathrm{for}\ \forall \tau,
\end{equation}
where $C_\mathrm{trun}\leq\exp(\mathcal{O}(s_P^2C_0^2))$ is a constant which does not depend on $N,t$ or $\delta$.
\end{Lemma}
Lemma~\ref{Lemma4} allows us to safely restrict the system-environment interaction Hamiltonian to the finite-dimensional subspace spanned by the first 
$d$ levels of each bosonic mode. More formally, we introduce $P_d=\prod_{i,j,n}P^n_{i,j}(d)$ and the truncated Hamiltonian 
\[
\begin{aligned}
\hat{H}_\mathrm{o}^d=\sum_{i\in\mathrm{odd}}\left(H_{i,i+1}+P_d(J_{i,i+1}^\dagger\tilde{A}_i(t)+\text{h.c.})P_d\right),\\
\hat{H}_\mathrm{e}^d=\sum_{i\in\mathrm{even}}\left(H_{i,i+1}+P_d(J_{i,i+1}^\dagger\tilde{A}_i(t)+\text{h.c.})P_d\right).
\end{aligned}
\]
We define the truncated Trotter decomposition unitary $\hat{U}_\mathrm{tro}(t,T)$ as one obtained by applying Lemma \ref{lemma1} as well as Eq. (\ref{eq:trotterization_expression}) to $\hat{H}_\mathrm{o}^d$, $\hat{H}_\mathrm{e}^d$, and the corresponding channel
\[
\hat{\mathcal{E}}_\mathrm{tro}(t,T)=\mathrm{Tr}_E[\hat{U}_\mathrm{tro}(t,T)((\cdot)\otimes\rho_E(0))\hat{U}_\mathrm{tro}^\dagger(t,T)].
\]

The error introduced by this truncation can be bounded by the following lemma:

\begin{Lemma}[Error bound on the evolution after quditization]
The difference between $\mathcal{E}_\mathrm{tro}(t,T)$ and $\hat{\mathcal{E}}_\mathrm{tro}(t,T)$ can be upper bounded by

\begin{equation}
\begin{aligned} & \left\lVert \mathcal{E}_\mathrm{tro}(t,T)(\rho_S(0))-\hat{\mathcal{E}}_\mathrm{tro}(t,T)(\rho_S(0))\right\rVert_{\mathrm{tr}}\\
&\quad\quad\leq  \mathcal{O}\bigg(\sqrt{d}e^{-d}\frac{N^{2}t^{2}}{\eta^{\frac{3}{2}}}r^2\bigg).
\end{aligned}
\end{equation}
\label{Lemma:EB_quditization_total}
\end{Lemma}
\noindent Consequently, we can choose 
\begin{equation}
d=\mathcal{O}\bigg(\log\frac{N^{2}t^{2}}{\delta\eta^\frac{3}{2}}r^2\bigg)\label{eq:expression_for_d_scaling}
\end{equation}
 such that this truncation error is at most $\delta/3$. 
Now, the truncated Trotterization
$\tilde{\mathcal{E}}_\mathrm{tro}$ is generated by the finite-dimensional qudit Hamiltonian $\hat{H}_\mathrm{e(o)}$, and therefore admits efficient simulation on a digital quantum computer.

\subsubsection{Smooth cut-off on coupling function $v_i$}

Uptil now, we have assumed that $v_i$ has a compact support that lies in $[-r, r]$. If $v_i$ is not compact but is smooth and superpolynomially decaying in $t$, we will next show that a smooth
cutoff can be applied to $v_i$, which only introduces negligible overhead in complexity scaling.  One choice of the smooth cutoff function can be
constructed as 
\begin{equation}
\tilde{B}(t)=(B_{t^{*}+1}*\varphi)(t),
\end{equation}
where $t^{*}$ is the cutoff value, $*$ denotes a convolution,
\[
B_{t^{*}+1}(t)=\begin{cases}
1, & |t|\leq t^{*}+1;\\
0, & |t|>t^{*}+1
\end{cases}
\]
is the hard cutoff function and $\varphi(t)$ is a smooth mollifier
\begin{equation}
\varphi(t)=\begin{cases}
\frac{1}{\int_{-1}^{1}e^{-\frac{1}{1-\tau^{2}}}d\tau}e^{-\frac{1}{1-t^{2}}}, & |t|\leq1;\\
0, & |t|>1.
\end{cases}
\end{equation}
Then the function $\tilde{v}_i(t)=\tilde{B}(t)v_i(t)$ is smooth, finitely
supported on $[-t^{*}-2,t^{*}+2]$ and identical to $v_i(t)$ on $[-t^{*},t^{*}]$.
The error introduced by this smooth truncation is superpolynomially suppressed in $t^{*}$, which only requires an additional gate cost scaling slower than any fractional power of $Nt/\delta$ and thus can be neglected compared to $(Nt/\delta)^{1/p}$.
For example, if $v_i(t)$ decays exponentially fast, one can choose $r=t^{*}+2=\mathcal{O}(\log(Nt/\delta))$,
which only adds an logarithmic correction and is negligible.

\subsubsection{Total resource counting\label{subsub:total_resource}}

Combining the above three steps, we arrive at the desired finite-dimensional simulation channel $\mathcal{E}$, which is implemented by first preparing all the ancilla modes $\{b_{i,j}^n\}$ in their vacuum state, evolving under the Hamiltonian $P_d\tilde{H}_{SE}(\tau)P_d$ for time $t$ and tracing out the ancillas.
The total approximation error is bounded by $\delta$, which comes from summing the errors over discretization, Trotterization, and quditization steps. In the following,
we analyze the computational resources required to implement this simulation on a digital quantum computer.

After Trotterization, we have in total $T=\mathcal{O}\big(\big(Nt/\delta\big)^{\frac{1}{P}}t\big)$
time steps. In each Trotterization stage, we can, for example, use Hamiltonian simulation whose circuit depth is proportional to the norm of the simulated Hamiltonian
\citep{Low_2019Qubitization, Kieferov2019Simulating}. Since each $C_{i,j}^{n}(t)$ has finite support,
for any fixed $t$ and $i$ there can be at most $j_\mathrm{max}(\frac{2 r }{\eta}+2)$ nonzero $C_{i,j}^n(t)$s. Therefore, the norm associated with a single jump operator term in $P_d\tilde{H}_{SE}(t)P_d$ can be bounded as
\begin{equation}
\begin{aligned} & \left\Vert J_{i,i+1}^{\dagger}\sum_{n}\sum_{j=0}^{j_{\mathrm{max}}}C_{i,j}^{n}(t)P_{d}b_{i,j}^{n}P_{d}+\text{h.c.}\right\Vert\\
&\qquad\leq  2j_\mathrm{max}\left(\frac{2 r }{\eta}+2\right)\max_{i,j,n}|C_{i,j}^{n}(t)|\lVert P_{d}b_{i,j}^{n}P_{d}\rVert,\\
&\qquad\leq  \mathcal{O}\bigg( r \sqrt{\frac{d}{\eta}}\bigg).
\end{aligned}
\end{equation}
In the last inequality, we use the norm bound on $C_{i,j}^{n}$ and the
fact that $\lVert P_{d}b_{i,j}^{n}P_{d}\rVert\leq\sqrt{d}$. Consequently,
the total circuit depth is 
\begin{equation}
\begin{aligned} & \mathcal{O}\left(t\bigg(\frac{Nt}{\delta}\bigg)^{\frac{1}{P}} r \sqrt{\frac{d}{\eta}}\right)\\
&\qquad = \tilde{\mathcal{O}}\bigg(t\bigg(\frac{Nt}{\delta}\bigg)^{\frac{1}{P}}\bigg(\frac{Nt}{\delta}\bigg)^{\frac{1/2}{j_{\mathrm{max}}+1}}\left( r \right)^{1+\frac{1}{j_{\mathrm{max}}+1}}\bigg).
\end{aligned}
\label{eq:time_cost_estimation}
\end{equation}
In the second line we use Eqs. (\ref{eq:expression_for_eta_scaling})
and (\ref{eq:expression_for_d_scaling}) and ignore the logarithmic
factor. 

If $ r $ is a constant which does not depend on $t,N$,
we can choose $P=2p$ and $j_{\mathrm{max}}=p$ such that the total
circuit depth is 
\begin{equation}
\text{Circuit Depth }=\mathcal{O}\left(t\bigg(\frac{Nt}{\delta}\bigg)^{\frac{1}{p}}\right),
\end{equation}
and the total gates counting is 
\begin{equation}
\text{Elementary Gate Complexity }=\mathcal{O}\left(Nt\bigg(\frac{Nt}{\delta}\bigg)^{\frac{1}{p}}\right).
\end{equation}

In the above algorithm, introducing all truncated modes $b_{i,j}^{n}$ at once requires
$\mathcal{O}(tN/\eta \log(d))=\tilde{\mathcal{O}}(tN(Nt/\delta)^{1/(p+1)})$
ancilla qubits.
However, at any fixed time $t$ only $j_\mathrm{max}N(\frac{2 r }{\eta}+2)$ of those modes actually enter $\tilde{H}_{SE}(t)$. 
By sequentially tracing out and refreshing the ancillas, the total number of additional qudits only needs to scale as $\tilde{\mathcal{O}}(N(Nt/\delta)^{\frac{1}{p+1}})$.

\subsection{Quantum algorithm for non-dissipative non-Markovian lattice model (Theorem \ref{theoremnondi})\label{subsec:Proof-of-Theorem-non-dissipative}}

Theorem \ref{theoremnondi} follows immediately from Lemma
\ref{lemma1}. Indeed, Lemma \ref{lemma1} provides a $P^\text{th}$ order Trotterization decomposition of $U_{SE}(t,0)$
into products of $V_{\mathrm{e}}$ and
$V_{\mathrm{o}}$ as 
\begin{equation}
U_\mathrm{tro}(t,T)=\prod_{j=1}^{T}V\left(jt/T,(j-1)t/T\right),
\end{equation}
with the average Trotterization error $\lVert\langle\mathcal{R}_\mathrm{tro}\rangle\rVert\leq\mathcal{O}\left(Nt^{P+1}/T^P\right)$ from Lemma \ref{lemma1} and $V$ is defined in Eq. (\ref{eq:trotterization_expression}).

We now examine a single stage in each time step, say $V_\mathrm{e}((j+f_i)t/T,(j+f_{i-1})t/T)$ in Eq. (\ref{eq:trotterization_expression}). An analogous analysis holds for $V_\mathrm{o}((j+e_i)t/T,(j+e_{i-1})t/T)$. The former can be explicitly expressed as
\begin{equation}
\begin{aligned}&V_{\mathrm{e}}((j+f_{i})t/T,(j+f_{i-1})t/T)\\  &=\mathcal{T}e^{-i\int_{(j+f_{i-1})t/T}^{(j+f_{i})t/T}\sum_{l\in\mathrm{even}}\left(H_{l,l+1}+J_{l,l+1}\xi_{l}(t)\right)dt},\\
  &=\prod_{l\in\mathrm{even}}\mathcal{T}e^{-i\int_{(j+f_{i-1})t/T}^{(j+f_{i})t/T}\left(H_{l,l+1}+J_{l,l+1}\xi_{l}(t)\right)dt}.
\end{aligned}
\end{equation}
For each term in the above equation,
we again apply a $P^\text{th}$ order Trotter decomposition on $[(j+f_{i-1})t/T,(j+f_i)t/T]$ by introducing a series of intermediate time points $\{\tilde{f}_{k,j,i}=(j+f_{i-1})t/T+f_k(f_i-f_{i-1})t/T\}$. This yields another product formula
\begin{equation}
\begin{aligned} & \mathcal{T}e^{-i\int_{(j+f_{i-1})t/T}^{(j+f_{i})t/T}\left(H_{l,l+1}+J_{l,l+1}\xi_{j}(t)\right)dt}\\
&=  \tilde{R}_\mathrm{tro}+\prod_{k=1}^{s_P}e^{-iH_{l,l+1}(\tilde{f}_{k,j,i}-\tilde{f}_{k-1,j,i})}e^{-i\int_{\tilde{f}_{k-1,j,i}}^{\tilde{f}_{k,j,i}}J_{l,l+1}\xi_{l}(t)dt},\\
&=  \tilde{R}_\mathrm{tro}+\prod_{k=1}^{s_P}e^{-iH_{l,l+1}(\tilde{f}_{k,j,i}-\tilde{f}_{k-1,j,i})}e^{-iJ_{l,l+1}\int_{\tilde{f}_{k-1,j,i}}^{\tilde{f}_{k,j,i}}\xi_{l}(t)dt},
\end{aligned}
\label{eq:another_trotter_step}
\end{equation}
where $\tilde{R}_\mathrm{tro}$ denotes the error remainder. By the same argument as in subsection \ref{subsec:Proof-of-Lemma}, after average over $\xi_l(t)$, the error remainder $\lVert\langle\tilde{R}\rangle\rVert$ scales
as $\mathcal{O}((\Delta t)^{P+1})$. 
Since each time step of $V(jt/T,(j-1)t/T)$ contains $s_PN$ exponentials which requires a second Trotterization as in Eq. (\ref{eq:another_trotter_step}), the error per time step picks up an additional factor of $N$. More formally, by introducing
\[
\begin{aligned}
&\tilde{V}_\mathrm{e}((f_i+j)t/T,(f_{i-1}+j)t/T)\\&=\prod_{l\in\mathrm{even}}\prod_{k=1}^{s_P}e^{-iH_{l,l+1}(\tilde{f}_{k,j,i}-\tilde{f}_{k-1,j,i})}e^{-iJ_{l,l+1}\int_{\tilde{f}_{k-1,j,i}}^{\tilde{f}_{k,j,i}}\xi_l(t)dt}
\end{aligned}
\]
and
\[
\begin{aligned}
&\tilde{V}_\mathrm{o}((e_i+j)t/T,(e_{i-1}+j)t/T)\\&=\prod_{l\in\mathrm{odd}}\prod_{k=1}^{s_P}e^{-iH_{l,l+1}(\tilde{e}_{k,j,i}-\tilde{e}_{k-1,j,i})}e^{-iJ_{l,l+1}\int_{\tilde{e}_{k-1,j,i}}^{\tilde{e}_{k,j,i}}\xi_l(t)dt}
\end{aligned}
\]
with $\tilde{e}_{k,j,i}=(j+e_{i-1})t/T+e_k(e_i-e_{i-1})t/T$, we can approximate $V(jt/T),(j-1)t/T$ to its fully Trotterized version
\[
\begin{aligned}
  &\tilde{V}\bigg((j + 1) \frac{t}{T}, j\frac{t}{T}\bigg)  = \prod_{i = 1}^{s_p} \bigg[\tilde{V}_\mathrm{e}\bigg((f_{i}+j) \frac{t}{T}, (f_{i - 1}+j) \frac{t}{T} \bigg) \times \nonumber \\
  &\qquad \qquad \qquad \tilde{V}_\mathrm{o}\bigg((e_{i}+j) \frac{t}{T}, (e_{i - 1}+j) \frac{t}{T}\bigg)\bigg]
\end{aligned}
\]
with
\[
\begin{aligned}
&\left\lVert \left\langle V\left((j+1)\frac{t}{T},j\frac{t}{T}\right)-\tilde{V}\left((j+1)\frac{t}{T},j\frac{t}{T}\right)\right\rangle\right\rVert_\mathrm{tr}\\&\quad\quad\quad\leq\mathcal{O}\left(N(\Delta t)^{P+1}\right). 
\end{aligned}
\]
We remark that the notation $\langle\cdot\rangle$ denotes the average over $\xi_l(t)$.
Concatenating $T$ time steps yields the fully Trotterized evolution unitary
\begin{equation}
\tilde{U}_\mathrm{tro}(t,T)=\prod_{j=0}^{T-1}\tilde{V}\left((j+1)\frac{t}{T},j\frac{t}{T}\right)
\end{equation}
with
the total Trotterization error bound $\mathcal{O}((\Delta t)^{P+1}\times N\times T)=\mathcal{O}(Nt^{P+1}/T^{P})$.
More precisely, We consider the channel on the system given by
\[
\tilde{\mathcal{E}}_\mathrm{tro}(t,T)(\cdot)=\mathbb{E}_{\xi_l(t)}\left(\tilde{U}_\mathrm{tro}(t,T)(\cdot)\tilde{U}_\mathrm{tro}^\dagger(t,T)\right).
\]
By combining Lemma \ref{lemma1} with Eq.~(\ref{eq:another_trotter_step}) gives
\begin{equation}
\begin{aligned}
&\lVert\mathcal{E}_S(t)(\rho_S(0))-\tilde{\mathcal{E}}_\mathrm{tro}(t,T)(\rho_S(0))\rVert_\mathrm{tr}\\
&\quad   \leq\lVert\mathcal{E}_S(t)(\rho_S(0))-\mathcal{E}_\mathrm{tro}(t,T)(\rho_S(0))\rVert_\mathrm{tr}+\\
&\quad \quad \quad \lVert\mathcal{E}_\mathrm{tro}(t,T)(\rho_S(0))-\tilde{\mathcal{E}}_\mathrm{tro}(t,T)(\rho_S(0))\rVert_\mathrm{tr},\\
&\quad \leq\mathcal{O}\left(N\frac{t^{P+1}}{T^P}\right).
\end{aligned}
\end{equation}
Our algorithm is to simulate $\tilde{\mathcal{E}}_\mathrm{tro}(t,T)$ which describes a circuit ensemble with each instance containing $\mathcal{O}(TN)$ exponentials.
Each exponential is of the form: $e^{-iH_{l,l+1}(t_{2}-t_{1})}$, or $e^{-iJ_{l,l+1}\int_{t_{1}}^{t_{2}}\xi_{l}(t)dt}$
with $t_{1}$, $t_{2}$ the intermediate Trotterization time points $\tilde{f}_{k,j,i}$ or $\tilde{e}_{k,j,i}$. Since $H_{l,l+1}$ and $J_{l,l+1}$ are Hermitian, each exponential is unitary. The first type, $e^{-iH_{l,l+1}(t_{2}-t_{1})}$,
can be directly implemented as a two-qubit gate. To implement the second type,
$e^{-iJ_{l,l+1}\int_{t_{1}}^{t_{2}}\xi_{l}(t)dt}$, we can define
a new set of Gaussian random variables $\xi_{l}^{t_{2},t_{1}}=\int_{t_{1}}^{t_{2}}\xi_{l}(t)dt$
with 
\begin{equation}
\langle \xi_l^{t_2,t_1}\rangle=0, \;\langle\xi_l^{t_2,t_1}\xi_{l'}^{t_2',t_1'}\rangle=\delta_{l,l'}\int_{t_1}^{t_2}ds\int_{t_{1'}}^{t_{2'}}ds' \bar{K}_l(s,s')
\label{eq:CM_of_xi_l_t_2_t_1}
\end{equation}
such that $e^{-iJ_{l,l+1}\int_{t_{1}}^{t_{2}}\xi_{l}(t)dt}=e^{-iJ_{l,l+1}\xi_{l}^{t_{2},t_{1}}}$.
This set of Gaussian random variables $\{\xi_{l}^{t_{2},t_{1}}\}$
has zero mean and a non-vanishing covariance matrix determined by
Eq. (\ref{eq:CM_of_xi_l_t_2_t_1}). So one can sample $\xi_{l}^{t_{2},t_{1}}$
classically, then implement $e^{-iJ_{j,j+1}\xi_{j}^{t_{2},t_{1}}}$ as another two-qudit gate.

To complete the proof of theorem \ref{theoremnondi}, one can
choose a $p^\text{th}$ order Trotter decomposition with $T=\mathcal{O}(t(Nt/\delta)^{1/p})$
such that the total Trotterization error is at most $\delta$. Next, one can sample
the Gaussian random variables $\{\xi_{l}^{t_{2},t_{1}}\}$ (the variables $\{Z\}$ in Theorem \ref{theoremnondi}) classically. Once these realizations are fixed, $\tilde{U}_\mathrm{tro}(t,T)$
can be implemented on a qauntum computer with $\mathcal{O}(Nt(Nt/\delta)^{1/p})$ gates.

This strategy fails for the dissipative non-Markovian dynamics.
In that setting, the memory kernel is not symmetric, so $A_i(t)$ and $A_i^\dagger(t)$ cannot be treated as classical random variables.

\subsection{Local error analysis (Theorem \ref{theorem_error_robustness})\label{sub:proof_of_error_robustness}}

In this subsection, we show how to analyze the error of local observables.
We investigate the errors in Trotterization, discretization, and quditization
procedure separately, where the tunable parameters are chosen as the
Trotter time $\Delta t$, the segment width $\eta$, and truncation
level in bosonic spectrum $d$, respectively. For the non-dissipative
case, the total error for local observables will be just the Trotterization
error plus the possible experimental noise. For the dissipative case,
the total error should also include the discretization and quditization errors.

We can rewrite the local error as 
\begin{equation}
\begin{aligned}&\delta_{\mathrm{loc}}:=\\
&\max_{\substack{\lVert O_{X}\rVert\leq1\\
\lVert\rho_S(0)\rVert_{\mathrm{tr}}=1
}
}|\mathrm{Tr}\left(O_{X}\mathcal{E}_S(t)(\rho_S(0))\right)
-\mathrm{Tr}\left(O_{X}\mathcal{E}(\rho_S(0))\right)|,
\end{aligned}
\end{equation}
where $\rho_S(0)$ is the initial system state, $O_{X}$ is a local system
operator supported on region $X$, $\mathcal{E}_S$ is the exact time evolution
channel and $\mathcal{E}$ is the simulation channel. We also define
the local Trotterization error as 
\begin{equation}
\begin{aligned}&\delta_{\mathrm{loc,tro}}:=\\
&\max_{\substack{\lVert O_{X}\rVert\leq1\\
\lVert\rho_S\rVert_{\mathrm{tr}}=1
}
}|\mathrm{Tr}\left(O_{X}\mathcal{E}_S(t)(\rho_S)\right)
-\mathrm{Tr}\left(O_{X}\mathcal{E}_\mathrm{tro}(t,T)(\rho_S)\right)|,
\end{aligned}
\end{equation}
where $\mathcal{E}_\mathrm{tro}(t,T)$ is the Trotterization formula for
$\mathcal{E}_S$ constructed in Lemma \ref{lemma1}. The local discretization
and quditization errors are defined similarly as
\begin{equation}
\begin{aligned}&\delta_{\mathrm{loc,dis}}:=\\
&\max_{\substack{\lVert O_{X}\rVert\leq1\\
\lVert\rho_S\rVert_{\mathrm{tr}}=1
}
}|\mathrm{Tr}\left(O_{X}\mathcal{E}_S(t)(\rho_S)\right)
-\mathrm{Tr}\left(O_{X}\tilde{\mathcal{E}}_S(t)(\rho_S)\right)|,
\end{aligned}
\end{equation}
\begin{equation}
\begin{aligned}&\delta_{\mathrm{loc,qud}}:=\\
&\max_{\substack{\lVert O_{X}\rVert\leq1\\
\lVert\rho_S\rVert_{\mathrm{tr}}=1
}
}|\mathrm{Tr}\left(O_{X}\mathcal{E}_\mathrm{tro}(t,T)(\rho_S)\right)
-\mathrm{Tr}\left(O_{X}\hat{\mathcal{E}}_\mathrm{tro}(t,T)(\rho_S)\right)|,
\end{aligned}
\end{equation}
with $\tilde{\mathcal{E}}_S(t)$, $\hat{\mathcal{E}}_\mathrm{tro}(t,T)$ defined in Lemma \ref{Lemma:EB_discretization_total} and \ref{Lemma:EB_quditization_total}, respectively.

Let us first analyze $\delta_{\mathrm{loc,tro}}$. In Ref. \citep{RahulUnpunlishedNoiseRobust},
the authors have already investigated the Trotterization error for
local observables in usual closed unitary dynamics. Here we extend
their results to non-Markovian case. Our strategy is to employ the
error remainders in Lemma \ref{lemma1} to obtain the explicit expression of
errors for local observables in Trotterization formula. We further
identify that each term in the expression can be bounded by the Lieb-Robinson
bound proved in Ref. \citep{trivedi2024liebrobinsonboundopenquantum}.
Therefore, we can recover an effective light cone. More specifically,
we have

\begin{Lemma}[Error in Trotterization for local observables]

For the non-Markovian dynamics in a $D$-dimensional lattice satisfying
the condition in Lemma \ref{lemma1}, the error for local observables with a
$P^\text{th}$ order Trotter formula can be bounded by 
\begin{equation}
\delta_{\mathrm{loc,tro}}\leq\mathcal{O}\left(t^{D+1}(\Delta t)^{P}\right).
\end{equation}
\label{Lemma:local_trotter}
\end{Lemma}

We emphasize that this Lemma only relies on a bounded memory kernel
and does not require any specific form of the non-Markovian dynamics.
It holds for both non-dissipative and dissipative case. We also notice
that $\delta_{\mathrm{loc,tro}}$ can be obtained from $\delta_{\mathrm{tro}}$
in Lemma \ref{lemma1} by replacing $N$ with $t^{D}$. This is intuitively    
true, as the existence of the effectively light cone restricts that
the operator $O_{X}$ can only propagate to a region consists of $\mathcal{O}(t^{D})$
sites in time $t$. We would expect only these $\mathcal{O}(t^{D})$
sites contributing to the local error $\delta_{\mathrm{loc,tro}}$.

For the non-dissipative case, this Trotterization error is the only
systematic error inherited in the simulation algorithm. Here, in the
absence of experimental noise, we can choose $\Delta t=\mathcal{O}(\delta_\mathrm{loc}/t^{D+1})^{1/P}$
such that the error for local observables are below $\delta_\mathrm{loc}$. The
corresponding circuit depth is $\frac{t}{\Delta t}=\mathcal{O}(t(t^{D+1}/\delta_\mathrm{loc})^{1/P})$,
which is independent of the system size. On the other case, if the
experimental noise for two-qubit gate is $\gamma$, the total error
for local observables would be 
\begin{equation}
\delta_{\mathrm{loc}}\leq\mathcal{O}\left(t^{D+1}(\Delta t)^{P}\right)+\mathcal{O}\left(\gamma\left(\frac{t}{\Delta t}\right)^{D+1}\right).
\end{equation}
Here the last term counting the experimental noise is derived in
Ref. \citep{RahulUnpunlishedNoiseRobust}, which can be understood as a strict light cone imposed by the circuit structure with depth $t/\Delta t$. As a result,
we can choose $\Delta t=\mathcal{O}(\gamma^{1/(P+D+1)})$ to obtain
the optimal local error as
\begin{equation}
\delta_{\mathrm{loc}}\leq\mathcal{O}\left(t^{D+1}\gamma^{P/(P+D+1)}\right),
\end{equation}
which does not depend on $N$ as we expected.

The dissipative case is a little more complicated, as we also need
to include the local discretization and quditization error. To bound
discretization error $\delta_{\mathrm{loc,dis}}$, we introduce a
length scale $l$ and four different observable values as 
\begin{equation}
\begin{aligned}
&O(t)=\mathrm{Tr}(O_{X}\mathcal{E}_S(t)(\rho_S(0))),\\
&O^{\eta}(t)=\mathrm{Tr}(O_{X}\tilde{\mathcal{E}}_S(\rho_{S}(0))),\\
&O_{X_{[l]}}(t)=\mathrm{Tr}(O_{X}\mathcal{E}_{S,X_{[l]}}(t)(\rho_{S}(0))),\\
&O_{X_{[l]}}^{\eta}(t)=\mathrm{Tr}(O_X\tilde{\mathcal{E}}_{S,X_{[l]}}(t)(\rho_{S}(0))),
\end{aligned}
\end{equation}
where the subscript
$X[l]$ denotes that the restricted evolution where only the interaction within
distance no greater than $l$ from $X$ is turned on, see \citep{SM} for a rigorous
definition. We can thus rewrite $\delta_{\mathrm{loc,dis}}$ as 
\[
\begin{aligned}\delta_{\mathrm{loc,dis}} & =\max_{\substack{\lVert O_{X}\rVert\leq1\\
\lVert\rho_S(0)\rVert_{\mathrm{tr}}=1
}
}|O(t)-O^{\eta}(t)|,\\
 & \leq\max_{\substack{\lVert O_{X}\rVert\leq1\\
\lVert\rho_S(0)\rVert_{\mathrm{tr}}=1
}
}\bigg(|O(t)-O_{X_{[l]}}(t)|+\\
 & |O_{X_{[l]}}(t)-O_{X_{[l]}}^{\eta}(t)|+|O_{X_{[l]}}^{\eta}(t)-O^{\eta}(t)|\bigg).
\end{aligned}
\]
The first and the third term correspond to the difference between
the original evolution and the restricted evolution. Due to the existence
of Lieb-Robinson velocity, these differences are exponentially suppressed
in $l$ and are negligible if $l$ is larger than $t$, see Ref. \citep{trivedi2024liebrobinsonboundopenquantum}.
The second term corresponds to the discretization error analyzed in
Lemma \ref{Lemma:EB_discretization_total} but restricted to a subsystem whose number of sites is only $\mathcal{O}(l^{D})$.
Thus, the second term is bounded by Lemma \ref{Lemma:EB_discretization_total} with $N$ replaced
by $l^{D}$. By a proper choice of $l$, we can prove $\delta_{\mathrm{loc,dis}}$
is independent of the system size. $\delta_{\mathrm{loc,qud}}$ can
be analyzed in a similar way. More concretely, we have

\begin{Lemma}[Error in discretization and quditization for local observables]

For a $D$-dimensional many-body dynamics, the errors for local observables
in discretization and quditization procedures are bounded by
\begin{equation}
\begin{aligned}
\delta_{\mathrm{loc,dis}}\leq\tilde{\mathcal{O}}\left(t^{D+1}\eta^{j_{\mathrm{max}}+1}\right),\\ \delta_{\mathrm{loc,qud}}\leq\mathcal{O}\left(d^{2D+1/2}e^{-d}\frac{t^{2+2D}}{\eta^{D+\frac{3}{2}}}\right),
\end{aligned}
\end{equation}
respectively. The corresponding circuit depth for simulation is $\mathcal{O}\left(\frac{t}{\Delta t}\sqrt{\frac{d}{\eta}}\right)$.
\label{Lemma:local_dis_and_qud}
\end{Lemma}

These two terms resembles the error bounds we obtained in Lemma \ref{Lemma:EB_discretization_total}
and Lemma \ref{Lemma:EB_quditization_total}, by replacing $N$ with $t^{D}$ as we expected. The
factor $\sqrt{d/\eta}$ in the circuit depth comes from the norm of
the local Hamiltonian as already analyzed in Subsubsec. \ref{subsub:total_resource}.
In the absence of experimental noise, the error for local observables
are given by 
\begin{equation}
\delta_{\mathrm{loc}}\leq\delta_{\mathrm{loc,dis}}+\delta_{\mathrm{loc,tro}}+\delta_{\mathrm{loc,qud}},
\end{equation}
where the right hand side is given by Lemma \ref{Lemma:local_trotter} and \ref{Lemma:local_dis_and_qud}, respectively.
In this case, we can choose 
\begin{equation}
\begin{aligned}\eta=\tilde{\mathcal{O}}\left(\left(\frac{\delta_\mathrm{loc}}{t^{D+1}}\right)^{\frac{1}{j_\mathrm{max}+1}}\right),\ \Delta t=\mathcal{O}\left(\left(\frac{\delta_\mathrm{loc}}{t^{D+1}}\right)^{\frac{1}{P}}\right),\\
d=\mathcal{O}\left(\log\left(\frac{t^{2D+2}}{\delta_\mathrm{loc}}\times\left(\frac{t^{D+1}}{\delta_\mathrm{loc}}\right)^{\frac{D+\frac{3}{2}}{j_\mathrm{max}+1}}\right)\right),
\end{aligned}
\end{equation}
such that the total error for local observables is below $\delta_\mathrm{loc}$.
With $j_\mathrm{max}=P/2$. the circuit depth is 
\begin{equation}
\text{Circuit Depth }=\tilde{\mathcal{O}}\left(t\left(\frac{t^{D+1}}{\delta_\mathrm{loc}}\right)^{\frac{2}{P}}\right),
\end{equation}
and the number of ancillary
qubits per system site is 
\begin{equation}
\mathcal{O}\left(\frac{\log(d)}{\eta}\right)=\tilde{\mathcal{O}}\left(\left(\frac{t^{D+1}}{\delta_\mathrm{loc}}\right)^{\frac{2}{P+2}})\right),
\end{equation}
which are both independent of $N$. Theorem \ref{theorem_error_robustness} is obtained by choosing the Trotter order $P=2p$. In the presence of experimental
noise where each two-qubit gate has noisy rate $\gamma$, a proper
choose of $d$, $\eta$, $\Delta t$ leading to the total error $\delta_{\mathrm{loc}}\leq\tilde{\mathcal{O}}\left(t^{D+1}\gamma^{P/(2D+P+2)}\right)$. Similarly, we can choose $P=2p$ which leads to $\delta_\mathrm{loc}\leq\tilde{\mathcal{O}}(t^{D+1}\gamma^{p/(D+p+1)})$, a similar expression as its closed system counterpart \cite{RahulUnpunlishedNoiseRobust}. This local error expression
is $N$-independent, thus robust to local noise. We emphasize that
this choise may not be optimal. In practice, one can try to numerically
optimized $d$, $\eta$ and $\Delta t$ to obtain the optimal performance.

\subsection{Quantum algorithm for commuting, non-dissipative Lindbladian Markovian lattice models (Theorems \ref{theorem_markovian_non}) \label{subsection:proof_markovian_non}}
\subsubsection{Algorithm description}
We next consider algorithms for simulating Markovian master equation. We first consider a 1D Lindbladian with a nearest neighbour Hamiltonian and nearest neighbour Hermitian jump operators:
\begin{equation}
\begin{aligned}
    \mathcal{L}=-&i[H_{S},\cdot]+\sum_{i=1}^{N-1}\bigg[J_{i,i+1}(\cdot) J_{i,i+1}-\frac{1}{2} \{J_{i,i+1}^{2}, \cdot\}\bigg],\label{eq:Commuting_non_dissipative_Markovian}
\end{aligned}
\end{equation}
with $H_{S}=\sum_{i=1}^{N-1}H_{i,i+1}$ and $J_{i, i + 1} = J_{i, i + 1}^\dagger$. Furthermore, we will also assume that the jump operators are commuting i.e., $[J_{i,i+1},J_{j,j+1}]=0$. As described in section \ref{subsec:Simulation-of-Markovian}, the dynamics generated by this Lindbladian can instead be expressed as an ensemble average of dynamics generated by an ensemble of Hamiltonians $H_\xi(t)$ where
\begin{align}\label{eq:H_xi_repeat}
H_\xi(t) = H_S + V_\xi(t) \text{ where } V_\xi(t) = \sum_{i = 1}^{N - 1} \xi_i(t) J_{i, i + 1},
\end{align}
with $\xi_i(t)$ being a zero-mean Gaussian white noise process satisfying $\langle \xi_i(t) \xi_j(s)\rangle = \delta_{i, j}\delta(t - s)$. The channel generated by the Lindbladian in Eq.~\eqref{eq:Commuting_non_dissipative_Markovian}, $\exp(\mathcal{L}t)$, can be related to the unitary $U_\xi(t, s) = \mathcal{T}\exp(-i\int_s^t H_\xi(\tau) d\tau)$ generated by $H_\xi(t)$ in Eq.~\eqref{eq:H_xi_repeat}:
\[
\exp(\mathcal{L}t) = \mathbb{E}_\xi\big(U_\xi(t, 0) (\cdot) U_\xi(0, t)\big).
\]
To simulate $\exp(\mathcal{L}t)$, the strategy indicated by this relation is to sample a single white-noise trajectory $\xi_i(t)$, and then simulate the corresponding unitary $U_\xi(t, 0)$. 

As discussed in section \ref{subsec:Simulation-of-Markovian}, it is not possible to directly employ a time-dependent Hamiltonian simulation algorithm to simulate the Hamiltonian $H_\xi(t)$ since $\xi_{i}(t)$ is both almost surely unbounded and non-differentiable. We instead first rotate to the interaction picture with respect to $V_\xi(t)$:
since the jump operators commute, the unitary generated by $V_\xi(t)$, $W_\xi(t,0)=\mathcal{T}e^{-i\int_{0}^{t}V_{\xi}(\tau)d\tau}$, can be explictly written as: 
\begin{equation}
W_\xi(t,0)=\prod_{i=1}^{N-1}e^{-iJ_{i,i+1}\mathcal{W}_{i}(t)},
\end{equation}
where $\mathcal{W}_{i}(t)=\int_{0}^{t}\xi_{i}(\tau)d\tau$ is a sample path of the Wiener process.
Therefore, transforming $H_\xi(t)$ to the interaction picture with respect to $V_\xi(t)$, we obtain the Hamiltonian $\tilde{H}_\xi(t)$ given by
\begin{equation}
\begin{aligned}\bar{H}_\xi (t) & = W_\xi(t, 0) H_{S}W_\xi(0, t)=\sum_{i=1}^{N-1}\bar{H}_{i,i+1}(t),
\end{aligned}
\label{eq:CNM_interaction_Hamiltonian}
\end{equation}
with
\begin{equation}
\begin{aligned}
\bar{H}_{i,i+1}(t)=\prod_{i'\in\{i-1,i,i+1\}}e^{iJ_{i',i'+1}\mathcal{W}_{i'}(t)}H_{i,i+1}e^{-iJ_{i',i'+1}\mathcal{W}_{i'}(t)}.
\end{aligned}\label{eq:Markovian_non_interaction}
\end{equation}

\begin{figure}
\includegraphics[width=1\columnwidth]{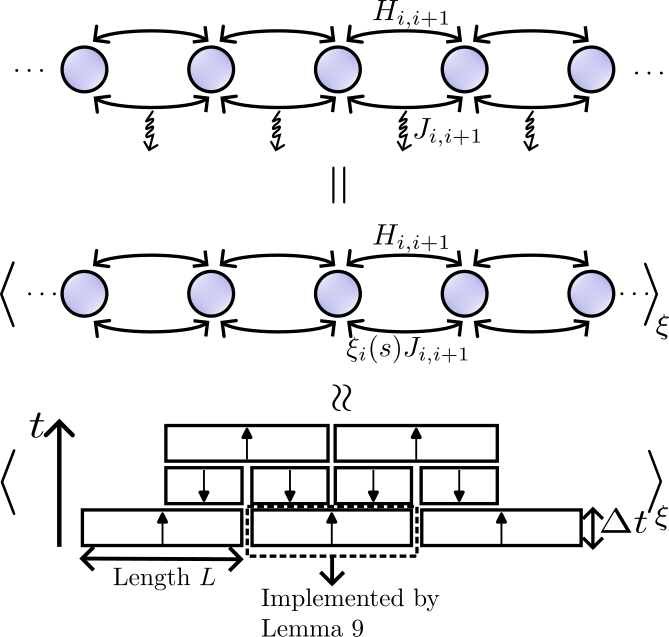}
\caption{The algorithm for simulating commuting, non-dissipative Lindbladian Markovian dynamics. The Lindbladian master equation is first unraveled into an ensemble of stochastic Schrodinger equations. For each realization in this ensemble, we use the algorithm in Ref. \cite{haah2021quantum} to split the total evolution into small blocks and implement each block using Lemma. \ref{Lemma:generalizing_time_dependent_simulation_closed}. Here, the dark uparrow (downarrow) represents the forward- (backward-) evolutions.}
\label{fig:commuting_non_dissipative}
\end{figure}

An important property is that the Wiener process can be classically
sampled. Once $\{\mathcal{W}_{i}(t)\}$ is sampled out, our task is to simulate
a time-dependent Hamiltonian Eq. (\ref{eq:CNM_interaction_Hamiltonian}).
We remark that since the Wiener process has almost surely continuous sample paths \cite{Peter2010Brownian}, $\bar{H}_{i,i+1}(t)$ is
also local, bounded, and almost surely continuous. 
Therefore, there exists a
non-vanishing Lieb-Robinson light cone for $\bar{H}_{\xi}(t)$ \cite{Anthony2023Speedlimit}.
We can apply the method from Ref. \cite{haah2021quantum} to divide the whole evolution into several time steps and several small blocks with alternating forward- and backward-evolution, as shown in Fig. \ref{fig:commuting_non_dissipative}. We can choose each time step as $\Delta t =\mathcal{O}(1/\mathrm{polylog(Nt/\delta)})$, with each block having size $L=\mathcal{O}(\mathrm{polylog}(Nt/\delta))$. 
The question is how to simulate the time-dependent Hamiltonian evolution (both forward and backward) in each small block.

A subtle point is that the Wiener process is almost surely non-differentiable, whereas most time-dependent Hamiltonian simulation algorithms assume the Hamiltonian is differentiable. To address this difficulty, we need to slightly generalize the Hamiltonian simulation algorithm. Now we assume that the Hamiltonian at each block $\bar{H}_{\xi, \mathrm{b}} $ can be expressed as a linear combination of unitaries
\begin{equation}
\bar{H}_{\xi,\mathrm{b}}(t)=\sum_{l}w_{l}(t)V_{l},
\label{eq:local_LCU}
\end{equation}
where $V_{l}$ is a local unitary operator and $w_l(t)$ the time-dependent coefficients.
Since the size of each block $L$ is bounded by $\mathcal{O}(\mathrm{polylog}(Nt/\delta))$, the number of indices $l$ can be bounded by $\mathcal{O}(\mathrm{polylog}(Nt/\delta))$. We further introduce two sets of oracles depending on a parameter $\epsilon$, which will be determined later:
\begin{equation}
\begin{aligned}
&O_{\mathrm{uni}}\ket{l}_{\mathrm{a,coef}}\ket{\psi}_{S}=\ket{l}_{\mathrm{a,coef}}V_{l}\ket{\psi}_{S}\ \forall l;\\
&O_{\mathrm{coe}}(\epsilon)\ket{l}_{\mathrm{a,coef}}\ket{j}_{\mathrm{a,time}}\ket{z}_{\mathrm{a,reg}}\\\quad\quad\quad&=\ket{l}_{\mathrm{a,coef}}\ket{j}_{\mathrm{a,time}}\ket{z\oplus w_{l}(j\epsilon)}_{\mathrm{a,reg}}\ \forall l,j,\label{eq:description_of_oracles_non_dissipative}
\end{aligned}
\end{equation}
where $\ket{l}_{\mathrm{a,coef}}, \ket{j}_{\mathrm{a,time}}, \ket{z}_{\mathrm{a,reg}}$ are different sets of ancillas and $\ket{\psi}_S$ is the system state. We have
\begin{Lemma}[Generalizing Hamiltonian simulation algorithm from \cite{Kieferov2019Simulating}]
If $\bar{H}_{i,i+1}$ is Holder-$\alpha$ continuous, the evolution in each block $\bar{U}_{\xi,\mathrm{b}}(t,0)=\mathcal{T}\exp(-i\int_0^t\bar{H}_{\xi,\mathrm{b}}(\tau))d\tau $ can be simulated with $\mathcal{O}(t\mathrm{polylog}(t/\delta)\mathrm{poly}(L))$ queries to the oracles $\{O_{\mathrm{uni}},\mathcal{O}_{\mathrm{coe}}(\epsilon)\}$ with $\epsilon=\mathcal{O}((\delta/Nt)^{1/\alpha})$ and additional $\mathcal{O}(t\mathrm{polylog}(t/\delta)\mathrm{poly}(L))$ elementary gates.\label{Lemma:generalizing_time_dependent_simulation_closed}
\end{Lemma}
In Ref. \cite{Kieferov2019Simulating}, the differentiability of the Hamiltonian is used to discretize the integration $\int_s^{s+\Delta t}\bar{H}_{i,i+1}(\tau )d\tau$ into the sum of $M$ pieces 
\[
\int_s^{s+\Delta t}\bar{H}_{i,i+1}(\tau)d\tau\approx\frac{\Delta t}{M}\sum_{j=1}^M\bar{H}_{i,i+1}(s_j)
\]
with $s_j=s+j\Delta t/M$ and $s=i\Delta t$ for some $i$. In \cite{SM}, we have shown that this discretization also holds if the Hamiltonian is only Holder continuous. The only cost is that we need to increase number of discretization points $M$ by a polynomial factor in $N,t,\delta$ to achieve the same accuracy. However, the complexity of the algorithm only scales logarithmically with $M$. Therefore, our algorithm in Lemma \ref{Lemma:generalizing_time_dependent_simulation_closed} parallels the one in Ref. \cite{Kieferov2019Simulating} with more discretization points $M$, which yields an additional $\mathrm{polylog}(Nt/\delta)$ factor in complexity. Correspondingly, the parameter $\epsilon$ in Eq. (\ref{eq:description_of_oracles_non_dissipative}) can be chosen as $\epsilon =\Delta t/M$.

According to Ref. \cite{Peter2010Brownian}, the Wiener process is Holder-$1/4$ continuous, i.e.,
\[
|\mathcal{W}_i(s)-\mathcal{W}_i(s')|\leq D_\mathrm{wie}|s-s'|^{\frac{1}{4}}
\]
almost surely with a constant $D_\mathrm{wie}$. From Eq. (\ref{eq:Markovian_non_interaction}), the Hamiltonian satisfies
\[
\lVert\bar{H}_{i,i+1}(s)-\bar{H}_{i,i+1}(s')\rVert\leq 6D_\mathrm{wie}|s-s'|^\frac{1}{4},
\]
which is Holder-$1/4$ continuous. Therefore, the algorithm in Lemma \ref{Lemma:generalizing_time_dependent_simulation_closed} can be directly applied to each block with $L=\mathcal{O}(\mathrm{polylog}(Nt/\delta))$.
The total cost of the algorithm is thus 
\[
\text{Query Complexity}=\mathcal{O}\left(Nt\mathrm{polylog}\left(\frac{Nt}{\delta}\right)\right)
\]
with addition $\mathcal{O}(Nt\mathrm{polylog}(Nt/\delta))$ elementary gates.

Finally, we remark that due to the locality of Eq. (\ref{eq:Markovian_non_interaction}), which only involves a constant number of matrix multiplications, the coefficients $w_{l}(t)$ can be efficiently computed by a classical computer once $\mathcal{W}_i(t)$ is sampled.
Therefore, if we have access to a Q-RAM with $\{w_l(j\epsilon)\}_{l,j}$ as the classical input, the two oracles in Eq. (\ref{eq:description_of_oracles_non_dissipative}) can be built from $\mathcal{O}(\mathrm{polylog}(Nt/\delta))$ elementary gates as shown in Ref. \cite{Berry2015simulating}. In total, if we have access to a Q-RAM, the elementary gate complexity is
\[
\text{Elementary Gate Complexity}=\mathcal{O}\left(Nt\mathrm{polylog}\left(\frac{Nt}{\delta}\right)\right),
\]
and the circuit depth is 
\[
\text{Circuit Depth}=\mathcal{O}\left(t\mathrm{polylog}\left(\frac{Nt}{\delta}\right)\right).
\]

In summary, our algorithm consists of two main steps:
(i), By applying the method in Ref. \cite{haah2021quantum}, the system is partitioned into $\mathcal{O}(N/\mathrm{polylog}(Nt/\delta))$ blocks, with $\mathcal{O}(\mathrm{polylog}(Nt/\delta))$
    sites in each block. The evolution time $t$ is also divided into
    $T$ time steps with alternating forward- and backward-evolutions, where $T=\mathcal{O}(t\mathrm{polylog}(Nt/\delta))$;
(ii), We apply Lemma \ref{Lemma:generalizing_time_dependent_simulation_closed} to simulate the Hamiltonian evolution in each small block and time step.

\subsubsection{Discussion on dissipative case\label{subsubsec_commuting_dissipative_th5}}
The key point of the above algorithm is to rotate the unraveled Schrodinger equation with respect to the stochastic part, and the resulting Hamiltonian can be easily computed in a constant time. This hinges on the fact that the classical Ito stochastic differential equation is solvable \citep{kloeden1992stochastic}, thereby allowing the associated Wiener process to be sampled in a purely classical manner.

Analogously, the commuting, dissipative Lindbladian dynamics, 
Eq. (\ref{eq:Lindbladian_dynamics}) with $[J_{i,i+1},J_{j,j+1}]=[J_{i,i+1},J_{j,j+1}^\dagger]=0$ for $\forall i\neq j$, 
might be simulated in a similar way. This Lindbladian can be unraveled to a Schrodinger equation
\begin{equation}
H_\mathrm{unr}=H_S+V_\mathrm{sto}(t), \label{eq:quantum_ito_stochastic}
\end{equation}
where $V_\mathrm{sto}(t)=\sum_{i=1}^{N-1}(J_{i,i+1}a_i^\dagger(t)+J_{i,i+1}^\dagger a_i(t))$, with $[a_i(t),a_{i'}^\dagger(t')]=\delta_{i,i'}\delta(t-t')$ being the bath operators. Unlike Eq. (\ref{eq:stochastic_schrodinger_eq}), here we have to introduce the environmental bath which is initialized in the joint vacuum state and traced out at the end. 

If we rotate $H_\mathrm{unr}$ to the interaction picture with respect to $V_\mathrm{sto}(t)$ as
\begin{equation}
\bar{H}_S(t)=W_\mathrm{sto}(0,t)H_SW_\mathrm{sto}(t,0),
\end{equation}
with $W_\mathrm{sto}(t,0)=\mathcal{T}\exp(-i\int_0^t V_\mathrm{sto}(t)dt)$,
the Lieb-Robinson velocity still exists. However, the difficulty is to realize $\bar{H}_S(t)$, which requires to solve the quantum Ito stochastic differential equation $W_\mathrm{sto}(t,0)$ efficiently. 

On the other hand, suppose we are given an oracle which has access to a quantum Ito solver, the commuting, dissipative Lindbladian dynamics can be simulated in a similar way as the non-dissipative case. 
More specifically, we define the quantum Ito solver oracle $O_\mathrm{qito}$ as taking three inputs $\{t,h,\ket{\psi}\}$ with $h$ a local, Hermitian operator, $\ket{\psi}$ a system-environment state, such that
 \begin{equation}
O_\mathrm{qito}(t,h,\ket{\psi})=
 \alpha W_\mathrm{sto}(0,t)hW_\mathrm{sto}(t,0)\ket{\psi}\ket{0}+\beta \ket{\varphi}\ket{1}.
 \end{equation}
Here, $\ket{0},\ket{1}$ is a register indicating if the procedure succeeds, $\alpha=1/\lVert h\rVert$, and $\ket{\varphi}$ is any other system-environment state.
The quantum Ito solver oracle allows us to realize the Hamiltonian in the interaction picture efficiently. In the Supplemental Material \citep{SM}, we further prove that we can show $\bar{H}_S(t)$ is Holder continuous such that Lemma \ref{Lemma:generalizing_time_dependent_simulation_closed} can be applied.
Therefore, the commuting, dissipative Lindbladian dynamics Eq. (\ref{eq:Lindbladian_dynamics}) can be simulated with $\mathcal{O}(NT\mathrm{polylog}(NT/\delta))$ queries to the quantum Ito solver oracle and $\mathcal{O}(NT\mathrm{polylog}(NT/\delta))$ additional elementary gates. The simulation algorithm is the same as the non-dissipative one as first divide the whole evolution into small blocks and time steps, and then simulate the unraveled Hamiltonian evolution in each block using $O_\mathrm{qito}$ and Lemma \ref{Lemma:generalizing_time_dependent_simulation_closed}. After tracing out the environmental bath, we obtain the corresponding Lindbladian dynamics. 
We leave the possibility of constructing this quantum Ito solver oracle to the future study.

\subsection{Quantum algorithms for general commuting Lindbladian Markovian lattice models based on dilation (Theorem \ref{theorem3})\label{subsec:Proof-of-Observation}}

In this subsection, we outline a strategy for simulating the Lindbladian dynamics of Eq. (\ref{eq:Lindbladian_dynamics})
on a digital quantum computer. The simulation error is defined analogously
as Eq. (\ref{eq:error_requirement}), replacing $\mathcal{E}_S(t)(\rho_S(0))$
with $e^{\mathcal{L}t}(\rho_S(0))$. The total simulation time $t$ is first divided into short  time steps of length $\Delta t$. On each time step, we dilate the original Lindbladian
dynamics into a Hamiltonian one with ancillas traced out as in Eq. (\ref{eq:definition_error_remainder_dilated}) and then implement
that dilated Hamiltonian on a quantum computer. The principal technique challenge
is to construct a $p^\text{th}$ order dilated Hamiltonian $H_\mathrm{dia}^{(p)}$ with only geometrically local interactions. Our main contribution is stated in Theorem \ref{theorem3}. 

In the previous method \citep{Linlin2024Simulating}, the Lindbladian dynamics is
first integrated via the Stochastic Schrodinger equation and the resulting quantum channel is expressed into the Kraus form. By the quantum Stinespring dilation
theorem, the Kraus operators can be realized as an isometry on the system plus one ancilla, which is 
implementable on a digital quantum computer. However, the Kraus operators produced by this procedure typically act non-locally across the system, thereby
forfeiting the geometric locality inherent to the original Lindbladian. For example,
at the first order in $\Delta t$, the Kraus operator 
can be written as 
\begin{equation}
\begin{aligned}
K_{0}&=I-i\Delta tH_{S}-\frac{\Delta t}{2}\sum_{i=1}^{N-1}J_{i,i+1}^{\dagger}J_{i,i+1},\\K_{i}&=\sqrt{\Delta t}J_{i,i+1}.
\end{aligned}
\end{equation}
The term $\sum_{i=1}^{N-1}J_{i,i+1}^{\dagger}J_{i,i+1}$ in $K_{0}$ is across the entire system. These non-local Kraus operators effectively mediate a non-local interactions to the dilated Hamiltonian, resulting in an overall cost that grows poorly with $N$.

To address this issue, we introduce a dedicated ancilla for each local jump
operator $J_{i,i+1}$. In this construction, each ancilla interacts only with its surrounding system sites, thereby preserving the locality inherent to the original Lindbladian dynamics. The total number of ancillas scales linearly in $N$.
The first order dilated Hamiltonian can be constructed as 
\begin{equation}
\begin{aligned}H_{\mathrm{dia}}^{(1)}(\sqrt{\Delta t})= & H_{S}+\\
 & \frac{1}{\sqrt{\Delta t}}\bigg(\sum_{j=1}^{N-1}J_{j,j+1}\otimes\ket{1}_{j}\bra{0}+\text{h.c.}\bigg).
\end{aligned}
\label{eq:first_order_dilated}
\end{equation}
Here $\ket{1}_{j}$ denotes the first excited state of the $j^\text{th}$ ancilla. It
can be directly verified that the dilation error $\mathcal{R}_\mathrm{dia}$ (defined in Eq. (\ref{eq:definition_error_remainder_dilated})) corresponding to $H_\mathrm{dia}^{(1)}$ scales as $\mathcal{O}((\Delta t)^2)$. We further analyze the leading order contribution  in $\mathcal{R}_\mathrm{dia}$ with respect to $\Delta t$, i.e., the superoperator $\mathcal{G}_\mathrm{dia}^{(2)}$ defined via
\begin{equation}\label{eq:second_order_dilation_error}
\mathcal{R}_\mathrm{dia}=\mathcal{G}_\mathrm{dia}^{(2)}\Delta t^2+\mathcal{O}((\Delta t)^3),
\end{equation}
and we can show that $\lVert \mathcal{G}_\mathrm{dia}^{(2)}\rVert_{\diamond}=\mathcal{O}(N)$.

To elevate the scheme to the second-order accuracy, we therefore augment $H_\mathrm{dia}^{(1)}$
with carefully chosen correction terms, and, if needed, additional ancilla levels, designed specifically to cancel
the $\mathcal{G}_\mathrm{dia}^{(2)}\Delta t^2$ component in Eq. (\ref{eq:second_order_dilation_error}).
More concretely, we have
\begin{equation}
\begin{aligned}\mathcal{G}_{\mathrm{dia}}^{(2)}(\rho_S) & =\sum_{j=1}^{N-1}\frac{1}{12}\bigg(\\
 & -i\rho_S H_{S}J_{j,j+1}^{\dagger}J_{j,j+1}+2i\rho_S J_{j,j+1}^{\dagger}H_{S}J_{j,j+1}-\\
 & i\rho_S J_{j,j+1}^{\dagger}J_{j,j+1}H_{S}+\text{h.c.}\bigg)+\cdots.
\end{aligned}
\label{eq:partial_error_term_dialated_2nd}
\end{equation}
The error term shown in Eq. (\ref{eq:partial_error_term_dialated_2nd})
can be canceled by introducing a compensating interaction term into $H_{\mathrm{dia}}^{(1)}$
as
\begin{equation}
\begin{aligned} & \frac{\Delta t}{12}\sum_{j=1}^{N-1}\bigg(2J_{j,j+1}^{\dagger}H_{S}J_{j,j+1}-\\
 & H_{S}J_{j,j+1}^{\dagger}J_{j,j+1}-J_{j,j+1}^{\dagger}J_{j,j+1}H_{S}\\
= & \frac{\Delta t}{12}\sum_{j=1}^{N-1}\bigg([J_{j,j+1}^{\dagger},H_{S}]J_{j,j+1}+J_{j,j+1}^{\dagger}[H_{S},J_{j,+1}]\bigg),
\end{aligned}
\label{eq:higher_order_dilated}
\end{equation}
and similarly for other error terms (represented by dots) in $\mathcal{G}_{\mathrm{dia}}^{(2)}$. By canceling
all such error terms,
we obtain the second order dilated Hamiltonian $H_\mathrm{dia}^{(2)}$. The full expression
can be seen in the Supplemental Material \citep{SM}.

We emphasize that the strength of compensating interaction terms we introduced is at one higher power of $\Delta t$ than the corresponding ones in $H_\mathrm{dia}^{(1)}$, so the first order dilation is intact. Furthermore, each compensator
can be written into the commutator of $[H_{S},J_{j,j+1}]$ and its Hermitian
conjugate, which guarantees the geometric locality of the $H_\mathrm{dia}^{(2)}$. This procedure
can be continued to construct the third order dilated Hamiltonian $H_\mathrm{dia}^{(3)}$. We report the explicit form of $H_\mathrm{dia}^{(3)}$ and its error
analysis in the Supplemental Material \citep{SM}. Due to the tedious algebraic
computations, it remains an open problem whether this method can be pushed to arbitary order while maintaining both geometric locality of $H_\mathrm{dia}$ and ensuring that the dilation error scales only linearly with $N$.

We require that all jump operators mutually commute. 
If this commutativity condition fails, one can show that the error component at $\mathcal{O}((\Delta t)^{3})$ in $\mathcal{R}_{\mathrm{dia}}$
cannot be canceled by introducing compensating terms
into $H_{\mathrm{dia}}^{(2)}$ via our iterative procedure. Hence, the
third-order dilated Hamiltonian for general Markovian dynamics cannot be obtained from our method.

Finally, we stress that the operator norm of our dilated Hamiltonian grows as $\mathcal{O}(\frac{1}{\sqrt{\Delta t}})$,
thereby evading the constant-norm assumption at the heart of the no-go theorem in Ref. \citep{cleve2019efficientquantumalgorithmssimulating}.
Importantly, this non-constant norm does not incur additional costs in the subseqeunt product
Trotterization decomposition, as long as one chooses a higher
order product formula. For example, to be compatible with the third-order accuracy of
$H_\mathrm{dia}^{(3)}$, one should choose a $7^\text{th}$ order product Trotter
formula.

\section{Discussion and Outlook\label{sec:Discussion-and-Outlook}}

In this paper, we study the quantum simulation of an open, geometrically local many-body quantum system. For a broad class of non-Markovian dynamics, we develop an algorithm which can achieve the nearly optimal scaling of gate complexity and circuit depth in terms of the system size $N$ and the evolution time $t$. When restricted to local observables, this algorithm is also error-robustness against local experimental noise. For Markovian dynamics, we introduce two algorithms based on the locally dilated Hamiltonian and sampling of the Wiener process, respectively, which can achieve better scaling in $N$ compared to existing methods.

Our complexity results for simulating non-Markovian dynamics have direct implications for phase classification. For open quantum systems, phases are commonly defined in two ways: (i) via continuous time evolution \cite{Coser2019Classification},---   two states lie in the same phase if they can be connected by a rapid-mixing open quantum dynamics; and (ii) via discrete time evolution \cite{Sarang2024Defining},--- two states lie in the same phase if they can be connected by local, reversible channels.  Theorem \ref{theoremdi} provides evidence that these two definitions are equivalent i.e., when translating a continuous time evolution into a discrete model, we only introduce an $N^{o(1)}$ overhead, which can be made smaller than any fractional power of $N$.

A natural question in simulating non-Markovian open quantum system is whether one can further improve our algorithm to achieve a gate complexity which scales as $Nt\log(Nt/\delta)$. In closed quantum systems, such scaling follows by employing the Lieb-Robinson bound \citep{haah2021quantum}. One might hope to recover the logarithmic overhead by letting the cutoff $j_\mathrm{max}$ in our discretization step grow appropriately with $N$ and $t$. However, this idea requires a universal bound on all higher-order derivatives of the temporal coupling function $v(\tau)$ such that $|\frac{\partial^k v(\tau)}{\partial \tau^k}|\leq \mathcal{O}(\exp(k))$. It remains unclear whether any physically realistic environment satisfies such a stringent condition.

Another interesting and key problem is the resource counting on the ancillary modes. Without recycling, our current algorithm requires $\mathcal{O}(Nt(Nt/\delta)^{1/(p+1)})$ ancillary modes in total. By contrast, a naive counting based on the rank of the underlying quantum channels would suggest a lower bound $\mathcal{O}(N)$ on ancillary modes. Closing the gap, either by reducing our ancilla overhead or by proving a tighter lower bound, remains an important direction for future work.

Finally, the optimal simulation of Markovian lattice models remains an open problem. For instance, is there a systematic scheme which can extend our dilation and Trotterization method to any order? Complementarily, deriving possible lower bounds on the gate complexity for simulating a generic geometrically local Markovian open many-body quantum system that rigorously establishes a separation between the open and closed system simulation remains an open challenge.
\begin{acknowledgments}
We thank Georgios Styliaris for valuable discussions. H. L. is supported by Forefront
Physics and Mathematics Program to Drive Transformation (FoPM), a World-leading Innovative Graduate Study
(WINGS) Program, the University of Tokyo. 
J.I.C. acknowledges funding from the Federal
Ministry of Education and Research Germany (BMBF)
via the project Almanaq.
R.T. acknowledges support from QuPIDC, an Energy Frontier Research Center, funded by the US Department of Energy (DOE), Office of Science, Basic Energy Sciences (BES), under the award number DE-SC0025620.  Work at MPQ is part of the Munich Quantum Valley, which is supported by the Bavarian
state government with funds from the Hightech Agenda
Bayern Plus.
\end{acknowledgments}
\bibliographystyle{apsrev4-2}
\bibliography{MyCollection}

\newpage
\onecolumngrid

\newcounter{equationSM} 
\newcounter{figureSM} 
\newcounter{tableSM} 
\stepcounter{equationSM} 
\setcounter{equation}{0} 
\setcounter{figure}{0} 
\setcounter{table}{0} 
\makeatletter 
\setcounter{equation}{0} 
\setcounter{figure}{0} 
\setcounter{table}{0} 
\setcounter{section}{0}
\makeatletter 
\renewcommand{\theequation}{S\arabic{equation}} 
\renewcommand{\thefigure}{S\arabic{figure}} 
\renewcommand{\thetable}{S\arabic{table}} 
\begin{center} 
{\large{\bf Supplemental Material for\\  ``\papertitle''}} 
\end{center}

In this supplemental material, we provide details in computing and
proving the main results of the main text. Concretely,
\begin{itemize}
\item In Sec. \ref{sec:summarynotations}, we summarize the notations we used in this Supplemental Material.
\item In Sec. \ref{sec:Detailed-Proof-ofLemma1}, we provide details in
proving Lemma \ref{lemma1};
\item In Sec. \ref{sec:Details-in-proving_theorem1}, we provide detailed
proof of Lemmas which are used in proving Theorem \ref{theoremdi};
\item In Sec. \ref{sec:Extension-to-general}, we extend our algorithms
in Theorem \ref{theoremdi} to general Gaussian initial states;
\item In Sec. \ref{sec:noise_robustness}, we analyze the noise robustness of our simulation procedure for local observables;
\item In Sec. \ref{Sec:SM_commuting_nondissi}, we show that the integration of quantum Ito process can also be discretized;
\item In Sec. \ref{sec:Details-proof-ofObservation1}, we provide details
in proving Theorem \ref{theorem3}.
\end{itemize}

\section{Summary of notations\label{sec:summarynotations}}
In this section, we summarize the notations we used in this Supplemental Material.
\begin{enumerate}
    \item $\vecket{\rho}$ denotes the vectorized density matrix of $\rho$ on a doubled Hilbert space
defined as 
\[
\sum_{i,j}\rho_{i,j}\ket{i}\bra{j}\to\sum_{i,j}\rho_{i,j}\ket{i}\otimes\ket{j}:=\vecket{\rho}.
\] The trace
norm on $\vecket{\rho}$ is defined as the trace norm of the original
density matrix $\rho$ as 
\begin{equation}
\lVert\vecket{\rho}\rVert_{\mathrm{tr}}=\lVert\rho\rVert_{\mathrm{tr}}=\mathrm{Tr}\sqrt{\rho^{\dagger}\rho}.
\end{equation}
\item
$\vecbra{O}$ denotes the vectorization of an observable $O$ to a doubled Hilbert space
as 
\[
\sum_{i,j}O_{i,j}\ket{i}\bra{j}\to\sum_{i,j}\bra{i}\otimes\bra{j}O_{i,j}=:\vecbra{O}.
\]
Specfically, the vectorized Identity matrix is $\vecbra{I}=\sum_{i}\bra{i}\bra{i}$, with $\{\bra{i}\}$ a set of orthonormal states.
\item 
The superoperator on the original density matrix
becomes an operator on the doubled Hilbert space. For simplicity,
we still call the operator on the doubled Hilbert space as the superoperator.
\item We will use the outline fonts to denote these superoperators, such as $\mathbb{H},\mathbb{U},\mathbb{V},\mathbb{J},\mathbb{A}\cdots.$
Remark: In this Supplemental Material, if $H$ is a Hamiltonian operator, the notation $\mathbb{H}$ corresponds to the commutator as $\mathbb{H}\vecket{\rho}=\vecket{[H,\rho]}$. It $U$ is a unitary operator, the notation $\mathbb{U}$ corresponds to the Adjoint as $\mathbb{U}\vecket{\rho}=\vecket{U\rho U^\dagger}$. 
\item We use the symbol $\langle O\rangle$ to denote $\mathrm{Tr}_E(O\rho_E(0))$. Equivalently. in the doubled space notation, we use $\langle \mathbb{O}\rangle$ to denote $\vecbraket{I|_E \mathbb{O}|\rho_E(0)}$, with $\vecbra{I}_E$ the vectorized Identity operator for the environmental space. 
\item The diamond norm of a superoperator is defined as
\begin{equation}
\lVert\mathbb{O}\rVert_{\diamond}:=\max_{\lVert\vecket{\rho}\rVert_{\mathrm{tr}}=1,n}\lVert\mathbb{O}\otimes \mathbb{I}_n\vecket{\rho}\rVert_{\mathrm{tr}},\label{eq:SM_definition_Superoperator_Norm}
\end{equation}
where $\mathbb{I}_n$ denotes the Identity channel (in the doubled space) acting on an auxiliary system of dimension $n$, and $\rho$ is any density matrix on the joint system composed of the supports of $\mathbb{O}$ and $\mathbb{I}_n$.
\item We introduce the adjoint expression for the superoperators as 
\[
\mathrm{ad}_{\mathbb{O}}:=[\mathbb{O},],
\]
\[
\mathrm{Ad}_{\mathbb{V}}(\cdot):=\mathbb{V}(\cdot)\mathbb{V}^\dagger.
\]
\item Unless explicitly noted, a repeated Greek
index, such as $\alpha, \beta$, implies summation over that index. 
\end{enumerate}

\section{Detailed Proof of Lemma \ref{lemma1}\label{sec:Detailed-Proof-ofLemma1}}

\subsection{Description of the Set up}

Suppose the original Hamiltonian is 
\begin{equation}
H_{SE}(t)=\sum_{i=1}^{N-1}H_{i,i+1}+\sum_{i=1}^{N-1}J_{i,i+1}^{\dagger}A_{i}(t)+J_{i,i+1}A_{i}^{\dagger}(t),
\end{equation}
where $J_{i,i+1},H_{i,i+1}$ act on the nearest neighbor sites $i,i+1$
with $H_{i,i+1}^{\dagger}=H_{i,i+1}$ and $\lVert H_{i,i+1}\rVert,\lVert J_{i,i+1}\rVert\leq1$.
Additionally, $A_{i}$ is the environmental bosonic operator satisfying
\begin{equation}
\langle A_{i}^{\sigma}(t)A_{j}^{\sigma'}(t')\rangle=\delta_{i,j}K_i^{(\sigma,\sigma')}(t,t'),\label{eq:SM_memory_kernel_operatorA}
\end{equation}
with 
$K^{(\sigma,\sigma')}_i(t,t')$
defined in the main text.

The Hamiltonian $\mathbb{H}_{SE}(t)$ on the doubled space,
which is an superoperator, is defined as 
\begin{equation}
\mathbb{H}_{SE}(t)\vecket{\rho_{SE}}=\vecket{[H_{SE}(t),\rho_{SE}]},
\end{equation}
where the right hand side represents the vectorized vector of $[H_{SE}(t),\rho_{SE}]$ on doubled space,
with $\rho_{SE}$ the system-environment
state. The explicit form of $\mathbb{H}_{SE}(t)$ reads
as
\begin{equation}
\mathbb{H}_{SE}(t)=\sum_{i=1}^{N-1}\mathbb{\mathbb{H}}_{i,i+1}+\sum_{i=1}^{N-1}\sum_{\alpha=1}^{4}\mathbb{J}_{i,i+1}^{\alpha}\mathbb{A}_{i}^{\alpha}(t),
\end{equation}
where $\mathbb{H}_{i,i+1}=H_{i,i+1}\otimes I-I\otimes H_{i,i+1}^{T}$,
$\mathbb{J}_{i,i+1}^{\alpha}=\{J_{i,i+1}^{\dagger}\otimes I,J_{i,i+1}\otimes I,-I\otimes J_{i,i+1}^{*},-I\otimes J_{i,i+1}^{T}\}$
and $\mathbb{A}_{i}^{\alpha}(t)=\{A_{i}(t)\otimes I,A_{i}^{\dagger}(t)\otimes I,I\otimes A_{i}^{\dagger}(t),I\otimes A_{i}(t)\}.$ Here and throughout this section, we only keep the notation $\otimes$ explicitly between the doubled space.
The memory kernel for $\mathbb{A}(t)$ follows from Eq. (\ref{eq:SM_memory_kernel_operatorA})
as
\begin{equation}
\langle\mathbb{A}_{i}^{\alpha}(t)\mathbb{A}_{j}^{\beta}(t')\rangle=\delta_{i,j}\begin{pmatrix}K_i^{(-,-)}(t,t') & K_i^{(-,+)}(t,t') & K_i^{(-,-)}(t',t)  & K_i^{(+,-)}(t',t)\\
K_i^{(+,-)}(t,t') & K_i^{(+,+)}(t,t') & K_i^{(-,+)}(t',t) & K_i^{(+,+)}(t',t)\\
K_i^{(-,-)}(t,t') & K_i^{(-,+)}(t,t') & K_i^{(-,-)}(t',t)  & K_i^{(+,-)}(t',t)\\
K_i^{(+,-)}(t,t') & K_i^{(+,+)}(t,t') & K_i^{(-,+)}(t',t) & K_i^{(+,+)}(t',t)
\end{pmatrix}_{\alpha\beta}:=\delta_{i,j}\mathbb{K}_{\alpha\beta}^i(t,t').\label{eq:doubled_space_Kernel_function}
\end{equation}
We also assume that 
\begin{equation}
\max_{t,t'}|K_i^{(\sigma,\sigma')}(t,t')|\leq m,\ \ \max_{t}\int_{-\infty}^{\infty}|K_i^{(\sigma,\sigma')}(t,t')|dt'\leq M.
\end{equation}
It is obvious that 
\begin{equation}
\sum_{\alpha=1}^4|\mathbb{K}_{\alpha\beta}^i(t,t')|,\sum_{\beta=1}^4|\mathbb{K}_{\alpha\beta}^i(t,t')|\leq 4m,
\;\;
\sum_{\alpha,\beta=1}^4|\mathbb{K}_{\alpha\beta}^i(t,t')|\leq16m.
\end{equation}

From the Holder inequality
\begin{equation}
\lVert \rho O\rVert_\mathrm{tr},\lVert O\rho\rVert_{\mathrm{tr}}\leq\lVert O\rVert\cdot\lVert\rho\rVert_{\mathrm{tr}},
\end{equation}
we immediately obtain that
\begin{equation}
\lVert\mathbb{H}_{i,i+1}\rVert_{\diamond}\leq2,\lVert\mathbb{J}_{i,i+1}^{\alpha}\rVert_{\diamond}\leq1.\label{eq:SM_bound_on_superoperator_norm}
\end{equation}

\subsection{Trotterization Formula}

In this subsection, we introduce the Trotterization formula for the following proof. Typically, we will use $\mathbb{U}$ and $\mathbb{V}$ to denote the time evolution superoperators on the doubled space.

First, we define the odd/even Hamiltonian on the doubled space as
\begin{equation}
\mathbb{H}_{\mathrm{\mathrm{o}}}:=\sum_{i\in\mathrm{odd}}\mathbb{H}_{i,i+1}+\sum_{i\in\mathrm{odd}}\sum_{\alpha=1}^{4}\mathbb{J}_{i,i+1}^{\alpha}\mathbb{A}_{i}^{\alpha},
\end{equation}
\begin{equation}
\mathbb{H}_{\mathrm{e}}:=\sum_{i\in\mathrm{even}}\mathbb{H}_{i,i+1}+\sum_{i\in\mathrm{even}}\sum_{\alpha=1}^{4}\mathbb{J}_{i,i+1}^{\alpha}\mathbb{A}_{i}^{\alpha}.
\end{equation}

The unitary evolution on the doubled space induced by different Hamiltonian
is written as

\begin{equation}
\mathbb{U}_{SE}(t_{2},t_{1}):=\mathcal{T}e^{-i\int_{t_{1}}^{t_{2}}\mathbb{H}_{SE}(t)dt},
\end{equation}
\begin{equation}
\mathbb{V}_{\mathrm{\mathrm{o}}}(t_{2},t_{1}):=\mathcal{T}e^{-i\int_{t_{1}}^{t_{2}}\mathbb{H}_{\mathrm{\mathrm{o}}}(t)dt},
\end{equation}
\begin{equation}
\mathbb{V}_{\mathrm{\mathrm{e}}}(t_{2},t_{1}):=\mathcal{T}e^{-i\int_{t_{1}}^{t_{2}}\mathbb{H}_{e}(t)dt}.
\end{equation}
It is obvious that the diamond norms for those three superoperators
are equal to $1$.

The $P^\text{th}$ order product Trotter decomposition to $\mathbb{U}_{SE}(t_{2},t_{1})$
is expressed as
\begin{equation}
\mathbb{U}_{SE}(t_{2},t_{1})\approx\mathbb{V}(t_{2},t_{1})=\prod_{i=1}^{s_P}\mathbb{V}_{\mathrm{\mathrm{e}}}(t_1+(t_2-t_1)f_i,t_1+(t_2-t_1)f_{i-1})\mathbb{V}_{\mathrm{\mathrm{o}}}(t_1+(t_2-t_1)e_{i},t_1+(t_2-t_1)e_{i-1})
\end{equation}
with $\{e_{i}\},\{f_{i}\}$ a series of intermediate time points defined in Eq. (\ref{eq:trotterization_expression}). We further assume that $b_{s_P}=a_{s_P}=1,b_{0}=a_{0}=0$ and
$0\leq b_{j},a_{j}\leq 1$. Notice that this is satisfied by the
Suzuki formula. We also define
\begin{equation}
\mathbb{V}_{j}(t_{2},t_{1})=\mathbb{V}_{\mathrm{\mathrm{e}}}(t_1+(t_2-t_1)f_i,t_1+(t_2-t_1)f_{i-1})\mathbb{V}_{\mathrm{\mathrm{o}}}(t_1+(t_2-t_1)e_{i},t_1+(t_2-t_1)e_{i-1}), \label{eq:SM_definition_1D_trotter_stage}
\end{equation}
which is the $j^\text{th}$ stage of the Trotter decomposition.

The full evolution $[0,t]$ is divided into $T$ time steps, with $\Delta t=t/T$.
The total error introduced by this Trotterization can be calculated
as 
\begin{equation}
\begin{aligned}\delta_{\mathrm{tro}}:= & \bigg\lVert\bigg\langle\mathbb{U}_{SE}(t,0)-\prod_{i=1}^{T}\mathbb{V}\bigg(i\Delta t,(i-1)\Delta t\bigg)\bigg\rangle\bigg\rVert_{\diamond},\\
= & \bigg\lVert\bigg\langle\sum_{i=1}^{T}\mathbb{U}_{SE}(t,i\Delta t)\bigg[\mathbb{U}_{SE}\bigg(i\Delta t,(i-1)\Delta t\bigg)-\mathbb{V}\bigg(i\Delta t,(i-1)\Delta t\bigg)\bigg]\times\\
 & \prod_{j=1}^{i-1}\mathbb{V}\bigg(j\Delta t,(j-1)\Delta t\bigg)\bigg\rangle\bigg\rVert_{\diamond}.
\end{aligned}
\label{eq:Chunk_error_p_order}
\end{equation}
Our task is to provide an upper bound for $\delta_{\mathrm{tro}}$.

\subsection{Some useful identities}

In this subsection, we introduce some identities corresponding to
the Wick contraction of $\mathbb{A}_{i}^{\alpha}(t)$. 
From the principle of the Wick contraction, we have the following recursive formula
\[
\langle\cdots \mathbb{A}_i^\alpha(t)\cdots\mathbb{O}\cdots\rangle=\langle\contraction{}{\cdots}{}{\mathbb{A}} \cdots \mathbb{A}_i^\alpha(t)\cdots\mathbb{O}\cdots\rangle +\langle\contraction{\cdots}{\mathbb{A}}{{}_i^\alpha(t)}{\cdots} \cdots \mathbb{A}_i^\alpha(t)\cdots\mathbb{O}\cdots\rangle+\langle\contraction{\cdots}{\mathbb{A}}{{}_i^\alpha(t)\cdots}{\mathbb{O}} \cdots \mathbb{A}_i^\alpha(t)\cdots\mathbb{O}\cdots\rangle+\langle\contraction{\cdots}{\mathbb{A}}{{}_i^\alpha(t)\cdots\mathbb{O}}{\cdots} \cdots \mathbb{A}_i^\alpha(t)\cdots\mathbb{O}\cdots\rangle,
\]
where $\cdots$ can be any expressions. If $\mathbb{O}=\mathbb{A}_j^\beta(t')$, we have $\contraction{}{\mathbb{A}}{{}_i^\alpha(t)}{\mathbb{A}}\mathbb{A}_i^\alpha(t)\mathbb{A}_j^\beta(t')=\delta_{i,j}\mathbb{K}_{\alpha\beta}^i(t,t')$. In the following, we provide the Wick-Contraction results for the other forms of $\mathbb{O}$ which we will use frequently.
Though we assume
$i\in\mathrm{odd}$, the similar identities apply if $i\in\mathrm{even}$:

\begin{equation}
\contraction{}{\mathbb{A}}{{}_i^\alpha(s)}{\mathbb{V}}\mathbb{A}_{i}^{\alpha}(s)\mathbb{V}_{\mathrm{\mathrm{o}}}(t_{2},t_{1})=-i\int_{t_{1}}^{t_{2}}\mathbb{V}_{\mathrm{\mathrm{o}}}(t_{2},t')\mathbb{J}_{i,i+1}^{\beta}\mathbb{V}_{\mathrm{\mathrm{o}}}(t',t_{1})\mathbb{K}^i_{\alpha\beta}(s,t')dt';\label{eq:contraction_with_A}
\end{equation}
\begin{equation}
\contraction{}{\mathbb{A}}{{}_i^\alpha(s)}{\mathbb{U}}\mathbb{A}_{i}^{\alpha}(s)\mathbb{U}_{SE}(t_{2},t_{1})=-i\int_{t_{1}}^{t_{2}}\mathbb{U}_{SE}(t_{2},t')\mathbb{J}_{i,i+1}^{\beta}\mathbb{U}_{SE}(t',t_{1})\mathbb{K}^i_{\alpha\beta}(s,t')dt';\label{eq:contraction_with_E}
\end{equation}
\begin{equation}
\contraction{}{\mathbb{A}}{{}_i^\alpha(s)}{Ad}\mathbb{A}_{i}^{\alpha}(s)\mathrm{Ad}_{\mathbb{V}_{\mathrm{\mathrm{o}}}(t_{1},t_{2})}=i\int_{t_{1}}^{t_{2}}\mathrm{Ad}_{\mathbb{V}_{\mathrm{\mathrm{o}}}(t_{1},t')}\mathrm{ad}_{\mathbb{J}_{i,i+1}^{\beta}}\mathrm{Ad}_{\mathbb{V}_{\mathrm{\mathrm{o}}}(t',t_{2})} \mathbb{K}^i_{\alpha\beta}(s,t').\label{eq:contraction_with_Ad}
\end{equation}
Similar identities hold if $\mathbb{A}_{i}^{\alpha}$ is on the right.
These identities can be proved by a straightforward calculation. Here
we show the proof to the second and the third identities.

For the second identity, we have

\begin{equation}
\begin{aligned} 
& \contraction{}{\mathbb{A}}{{}_i^\alpha(s)}{\mathbb{U}}\mathbb{A}_{i}^{\alpha}(s)\mathbb{U}_{SE}(t_{2},t_{1})\\
&= \contraction{}{\mathbb{A}}{{}_{i}^{\alpha}(s)\sum_{n=0}^\infty(-i)^{n}\int_{t_{2}>\tau_{1}>\cdots>\tau_{n}>t_{1}}}{\left(\mathbb{H}_{SE}(\tau_{1})\cdots\mathbb{H}_{SE}\right)}\mathbb{A}_{i}^{\alpha}(s)\sum_{n=0}^\infty(-i)^{n}\int_{t_{2}>\tau_{1}>\cdots>\tau_{n}>t_{1}}\left(\mathbb{H}_{SE}(\tau_{1})\cdots\mathbb{H}_{SE}(\tau_{n})\right),\\
&=  \sum_{n=0}^\infty\sum_{j=1}^{n}\int_{t_{2}>\tau_{1}>\cdots>\tau_{j}=t'>\cdots>\tau_{n}>t_{1}}(-i)^{n}\mathbb{H}_{SE}(\tau_{1})\cdots\contraction{}{\mathbb{A}}{{}_{i}^{\alpha}(s)}{\mathbb{H}}\mathbb{A}_{i}^{\alpha}(s)\mathbb{H}_{SE}(t')\cdots\mathbb{H}_{SE}(\tau_{n}),\\
&=  \sum_{n_{1},n_{2}=0}^\infty\int_{t_{2}>\tau_{1}>\cdots\tau_{n_{1}}>t'>\tilde{\tau}_{1}\cdots>\tilde{\tau}_{n_{2}}>t_{1}}(-i)^{n_{1}+n_{2}+1}\mathbb{H}_{SE}(\tau_{1})\cdots\mathbb{H}_{SE}(\tau_{n_{1}})\times\\
 & \contraction{}{\mathbb{A}}{{}_{i}^{\alpha}(s)}{\mathbb{H}}\mathbb{A}_{i}^{\alpha}(s)\mathbb{H}_{SE}(t')\mathbb{H}_{SE}(\tilde{\tau}_{1})\cdots\mathbb{H}_{SE}(\tilde{\tau}_{n_{2}}),\\
&=  -i\int_{t_{1}}^{t_{2}}dt'\sum_{n_{1}=0}^\infty(-i)^{n_{1}}\int_{t_{2}>\tau_{1}>\cdots\tau_{n_{1}}>t'}\mathbb{H}_{SE}(\tau_{1})\cdots\mathbb{H}_{SE}(\tau_{n_{1}})\contraction{}{\mathbb{A}}{{}_{i}^{\alpha}(s)}{\mathbb{H}}\mathbb{A}_{i}^{\alpha}(s)\mathbb{H}_{SE}(t')\times\\
 & \sum_{n_{2}=0}^\infty(-i)^{n_{2}}\int_{t'>\tilde{\tau}_{1}>\cdots>\tilde{\tau}_{n_{2}}>t_{1}}\mathbb{H}_{SE}(\tilde{\tau}_{1})\cdots\mathbb{H}_{SE}(\tilde{\tau}_{n_{2}}),\\
&=  -i\int_{t_{1}}^{t_{2}}dt'\mathbb{U}_{SE}(t_{2},t')\contraction{}{\mathbb{A}}{{}_{i}^{\alpha}(s)}{\mathbb{H}}\mathbb{A}_{i}^{\alpha}(s)\mathbb{H}_{SE}(t')\mathbb{U}_{SE}(t',t_{1}),\\
&=  -i\int_{t_{1}}^{t_{2}}\mathbb{U}_{SE}(t_{2},t')\mathbb{J}_{i,i+1}^{\beta}\mathbb{U}_{SE}(t',t_{1})\mathbb{K}_{\alpha\beta}^i(s,t')dt'.
\end{aligned}
\end{equation}
In the second equality, the Wick contraction is applied to contract
$\mathbb{A}_{i}^{\alpha}$ with $\mathbb{H}_{SE}(\tau_{j})$ and we replace
$\tau_{j}$ by $t'$. In the third equality, we relabel the integration
variables and reorder the sum. In the forth equality, we interchange
the integration order. In the fifth equality, we rewrite the time-ordered
integration to the evolution operator and in the final equality, we
use Eq. (\ref{eq:doubled_space_Kernel_function}).

The third identity is proved as 
\begin{equation}
\begin{aligned} & \contraction{}{\mathbb{A}}{{}_{i}^{\alpha}(s)}{Ad}\mathbb{A}_{i}^{\alpha}(s)\mathrm{Ad}_{\mathbb{V}_{\mathrm{\mathrm{o}}}(t_{1},t_{2})}(\cdot)\\
&=  \contraction{}{\mathbb{A}}{{}_{i}^{\alpha}(s)\mathbb{V}_{\mathrm{\mathrm{o}}}(t_{1},t_{2})(\cdot)}{\mathbb{V}}\mathbb{A}_{i}^{\alpha}(s)\mathbb{V}_{\mathrm{\mathrm{o}}}(t_{1},t_{2})(\cdot)\mathbb{V}_{\mathrm{\mathrm{o}}}(t_{2},t_{1})+\contraction{}{\mathbb{A}}{{}_{i}^{\alpha}(s)}{\mathbb{V}}\mathbb{A}_{i}^{\alpha}(s)\mathbb{V}_{\mathrm{\mathrm{o}}}(t_{1},t_{2})(\cdot)\mathbb{V}_{\mathrm{\mathrm{o}}}(t_{2},t_{1}),\\
&=  -i\int_{t_{1}}^{t_{2}}\mathbb{K}^i_{\alpha\beta}(s,t')\mathbb{V}_{\mathrm{\mathrm{o}}}(t_{1},t_{2})(\cdot)\mathbb{V}_{\mathrm{\mathrm{o}}}(t_{2},t')\mathbb{J}_{i,i+1}^{\beta}\mathbb{V}_{\mathrm{\mathrm{o}}}(t',t_{1}) dt'\\
 & +i\int_{t_{1}}^{t_{2}}\mathbb{K}^i_{\alpha\beta}(s,t')\mathbb{V}_{\mathrm{\mathrm{o}}}(t_{1},t')\mathbb{J}_{i,i+1}^{\beta}\mathbb{V}_{\mathrm{\mathrm{o}}}(t',t_{2})(\cdot)\mathbb{V}_{\mathrm{\mathrm{o}}}(t_{2},t_{1}) dt',\\
&=  -i\int_{t_{1}}^{t_{2}}\mathbb{K}^i_{\alpha\beta}(s,t')\bigg(\mathbb{V}_{\mathrm{\mathrm{o}}}(t_{1},t')\mathbb{V}_{\mathrm{\mathrm{o}}}(t',t_{2})(\cdot)\mathbb{V}_{\mathrm{\mathrm{o}}}(t_{2},t')\mathbb{J}_{i,i+1}^{\beta}\mathbb{V}_{\mathrm{\mathrm{o}}}(t',t_{1})\\
 & -\mathbb{V}_{\mathrm{\mathrm{o}}}(t_{1},t')\mathbb{J}_{i,i+1}^{\beta}\mathbb{V}_{\mathrm{\mathrm{o}}}(t',t_{2})(\cdot)\mathbb{V}_{\mathrm{\mathrm{o}}}(t_{2},t')\mathbb{V}_{\mathrm{\mathrm{o}}}(t',t_{1}) dt'\bigg),\\
&=  i\int_{t_{1}}^{t_{2}}\mathbb{K}^i_{\alpha\beta}(s,t')\mathbb{V}_{\mathrm{\mathrm{o}}}(t_{1},t')[\mathbb{J}_{i,i+1}^{\beta},\mathbb{V}_{\mathrm{\mathrm{o}}}(t',t_{2})(\cdot)\mathbb{V}_{\mathrm{\mathrm{o}}}(t_{2},t')]\mathbb{V}_{\mathrm{\mathrm{o}}}(t',t_{1}) dt',\\
&=  i\int_{t_{1}}^{t_{2}}\mathbb{K}^i_{\alpha\beta}(s,t')\mathrm{Ad}_{\mathbb{V}_{\mathrm{\mathrm{o}}}(t_{1},t')}\{\mathrm{ad}_{\mathbb{J}_{i,i+1}^{\beta}}[\mathrm{Ad}_{\mathbb{V}_{\mathrm{\mathrm{o}}}(t',t_{2})}(\cdot)]\} dt'.
\end{aligned}
\end{equation}
In the first equality, we employ Eq. (\ref{eq:contraction_with_A})
to contract $\mathbb{A}_{i}^{\alpha}(s)$ with either $\mathbb{V}_{\mathrm{\mathrm{o}}}(t_{2},t_{1})$
or $\mathbb{V}_{\mathrm{\mathrm{o}}}(t_{1},t_{2})$. Since $\mathbb{V}_{\mathrm{o}}(t_{1},t_{2})$
is a backward evolution as $t_{1}<t_{2}$, we got another minus sign
when performing the Wick contraction between $\mathbb{A}_{i}^{\alpha}(s)$
and $\mathbb{V}_{\mathrm{o}}(t_{1},t_{2})$.

\subsection{The error bound}

\subsubsection{rewrite the remainder error expression}

Let us compute each term in the summand of Eq. (\ref{eq:Chunk_error_p_order}).
For simplicity, we denote $t_{2}=i\Delta t,t_{1}=(i-1)\Delta t$ and
$\widetilde{\mathbb{V}}(t_{1},0)=\prod_{j=1}^{i-1}\mathbb{V}\big(j\Delta t,(j-1)\Delta t\big)$.
Then, we are required to calculate
\begin{equation}
\bigg\langle\mathbb{U}_{SE}(t,t_{2})\bigg[\mathbb{U}_{SE}(t_{2},t_{1})-\mathbb{V}(t_{2},t_{1})\bigg]\widetilde{\mathbb{V}}(t_{1},0)\bigg\rangle.\label{eq:rewritten_each_chunk_term}
\end{equation}
From the explicit error remainder form in Ref. \cite{childs2019nearly}
and the procedure of transforming to the interaction picture described
in the main text, we obtain
\begin{equation}
\mathbb{U}_{SE}(t_{2},t_{1})-\mathbb{V}(t_{2},t_{1})=\int_{0}^{\Delta t}\mathbb{U}_{SE}(t_{2},t_{1}+\tau)\mathbb{V}(t_{1}+\tau,t_{1})\mathcal{J}(\tau)d\tau.
\end{equation}
Here $\mathcal{J}(\tau)$ can be expressed as 
\begin{equation}
\mathcal{J}(\tau)=P\int_{0}^{1}dx(1-x)^{P-1}(\mathcal{J}_{P}^{1}(x\tau)+\mathcal{J}_{P}^{2}(x\tau))\frac{\tau^{P}}{P!},
\end{equation}
whose explicit form can be expressed in a simpler way by introducing $\tilde{e}_l(\tau)=t_1+\tau e_l$ and $\tilde{f}_l(\tau)=t_1+\tau f_l$ as
\begin{equation}
\begin{aligned}\mathcal{J}_{P}^{1}(\tau) & =\sum_{k=1}^{s_P}\sum_{w_{1}+\cdots+w_{2(k-1)}=P}\prod_{l=k-1}^{1}\mathrm{Ad}_{\mathbb{V}_{\mathrm{o}}(\tilde{e}_{l-1}(\tau),\tilde{e}_{l}(\tau))}\mathrm{ad}_{\mathbb{H}_{\mathrm{o}}(\tilde{e}_{l}(\tau))}^{w_{2l-1}}\mathrm{Ad}_{\mathbb{V}_{\mathrm{e}}(\tilde{f}_{l-1}(\tau),\tilde{f}_{l}(\tau))}\mathrm{ad}_{\mathbb{H}_{\mathrm{e}}(\tilde{f}_{l}(\tau))}^{w_{2l}}\times\\
 & \bigg(c_{k}\mathbb{H}_{\mathrm{o}}(\tilde{e}_{k-1}(\tau))+d_{k-1}\mathbb{H}_{\mathrm{\mathrm{e}}}(\tilde{f}_{k-1}(\tau))\bigg)\frac{P!}{\prod_{i=1}^{2(k-1)}w_{i}!},
\end{aligned}
\label{eq:function_to_estimate}
\end{equation}
and
\begin{equation}
\begin{aligned}\mathcal{J}_{P}^{2}(\tau) & =-\sum_{k=1}^{s_P}\sum_{w_{1}+\cdots+w_{2k}=P}\prod_{l=k}^{1}\mathrm{Ad}_{\mathbb{V}_{\mathrm{o}}(\tilde{e}_{l-1}(\tau),\tilde{e}_{l}(\tau))}\mathrm{ad}_{\mathbb{H}_{\mathrm{o}}(\tilde{e}_{l}(\tau))}^{w_{2l-1}}\mathrm{Ad}_{\mathbb{V}_{\mathrm{e}}(\tilde{f}_{l-1}(\tau),\tilde{f}_{l}(\tau))}\mathrm{ad}_{\mathbb{H}_{\mathrm{e}}(\tilde{f}_{l}(\tau))}^{w_{2l}}\times\\
 & \bigg(c_{k}\mathbb{H}_{\mathrm{o}}(\tilde{e}_{k}(\tau))+d_{k-1}\mathbb{H}_{\mathrm{\mathrm{e}}}(\tilde{f}_{k}(\tau))\bigg)\frac{P!}{\prod_{i=1}^{2k}w_{i}!}.
\end{aligned}
\end{equation}
In the above expression, $c_{k},d_{k-1}$ are two constants which
satisfy $|c_{k}|,|d_{k-1}|\leq k$. Each $\mathbb{H}_{\mathrm{o}}$,
$\mathbb{H}_{\mathrm{\mathrm{e}}}$ contains the bosonic operator$\mathbb{A}_{i}^{\alpha}$
which we need to contract with others using the Wick theorem and Eq.
(\ref{eq:doubled_space_Kernel_function}). We start by upper bounding the term involving
$\mathcal{J}_{P}^{1}(\tau)$. 

We define the monomial as the product whose factor is from $\mathbb{A}_{i}^{\alpha}\mathbb{J}_{i,i+1}^{\alpha}$,
$\mathbb{H}_{i,i+1}$ and the time evolution operator. Our strategy
comprises three steps. First, we will bound the number of monomials
in $\mathcal{J}_{P}^{1}(\tau)$. Second, we will bound the diamond
norm of each monomial after the Wick contraction. Finally, combining
these two bounds together gives us an upper bound on the Trotterization error measured in diamond norm.
We explain each step in details in the following three subsubsections.

\subsubsection{Counting the number of monomials}

Let's fix $k$ and $w_{1}\cdots w_{2(k-1)}$ at this moment. We start by bounding
the number of nested commutators as well as the support of each monomial
in $\mathcal{J}_{P}^{1}(\tau)$. Then we can arrive at the bound of
the number of monomials.

In $\mathcal{J}_{P}^{1}(\tau)$, there is at most $(P+1)$ bosonic
operators $\mathbb{A}_{i}^{\alpha}$ which are multiplied together.
Each time one of these $\mathbb{A}_{i}^{\alpha}$ is contracted with
an Ad, we will replace the original Ad by a new ad and two new Ads,
as shown in Eq. (\ref{eq:contraction_with_Ad}). In total, there will
be at most $(2P+1)$ nested-ad and $(2k-2+P+1)$ nested-Ad after we
contract all the bosonic operator $\mathbb{A}_{i}^{\alpha}$ in $\mathcal{J}_{P}^{1}(\tau)$.
However, the successive application of $\mathrm{Ad}_{\mathbb{V}_{\mathrm{o}}}$and
$\mathrm{ad}_{\mathbb{H}_{\mathrm{o}}}$will only enlarge the support
of an operator once, as they contains the same Hamiltonians. The same
thing holds for the successive application of $\mathrm{Ad}_{\mathbb{V}_{\mathrm{e}}}$and
$\mathrm{ad}_{\mathbb{H}_{\mathrm{e}}}$.

Hence, if we start from one term in $\mathbb{H}_{\mathrm{o}}(\tilde{e}_{k-1}(\tau))$
or $\mathbb{H}_{\mathrm{\mathrm{e}}}(\tilde{f}_{k-1}(\tau))$ whose support
is two, the final term has the support at most $4k-2$. This is because
there are only $2k-2$ adjoint operations (either Ad or ad) which
can enlarge the support, as stated above, and each operation enlarges
the support by two sites.

Now, let's upper bound the number of monomials in $\mathcal{J}_{P}^{1}(\tau)$
with fixed $w_{1}\cdots w_{2(k-1)}$. Each commutator ad will increase
the number of monomials by at most a factor of $2(4k-2)$ because
there will be $(2k-1)$ $\mathbb{H}_{i,i+1}$ and $(2k-1)$ $\mathbb{A}_{i}^{\alpha}\mathbb{J}_{i,i+1}^{\alpha}$in
$\mathbb{H}_{\mathrm{o}}$ or $\mathbb{H}_{\mathrm{e}}$ sharing a
non-vanishing overlap with an operator supported on $4k-2$ sites.
The total number of monomials in Eq. (\ref{eq:function_to_estimate})
is thus bounded by $2N(8k-4)^{2P+1}$. Here $2N$ comes from the number
of monomials in $c_{k}\mathbb{H}_{\mathrm{o}}(\tilde{e}_{k-1}(\tau))+d_{k-1}\mathbb{H}_{\mathrm{\mathrm{e}}}(\tilde{f}_{k-1}(\tau))$.

\subsubsection{Bounding the diamond norm for each monomial}

Next, we need to bound the diamond norm defined in Eq. (\ref{eq:SM_definition_Superoperator_Norm})
for each monomial in the above counting. As the monomial still involves
the bosonic operator $\mathbb{A}_{i}^{\alpha}$, we need to contract
all of $\mathbb{A}_{i}^{\alpha}$ before bounding the superoperator
norm. For a fixed $\mathbb{A}_{i}^{\alpha}\mathbb{J}_{i,i+1}^{\alpha}$,
we now discuss all the possible ways it can contract with:
\begin{itemize}
\item If it is contracted with another $\mathbb{A}_{j}^{\beta}\mathbb{J}_{j,j+1}^{\beta}$,
the diamond norm will increase by at most a factor of $\sum_{\alpha\beta=1}^4|\mathbb{K}^i_{\alpha\beta}(t,t')|\leq16m$.
There are at most $P$ other $\mathbb{A}_{j}^{\beta}\mathbb{J}_{j,j+1}^{\beta}$
in the monomial, which it can contract with. Thus, by this kind of
contraction, the diamond norm increases by at most a factor
of $16Pm$.
\item If it is contracted with $\mathbb{U}_{SE}(t,t_{1}+\tau)$, the superoperator
norm increases by at most a factor of $16M$, which is evident from
Eq. (\ref{eq:contraction_with_E}) as 
\begin{equation}
\begin{aligned} & \lVert\contraction{}{\mathbb{U}}{{}_{SE}(t,t_{1}+\tau)}{\mathbb{A}}\mathbb{U}_{SE}(t,t_{1}+\tau)\mathbb{A}_{i}^{\alpha}(s)\mathbb{J}_{i,i+1}^{\alpha}\rVert_{\diamond}\\
&=  \lVert -i\int_{t_{1}+\tau}^{t}dt'\mathbb{U}_{SE}(t,t')\mathbb{J}_{i,i+1}^{\beta}\mathbb{U}_{SE}(t',t_{1}+\tau)\mathbb{K}^i_{\beta\alpha}(t',s)\mathbb{J}_{i,i+1}^{\alpha}\rVert_{\diamond},\\
&\leq  \sum_{\alpha,\beta=1}^4\int_{t_{1}+\tau}^{t}|\mathbb{K}^i_{\beta\alpha}(t',s)|dt'\leq16M.
\end{aligned}
\end{equation}
\item Similarly, if it is contracted with $\mathbb{V}(t_{1}+\tau,t_{1})$,
the diamond norm increases by at most a factor of $16s_PM$ as
it is evident from Eq. (\ref{eq:contraction_with_A}) and there are
$s_P$ $\mathbb{V}_{\mathrm{o}}(\mathbb{V}_{\mathrm{e}})$ in $\mathbb{V}$.
\item If it is contracted with Ad, the diamond norm of a monomial
increases by at most a factor of $16M$ as it is evident from Eq. (\ref{eq:contraction_with_Ad}).
There are in total $(2k-2+P+1)$ Ad and hence by this kind of contraction,
the diamond norm increases by at most a factor of $16(2k+P-1)M$.
\item If it is contracted with $\widetilde{\mathbb{V}}(t_{1},0)$, the
diamond norm increased by at most a factor of $16s_PM$. Actually, for any $j$,
with Eq. (\ref{eq:contraction_with_A}), we obtain 
\begin{align}
 & \lVert\contraction{}{\mathbb{A}}{{}_{i}^{\alpha}(s)\mathbb{J}_{i,i+1}^{\alpha}}{\mathbb{V}}\mathbb{A}_{i}^{\alpha}(s)\mathbb{J}_{i,i+1}^{\alpha}\mathbb{V}\left(j\Delta t,(j-1)\Delta t \right)\rVert_{\diamond}\nonumber \\
&=  \lVert\sum_{\text{Contraction of }\mathbb{A}_i^\alpha(s) \text{ with each }\mathbb{V}_\mathrm{e}\text{ or }\mathbb{V}_\mathrm{o}}\mathbb{A}_{i}^{\alpha}(s)\mathbb{J}_{i,i+1}^{\alpha}\prod_{m=1}^{s_P}\mathbb{V}_{\mathrm{e}}((j-1+f_{m})\Delta t,(j-1+f_{m-1})\Delta t)\mathbb{V}_{\mathrm{o}}((j-1+e_{m})\Delta t,(j-1+e_{m-1})\Delta t)\rVert_{\diamond}\nonumber, \\
&\leq  \sum_{\alpha,\beta=1}^4\sum_{m=1}^{s_P}\int_{(j-1+\min\{f_{m-1},e_{m-1}\})\Delta t}^{(j-1+\max\{f_{m},e_{m}\})\Delta t}|\mathbb{K}^i_{\alpha\beta}(s,t')|dt'\nonumber, \\
&\leq  s_P\sum_{\alpha,\beta=1}^4\int_{(j-1)\Delta t}^{j\Delta t}|\mathbb{K}^i_{\alpha\beta}(s,t')|dt'.
\end{align}
Here we use the assumption that $0\leq a_{m},b_{m}\leq 1$.
Therefore, 
\begin{equation}
\begin{aligned} & \lVert\contraction{}{\mathbb{A}}{{}_{i}^{\alpha}(s)\mathbb{J}_{i,i+1}^{\alpha}}{\widetilde{\mathbb{V}}}\mathbb{A}_{i}^{\alpha}(s)\mathbb{J}_{i,i+1}^{\alpha}\widetilde{\mathbb{V}}(t_{1},0)\rVert_{\diamond}\\
&=  \left\lVert\sum_{j^*=1}^{t_1/\Delta t}\contraction[2.5ex]{}{\mathbb{A}}{{}_{i}^{\alpha}(s)\mathbb{J}_{i,i+1}^{\alpha}\left(\prod_{j=j^*+1}^{t_1/\Delta t}\mathbb{V}\left(j\Delta t,(j-1)\Delta t\right)\right)}{\mathbb{V}}\mathbb{A}_{i}^{\alpha}(s)\mathbb{J}_{i,i+1}^{\alpha}\left(\prod_{j=j^*+1}^{t_1/\Delta t}\mathbb{V}\left(j\Delta t,(j-1)\Delta t\right)\right)\mathbb{V}\left(j^*\Delta t,(j^*-1)\Delta t\right)\left(\prod_{j=1}^{j^*-1}\mathbb{V}\left(j\Delta t,(j-1)\Delta t\right)\right)\right\rVert_{\diamond},\\
&\leq  s_P\sum_{\alpha,\beta=1}^4\sum_{j=1}^{t_1/\Delta t}\int_{(j-1)\Delta t}^{j\Delta t}|\mathbb{K}^i_{\alpha\beta}(s,t')|dt',\\
&\leq  s_P\sum_{\alpha,\beta=1}^4\int_{-\infty}^{\infty}|\mathbb{K}^i_{\alpha\beta}(s,t')|dt'\leq16s_PM.
\end{aligned}
\end{equation}
\end{itemize}
The above are all possible contractions of one $\mathbb{A}_{i}^{\alpha}\mathbb{J}_{i,i+1}^{\alpha}$.
All together, after we contract one $\mathbb{A}_{i}^{\alpha}\mathbb{J}_{i,i+1}^{\alpha}$,
the diamond norm is increased by at most a factor of $\max\big\{16Pm+(32k+16P+32s_P)M,1\big\}<\big(16Pm+(32k+16P+32s_P)M+1\big)$.
We have at most $(P+1)$ $\mathbb{A}_{i}^{\alpha}\mathbb{J}_{i,i+1}^{\alpha}$
in each monomial, thus, the diamond norm of each monomial is
bounded by $2^{P+1}k\big(16Pm+(32k+16P+32s_P)M+1\big)^{P+1}$. Here $k$
is from the norm bound on $c_{k},d_{k-1}.$ The factor of $2^{P+1}$
is from $\lVert\mathbb{H}_{i,i+1}\rVert_{\diamond}\leq2$ and there will
be at most $P+1$ $\mathbb{H}_{i,i+1}$ in each monomial.

\subsubsection{The Trotter error scaling}

Combining the above bounds on the number of monomials and on the superoperator
norm of each monomial, we obtain
\begin{equation}
\begin{aligned} & \bigg\lVert\bigg\langle\mathbb{U}_{SE}(t,t_{1}+\tau)\mathbb{V}(t_{1}+\tau,t_{1})\times\\
 & \sum_{w_{1}+\cdots+w_{2(k-1)}=P}\prod_{l=k-1}^{1}\mathrm{Ad}_{\mathbb{V}_{\mathrm{o}},\tilde{e}_{l-1}(\tau),\tilde{e}_{l}(\tau)}\mathrm{ad}_{\mathbb{H}_{\mathrm{o}}(\tilde{e}_{l}(\tau))}^{w_{2l-1}}\mathrm{Ad}_{\mathbb{V}_{\mathrm{e}},\tilde{f}_{l-1}(\tau),\tilde{f}_{l}(\tau)}\mathrm{ad}_{\mathbb{H}_{\mathrm{e}}(\tilde{f}_{l}(\tau))}^{w_{2l}}\times\\
 & \bigg(c_{k}\mathbb{H}_{\mathrm{o}}(\tilde{e}_{k-1}(\tau))+d_{k-1}\mathbb{H}_{\mathrm{\mathrm{e}}}(\tilde{f}_{k-1}(\tau))\bigg)\widetilde{\mathbb{V}}(t_{1},0)\frac{P!}{\prod_{i=1}^{2(k-1)}w_{i}!}\bigg\rangle\bigg\rVert_{\diamond}\\
&\leq  \sum_{w_{1}+\cdots+w_{2(k-1)}=P}\frac{P!}{\prod_{i=1}^{2(k-1)}w_{i}!}\times\\
 & 2N(8k-4)^{2P+1}k\bigg(32Pm+(64k+32P+64s_P)M+2\bigg)^{P+1},\\
&=  (2k-2)^{P}\times2N(8k-4)^{2P+1}k\bigg(32Pm+(64k+32P+64s_P)M+2\bigg)^{P+1}.
\end{aligned}
\end{equation}

Following the same procedure, we can bound the error remainder related
to$\mathcal{J}_{P}^{(2)}$, which is 
\begin{equation}
(2k)^{P}\times2N(8k+4)^{2P+1}k\bigg(32Pm+(64k+32P+64s_P+64)M+2\bigg)^{P+1}.
\end{equation}

Thus, the upper bound for Eq. (\ref{eq:rewritten_each_chunk_term})
is 
\begin{equation}
\begin{aligned} & \bigg\lVert\bigg\langle\mathbb{U}_{SE}(t,t_{2})\bigg[\mathbb{U}_{SE}(t_{2},t_{1})-\mathbb{V}(t_{2},t_{1})\bigg]\widetilde{\mathbb{V}}(t_{1},0)\bigg\rangle\bigg\rVert_{\diamond}\\
&\leq  \int_{0}^{\Delta t}d\tau\int_{0}^{1}dx(1-x)^{P-1}\frac{\tau^{P}}{(P-1)!}\times\\
 & \sum_{k=1}^{s_P}k\bigg\{(2k-2)^{P}2N(8k-4)^{2P+1}\bigg(32Pm+(64k+32P+64s_P)M+2\bigg)^{P+1}\\
 & +(2k)^{P}2N(8k+4)^{2P+1}\bigg(32Pm+(64k+32P+64s_P+64)M+2\bigg)^{P+1}\bigg\},\\
&=  \frac{(\Delta t)^{P+1}N}{(P+1)!}\times\\
 & \sum_{k=1}^{s_P}k\bigg\{(2k-2)^{P}\times2(8k-4)^{2P+1}\bigg(32Pm+(64k+32P+64s_P)M+2\bigg)^{P+1}\\
 & +(2k)^{P}\times2(8k+4)^{2P+1}\bigg(32Pm+(64k+32P+64s_P+64)M+2\bigg)^{P+1}\bigg\},\\
&=  \mathcal{O}\bigg(N(\Delta t)^{P+1}\bigg).
\end{aligned}
\end{equation}
The total Trotterization error after summing over all the chunks in
Eq. (\ref{eq:Chunk_error_p_order}) is thus
\begin{equation}
\delta_{\mathrm{tro}}\leq\mathcal{O}\bigg(N\frac{t^{P+1}}{T^{P}}\bigg),
\end{equation}
the same as the $P^\text{th}$ order product Trotter formula for a closed
quantum system. We complete the proof of Lemma \ref{lemma1}.

\section{Details in proving Theorem \ref{theoremdi}\label{sec:Details-in-proving_theorem1}}

\subsection{Some property of the coefficients $C_{i,j}^{n}(t)$}

For convenience, in this subsection we provide some properties of
the time-dependent expansion coefficients $C_{i,j}^{n}(t)$.

In the main text, $C_{i,j}^{n}(t)$ is defined as 
\begin{equation}
C_{i,j}^{n}(t)=\int_{n\eta}^{(n+1)\eta}v_i(t-s)P_{j}^{n}(s)ds.\label{eq:SM_definition_of_Cjn}
\end{equation}
Here $v_i(t)$ is the temporal coupling function, which is compactly
supported on $t\in[- r , r ]$ and bounded
by $|v(t)|\leq C_0$. $P_{j}^{n}(\tau)$ is the local orthonormal
polynomial basis. From the Cauchy-Schwartz inequality, the norm of
$C_{i,j}^{n}(t)$ is bounded by
\begin{equation}
|C_{i,j}^{n}(t)|\leq\bigg[\int_{n\eta}^{(n+1)\eta}|v_i(t-s)|^{2}ds\int_{n\eta}^{(n+1)\eta}\bigg(P_{j}^{n}(s)\bigg)^{2}ds\bigg]^{\frac{1}{2}}\leq C_0\sqrt{\eta}.
\end{equation}

For any fixed $n$, $C_{i,j}^{n}(t)$ is non-vanishing only if 
\begin{equation}
t-(n+1)\eta\leq r \ \mathrm{and}\ t-n\eta\geq- r .
\end{equation}
Therefore, $C_{i,j}^{n}(t)\neq0$ only if 
\begin{equation}
t\in[n\eta- r ,(n+1)\eta+ r ],\label{eq:SM_support_of_t_for_C_ijn}
\end{equation}
 and hence 
\begin{equation}
\int_{-\infty}^{\infty}|C_{i,j}^{n}(\tau)|d\tau\leq C_0\sqrt{\eta}(2 r +\eta).\label{eq:SM_bounding_on_integral_Cjn}
\end{equation}

Similarly, for any fixed $t$, $C_{i,j}^{n}(t)$ is non-vanishing only
if
\begin{equation}
\frac{t- r }{\eta}-1\leq n\leq\frac{t+ r }{\eta}.\label{eq:SM_region_of_n_fixed_t}
\end{equation}
Consequently, for a fixed $t,i$, there are at most $j_{\mathrm{max}}\big(\frac{2 r }{\eta}+2\big)$
non-zero $C_{i,j}^{n}(t)$.

\subsection{Proof of Lemma \ref{lemma:EB_Kernel}}

The new memory kernel related to the discretized Hamiltonian $\tilde{H}_{SE}(t)$
is
\begin{equation}
\delta_{i,i'}\tilde{K}_i^{(-,+)}(t,t')=\sum_{n,n'}\sum_{j,j'=0}^{j_{\mathrm{max}}}\langle C_{i,j}^{n}(t)b_{i,j}^{n}C_{i',j'}^{n'*}(t')b_{i',j'}^{ n'\dagger}\rangle=\delta_{i,i'}\sum_{n}\sum_{j=0}^{j_{\mathrm{max}}}C_{i,j}^{n}(t)C_{i,j}^{n*}(t').
\end{equation}
Using the definition of $C_{i,j}^{n}(t)$ in Eq. (\ref{eq:SM_definition_of_Cjn}),
we can rewrite $\tilde{K}_i^{(-,+)}(t,t')$ as 
\begin{equation}
\begin{aligned}\tilde{K}_i^{(-,+)}(t,t') & =\sum_{n}\sum_{j=0}^{j_{\mathrm{max}}}C_{i,j}^{n}(t)C_{i,j}^{n*}(t'),\\
 & =\sum_{n}\sum_{j=0}^{j_{\mathrm{max}}}\int_{n\eta}^{(n+1)\eta}v_i(t-s_{1})P_{j}^{n}(s_{1})C_{i,j}^{n*}(t')ds_{1},\\
 & =\sum_{n}\int_{n\eta}^{(n+1)\eta}v_i(t-s_{1})ds_{1}\sum_{j=0}^{j_{\mathrm{max}}}P_{j}^{n}(s_{1})C_{i,j}^{n*}(t').
\end{aligned}
\end{equation}
Recall that from the definition, the original memory kernal $K_i^{(-,+)}(t,t')$ is written as
\begin{equation}
K_i^{(-,+)}(t,t')=\int_{-\infty}^\infty  v_i(t-s_{1})v_i^{*}(t'-s_{1})ds_{1},
\end{equation}
we obtain 
\begin{equation}
\begin{aligned} & |\tilde{K}_i^{(-,+)}(t,t')-K_i^{(-,+)}(t,t')| \leq  \sum_{n}\int_{n\eta}^{(n+1)\eta}\bigg|v_i(t-s_{1})\bigg(\sum_{j=0}^{j_{\mathrm{max}}}P_{j}^{n}(s_{1})C_{i,j}^{n*}(t')-v^{*}_i(t'-s_{1})\bigg)\bigg|ds_{1}.
\end{aligned}
\end{equation}

If we recognize $v^{*}_i(t'-s_{1})$ as a function of $s_{1}$
with $t'$ fixed, $\sum_{j=0}^{j_{\mathrm{max}}}P_{j}^{n}(s_{1})C_{i,j}^{n*}(t')$
is a $j_{\mathrm{max}}^\text{th}$ order polynomial fitting for $v^{*}_i(t'-s_{1})$
on the region $s_{1}\in[n\eta,(n+1)\eta)$. So the error in the
brackets should be small. More concretely, the Taylor's remainder
theory allows us to rewrite $v_i^{*}(t'-s_{1})$ as 
\begin{equation}
v_i^{*}(t'-s_{1})=f_{t'}^{*}(s_{1}-n\eta)+R_{t'}^{*}(s_{1}).
\end{equation}
Here $f_{t'}^{*}(s_{1}-n\eta)$ is a polynomial on $s_{1}\in[n\eta,(n+1)\eta)$
whose order is $j_{\mathrm{max}}$. $R^{*}_{t'}$ is the Taylor's remainder,
which can can be upper bounded as 
\begin{equation}
|R_{t'}(s_{1})|\leq\frac{\eta^{j_{\mathrm{max}}+1}\lVert v_i^{(j_{\mathrm{max}}+1)}\rVert_{\infty}}{(j_{\mathrm{max}}+1)!}\ \mathrm{for}\ s_{1}\in[n\eta,(n+1)\eta),
\end{equation}
where $\lVert v_i^{(j_{\mathrm{max}}+1)}\rVert_{\infty}=\max_{\tau}|\frac{\partial^{j_{\mathrm{max}}+1}}{\partial\tau^{j_{\mathrm{max}}+1}}v_i(\tau)|=D_{j_\mathrm{max}+1}$ defined in Eq. (\ref{eq:smoothness_assumption}) and we assume that it is uniformly bounded for $\forall i$.
Because $f^{*}$ is a polynomial whose order is $j_{\mathrm{max}}$,
it can be exactly captured by $P_{j}^{n}(s_{1})$ for $s_{1}\in[n\eta,(n+1)\eta)$.
Thus, we obtain
\begin{equation}
\begin{aligned} & |\tilde{K}_i^{(-,+)}(t,t')-K_i^{(-,+)}(t,t')|\\
&\leq  \frac{\eta^{j_{\mathrm{max}}+1}D_{j_\mathrm{max}+1}}{(j_{\mathrm{max}}+1)!}\sum_{n}\int_{n\eta}^{(n+1)\eta}\bigg|v_i(t-s_{1})\bigg|ds_{1},\\
& \leq  \frac{\eta^{j_{\mathrm{max}}+1}D_{j_\mathrm{max}+1}}{(j_{\mathrm{max}}+1)!}2C_0 r .
\end{aligned}
\label{eq:error_bound_by_discretization}
\end{equation}
In the last inequality, we use the fact that $v_i(t)$ is bounded and
compactly supported. Furthermore, since $K_i^{(-,+)}(t,t')\neq0$ only
if $t'\in[t-2 r ,t+2 r ]$ and 
\begin{equation}
\tilde{K}_i^{(-,+)}(t,t')\neq0\;
\text{only if}\; t'\in[t-2 r -\eta,t+2 r +\eta],\label{eq:SM_finite_support_of_tilde_K}
\end{equation}
which is a consequence of the finite support of $C_{i,j}^{n}(t)$, we
have

\begin{equation}
\begin{aligned} & \int_{0}^{t}\int_{0}^{t}|\tilde{K}_i^{(-,+)}(s,s')-K_i^{(-,+)}(s,s')|dsds'\\
&\leq  t(4 r +2\eta)\frac{\eta^{j_{\mathrm{max}}+1}D_{j_\mathrm{max}+1}}{(j_{\mathrm{max}}+1)!}2C_0 r, \\
&=  \mathcal{O}\left(t\eta^{j_{\mathrm{max}}+1}r^2\right).
\end{aligned}
\end{equation}

Last, we will show the $L^1$ bound of $\tilde{K}_i$. We start from
\begin{equation}
\begin{aligned}
|\tilde{K}_{i}^{(-,+)}(t,t')| & =|\sum_{n}\sum_{j=0}^{j_{\mathrm{max}}}C_{i,j}^{n}(t)C_{i,j}^{n*}(t')|,\\
 & \leq\sqrt{\eta}C_{0}|\sum_{n}\sum_{j=0}^{j_{\mathrm{max}}}C_{i,j}^{n}(t)|,\\
 & \leq\sqrt{\eta}C_{0}(j_{\mathrm{max}}+1)\sqrt{\eta}C_{0}(\frac{2 r }{\eta}+1),\\
 & \leq C_{0}^{2}(j_{\mathrm{max}}+1)(2 r +1),
\end{aligned}
\end{equation}
where in the first inequality, we used the fact that $|C_{i,j}^{n}(t)|\leq C_{0}\sqrt{\eta}$, in the second inequality, we used Eq. (\ref{eq:SM_region_of_n_fixed_t}) and in the last inequality we used $\eta<1$. Combining the above result and Eq. (\ref{eq:SM_finite_support_of_tilde_K}), we arrive at
\begin{equation}
\int_{-\infty}^{\infty}|\tilde{K}_{i}^{(-,+)}(t,t')|dt'\leq2C_{0}^{2}(j_{\mathrm{max}}+1)(2 r +1)^{2}.
\end{equation}
Therefore, we complete the proof of Lemma \ref{lemma:EB_Kernel}.

\subsection{Proof of Lemma \ref{Lemma4}\label{subsec:Proof-of-Lemma4}}

We denote the number operator for the mode $b_{i,j}^{n}$ as $N_{i,j}^{n}=b_{i,j}^{n\dagger}b_{i,j}^{n}$.
We will use the following Lemma to prove Lemma \ref{Lemma4}:

\begin{Lemma}[Bounding on occupation number from expectation value] 

Consider a system-environment state $\ket{\psi(\tau)}$, if there exists some constant $\theta, C$ such that
\begin{equation}
\bra{\psi(\tau)}e^{\theta N_{i,j}^{n}}\ket{\psi(\tau)}\leq e^{C},
\end{equation}
then 
\begin{equation}
\lVert(I-P_{i,j}^{n}(d))\ket{\psi(\tau)}\rVert\leq e^{\frac{C}{2}}e^{-\frac{\theta d}{2}}.
\end{equation}
\label{Lemma:bounding_from_expectation_value}
\end{Lemma}
\begin{proof}From the definition, we have
\begin{equation}
\bra{\psi(\tau)}(I-P_{i,j}^{n}(d))\ket{\psi(\tau)}e^{\theta d}\leq\bra{\psi(\tau)}e^{\theta N_{i,j}^{n}}\ket{\psi(\tau)}\leq e^{C}.
\end{equation}
Rearranging the above inequality leads to 
\begin{equation}
\lVert(I-P_{i,j}^{n}(d))\ket{\psi(\tau)}\rVert\leq e^{\frac{C}{2}}e^{-\frac{\theta d}{2}},
\end{equation}
which completes the proof.
\end{proof}

In order to use Lemma \ref{Lemma:bounding_from_expectation_value},
we need to upper bound $\bra{\psi(\tau)}e^{\theta N_{i,j}^{n}}\ket{\psi(\tau)}$
for any $\tau$. For convenience, we introduce the falling factory
which is defined as 
\begin{equation}
\underline{(N_{i,j}^{n})^{k}}:=\prod_{l=0}^{k-1}(N_{i,j}^{n}-l)=(b_{i,j}^{n\dagger})^{k}(b_{i,j}^{n})^{k},
\end{equation}
and
\begin{equation}
F_{i,j}^{n,k}(\tau)=\bra{\psi(\tau)}\underline{(N_{i,j}^{n})^{k}}\ket{\psi(\tau)}.
\end{equation}
Our starting point would be the equation of motion for $F_{i,j}^{n,k}(\tau)$.
It is straight forward to derive that
\begin{equation}
\begin{aligned}\frac{d}{d\tau}F_{i,j}^{n,k}(\tau) & =i\bra{\psi(\tau)}[\tilde{H}_{SE}(\tau),\underline{(N_{i,j}^{n})^{k}}]\ket{\psi(\tau)},\\
 & =ikC_{i,j}^{n}(\tau)\bra{\psi(\tau)}J_{i,i+1}^{\dagger}(b_{i,j}^{n\dagger})^{k-1}(b_{i,j}^{n})^{k}\ket{\psi(\tau)}-\\
 & ikC_{i,j}^{n*}(\tau)\bra{\psi(\tau)}J_{i,i+1}(b_{i,j}^{n\dagger})^{k}(b_{i,j}^{n})^{k-1}\ket{\psi(\tau)},\\
 & =ikC_{i,j}^{n}(\tau)\bra{\psi(\tau)}J_{i,i+1}^{\dagger}\underline{(N_{i,j}^{n})^{k-1}}b_{i,j}^{n}\ket{\psi(\tau)}-\\
 & ikC_{i,j}^{n*}(\tau)\bra{\psi(\tau)}J_{i,i+1}b_{i,j}^{n\dagger}\underline{(N_{i,j}^{n})^{k-1}}\ket{\psi(\tau)}.
\end{aligned}
\label{eq:SM_new_prove_derivative_of_Fijnk}
\end{equation}
By Cauchy-Schwarz inequality, we can upper bound $\bra{\psi(\tau)}J_{i,i+1}^{\dagger}\underline{(N_{i,j}^{n})^{k-1}}b_{i,j}^{n}\ket{\psi(\tau)}$
as 
\begin{equation}
\begin{aligned} & |\bra{\psi(\tau)}J_{i,i+1}^{\dagger}\underline{(N_{i,j}^{n})^{k-1}}b_{i,j}^{n}\ket{\psi(\tau)}|\\
&\leq  \left(\bra{\psi(\tau)}J_{i,i+1}^{\dagger}\underline{(N_{i,j}^{n})^{k-1}}J_{i,i+1}\ket{\psi(\tau)}\bra{\psi(\tau)}b_{i,j}^{n\dagger}\underline{(N_{i,j}^{n})^{k-1}}b_{i,j}^{n}\ket{\psi(\tau)}\right)^{\frac{1}{2}},\\
&=  \left(\bra{\psi(\tau)}\sqrt{\underline{(N_{i,j}^{n})^{k-1}}}J_{i,i+1}^{\dagger}J_{i,i+1}\sqrt{\underline{(N_{i,j}^{n})^{k-1}}}\ket{\psi(\tau)}\bra{\psi(\tau)}\underline{(N_{i,j}^{n})^{k}}\ket{\psi(\tau)}\right)^{\frac{1}{2}},\\
&\leq  \left(\lVert J_{i,i+1}^{\dagger}J_{i,i+1}\rVert\bra{\psi(\tau)}\underline{(N_{i,j}^{n})^{k-1}}\ket{\psi(\tau)}\bra{\psi(\tau)}\underline{(N_{i,j}^{n})^{k}}\ket{\psi(\tau)}\right)^{\frac{1}{2}},
\end{aligned}\label{eq:SM_CS_inequality_used_in_prove}
\end{equation}
where in the equality, we used the fact that $\underline{(N_{i,j}^{n})^{k-1}}$
is positive and commutes with $J_{i,i+1}$ (as one acts on bosonic
mode and one acts on the system), and also $b_{i,j}^{n\dagger}\underline{(N_{i,j}^{n})^{k-1}}b_{i,j}^{n}=\underline{(N_{i,j}^{n})^{k}}$
from the definition. A similar bound holds for the second term in
Eq. (\ref{eq:SM_new_prove_derivative_of_Fijnk}). Therefore, we arrive
at 
\begin{equation}
\frac{d}{d\tau}F_{i,j}^{n,k}(\tau)\leq2k|C_{i,j}^{n}(\tau)|\sqrt{F_{i,j}^{n,k}(\tau)F_{i,j}^{n,k-1}(\tau)},
\end{equation}
or equivalently
\begin{equation}
\frac{d}{d\tau}\sqrt{F_{i,j}^{n,k}(\tau)}\leq k|C_{i,j}^{n}(\tau)|\sqrt{F_{i,j}^{n,k-1}(\tau)}.\label{eq:SM_recursive_inequality_for_F_nkij}
\end{equation}
This is an recursive inequality with the initial condition $F_{i,j}^{n,0}(\tau)=1$
and $F_{i,j}^{n,k}(0)=0$ for $\forall k\geq1$. To solve the above
recursive inequality, we first assume purely forward evolution (i.e.,
neglecting the Trotterization procedure at this stage), and later
account for the effect of Trotterization. Setting $k=1$ in Eq. (\ref{eq:SM_recursive_inequality_for_F_nkij})
and using $|C_{i,j}^{n}(\tau)|\leq C_{0}\sqrt{\eta}$ , along with
the condition from Eq. (\ref{eq:SM_support_of_t_for_C_ijn}) that $C_{i,j}^{n}(\tau)\neq0$ only if
$\tau\in[n\eta- r ,(n+1)\eta+ r ]$ , we
obtain
\begin{equation}
\sqrt{F_{i,j}^{n,1}(\tau)}\leq C_{0}\sqrt{\eta}(\tau-\tau_{n,l})\ \mathrm{for}\ \tau_{n,l}\leq\tau\leq\tau_{n,u},\label{eq:SM_F_n1ij_vacuum}
\end{equation}
where $\tau_{n,l}=n\eta- r $, $\tau_{n,u}=(n+1)\eta+ r $.
We next set $k=2$ in Eq. (\ref{eq:SM_recursive_inequality_for_F_nkij})
and substitute Eq. (\ref{eq:SM_F_n1ij_vacuum}), which yields
\begin{equation}
\sqrt{F_{i,j}^{n,2}(\tau)}\leq C_{0}\sqrt{\eta}\int_{\tau_{n,l}}^{\tau}d\tau'2\sqrt{F_{i,j}^{n,1}(\tau')}\leq(C_{0}\sqrt{\eta})^{2}(\tau-\tau_{n,l})^{2}\ \mathrm{for}\ \tau_{n,l}\leq\tau\leq\tau_{n,u}.
\end{equation}
This procedure can be continued, leading to 
\begin{equation}
\sqrt{F_{i,j}^{n,k}(\tau)}\leq(C_{0}\sqrt{\eta})^{k}(\tau-\tau_{n,l})^{k}\ \mathrm{for}\ \tau_{n,l}\leq\tau\leq\tau_{n,u}.
\end{equation}
Alternatively, we can express this upper bound as 
\begin{equation}
F_{i,j}^{n,k}(\tau)\leq(C_{0}\sqrt{\eta})^{2k}(\tau_{n,u}-\tau_{n,l})^{2k}\ \mathrm{for}\ \forall\tau.
\end{equation}

In the above expression, $(\tau_{n,u}-\tau_{n,l})$ is the total interaction
time associated with the driving field $C_{i,j}^{n}(\tau)$. Including
the Trotterization procedure, which involves both forward- and backward-evolution,
is equivalent to extending the interaction time up to a factor $s_{P}$.
Therefore, the final upper bound for $F_{i,j}^{n,k}(\tau)$ after
the Trotterization is given by
\begin{equation}
F_{i,j}^{n,k}(\tau)\leq(C_{0}\sqrt{\eta})^{2k}(s_{P}(\tau_{n,u}-\tau_{n,l}))^{2k}=\left(s_{P}C_{0}\sqrt{\eta}(2 r +\eta)\right)^{2k}
\end{equation}
for $\forall\tau$. By a simple algebraic calculation, we can derive
\begin{equation}
\begin{aligned} & \bra{\psi(\tau)}e^{\theta N_{i,j}^{n}}\ket{\psi(\tau)}\\
&=  \sum_{k=0}^{\infty}(e^{\theta}-1)^{k}\frac{F_{i,j}^{n,k}(\tau)}{k!},\\
&\leq  \exp\left((e^{\theta}-1)(s_{P}C_{0}\sqrt{\eta}(2 r +\eta))^{2}\right),
\end{aligned}
\end{equation}
where the first equality can be verified, for example, by expanding
both operators under the Fock basis. From Lemma \ref{Lemma:bounding_from_expectation_value},
we have
\begin{equation}
\lVert(I-P_{i,j}^{n}(d))\ket{\psi(\tau)}\rVert\leq e^{-\frac{\theta d}{2}+\frac{(e^{\theta}-1)(s_{P}C_{0}\sqrt{\eta}(2 r +\eta))^{2}}{2}}
\end{equation}
for $\forall\tau,\theta>0$. For example, one can choose $\theta=2$
such that 
\begin{equation}
\lVert(I-P_{i,j}^{n}(d))\ket{\psi(\tau)}\rVert\leq C_{\mathrm{trun}}e^{-d},\label{eq:SM_final_result_on_bounding_occupation_from_vacuum}
\end{equation}
where $C_{\mathrm{trun}}=\exp((e^{2}-1)(s_{P}C_{0}\sqrt{\eta}(2 r +\eta))^{2}/2)$
which is a constant. Notice that $\sqrt{\eta} r \leq\mathcal{O}(1)$, $C_\mathrm{trun}$ can thus be bounded by $C_\mathrm{trun}\leq \exp((e^2-1)s_P^2C_0^2)$. We remark that in principle, one can choose $\theta\sim\log d$
such that the probability bound can be improved to $d^{-d}$. However,
for our purpose the exponential decay probability bound is enough.

\subsection{Proof of Lemma \ref{Lemma:EB_quditization_total}\label{subsec:Proof-of-Lemma5}}

We now consider the error in state evolution by truncating the Hamiltonian
$\tilde{H}_{SE}$. More concretely, we denote $U_\mathrm{tro}$ as the Trotterization formula Eq. (\ref{eq:trotterization_expression}) applied to 
\[
\begin{aligned}
\tilde{H}_\mathrm{e}(t)&=\sum_{i\in\mathrm{even}}\left(H_{i,i+1}+(J_{i,i+1}^\dagger\tilde{A}_i(t)+\text{h.c.})\right),\\
\tilde{H}_\mathrm{o}(t)&=\sum_{i\in\mathrm{odd}}\left(H_{i,i+1}+(J_{i,i+1}^\dagger\tilde{A}_i(t)+\text{h.c.})\right),
\end{aligned}
\]
and $\hat{U}_\mathrm{tro}(t,T)$ as
the Trotterization formula Eq. (\ref{eq:trotterization_expression}) applied to
\[
\begin{aligned}
\hat{H}^d_\mathrm{e}(t)&=\sum_{i\in\mathrm{even}}\left(H_{i,i+1}+P_d(J_{i,i+1}^\dagger\tilde{A}_i(t)+\text{h.c.})P_d\right),\\
\hat{H}^d_\mathrm{o}(t)&=\sum_{i\in\mathrm{odd}}\left(H_{i,i+1}+P_d(J_{i,i+1}^\dagger\tilde{A}_i(t)+\text{h.c.})P_d\right),
\end{aligned}
\]
where
$P_d$ is defined in Lemma \ref{Lemma:EB_quditization_total}.
In the following, we analyze the error $\lVert U_\mathrm{tro}(t,T)\ket{\psi(0)}-\hat{U}_\mathrm{tro}(t,T)\ket{\psi(0)}\rVert$,
with $\ket{\psi(0)}$ the initial system-environment state. We first consider purely forward evolution (i.e., neglecting the Trotterization procedure at this stage). The effect of Trotterization, which involves forward- and backward-evolution, is to extend the interaction time by a factor $s_P$.

For the purely forward evolution, we denote
\[
\begin{aligned}
\tilde{U}_{SE}(t,0)=&\mathcal{T}\exp\left(-i\int_0^t\tilde{H}_\mathrm{SE}(\tau)d\tau\right),\\
\hat{U}_{SE}(t,0)=&\mathcal{T}\exp\left(-i\int_0^tP_d\tilde{H}_\mathrm{SE}(\tau)P_dd\tau\right),
\end{aligned}
\]
and 
\[
\ket{\psi(\tau)}=\tilde{U}_{SE}(\tau,0)\ket{\psi(0)}.
\]
Let us introduce $Q_{d}=I-P_{d}$ the projector to the orthogonal subspace of $P_d$. From Ref. \cite{trivedi2022descriptioncomplexitynonmarkovianopen},
we have the following inequality:

\begin{equation}
\begin{aligned} & \lVert \tilde{U}_{SE}(t,0)\ket{\psi(0)}-\hat{U}_{SE}(t,0)\ket{\psi(0)}\rVert \leq  \lVert Q_{d}\ket{\psi(t)}\rVert +\int_{0}^{t}\lVert P_{d}\tilde{H}_{SE}(\tau)Q_{d}Q_{d}\ket{\psi(\tau)}\rVert d\tau.
\end{aligned}\label{eq:SM_inequality_for_the_difference_pure_state}
\end{equation}
By employing Lemma \ref{Lemma4}, we can bound each term separately. For the first
term, we have

\begin{equation}
\begin{aligned} & \lVert Q_{d}\ket{\psi(\tau)}\rVert \leq  \sum_{n,i,j}\lVert (I-P_{i,j}^{n}(d))\ket{\psi(\tau)}\rVert\leq Nj_{\mathrm{max}}\frac{2 r +t}{\eta}C_\mathrm{trun}e^{-d} =  \mathcal{O}\left(\frac{Nt}{\eta} r e^{-d}\right)\text{  for  }\forall\tau,
\end{aligned}
\end{equation}
where $\frac{2 r +t}{\eta}$ is the total number of segments
when we discretize the continuum of the bosonic modes and $Nj_\mathrm{max}\frac{2 r +t}{\eta}$  is the total number of ancillary modes $b_{i,j}^n$ in the environment. $C_\mathrm{trun}$ is defined below Eq. (\ref{eq:SM_final_result_on_bounding_occupation_from_vacuum}).

When bounding the second term, we notice that at each time $\tau$,
there will be at most $Nj_{\mathrm{max}}(\frac{2 r }{\eta}+2)$
modes entering $\tilde{H}_{SE}(\tau)$, which is a consequence from the finite
support of $C_{j}^{n}(\tau)$. For each bosonic mode $b_{i,j}^{n}$,
we have 
\begin{equation}
\lVert P_{d}b_{i,j}^{n}Q_{d}\rVert\leq\sqrt{\lVert P_{d}b_{i,j}^{n}b_{i,j}^{\dagger n}P_{d}\rVert}\leq\sqrt{d+1},
\end{equation}
thus, the upper bound on the second term is 
\begin{equation}
\begin{aligned} & \int_{0}^{t}\lVert P_{d}\tilde{H}_{SE}(\tau)Q_{d}Q_{d}\ket{\psi(\tau)}\rVert d\tau\\
&\leq  \int_{0}^{t}d\tau\sum_{i,j,n}\left(\lVert C_{j}^{n}(\tau)P_{d}b_{i,j}^{n}Q_{d}\rVert +\lVert C_{j}^{n*}(\tau)P_{d}b_{i,j}^{n\dagger }Q_{d}\rVert \right)\lVert Q_{d}\ket{\psi(\tau)}\rVert, \\
&\leq  2C_0\sqrt{\eta}Nj_{\mathrm{max}}\left(\frac{2 r }{\eta}+2\right)\sqrt{d+1}\int_{0}^{t}d\tau\lVert Q_{d}\ket{\psi(\tau)}\rVert.
\end{aligned}
\end{equation}
Combine these two bounds, we have
\begin{equation}
\begin{aligned} & \lVert \tilde{U}_{SE}(t,0)\ket{\psi(0)}-\hat{U}_{SE}(t,0)\ket{\psi(0)}\rVert\\
&\leq  \mathcal{O}\bigg(\frac{Nt}{\eta} r e^{-d}(1+2C_0\sqrt{\eta}tNj_{\mathrm{max}}(\frac{2 r }{\eta}+2)\sqrt{d+1}\bigg),\\
&=  \mathcal{O}\bigg(\sqrt{d}e^{-d}\frac{N^{2}t^{2}}{\eta^{\frac{3}{2}}}r^2\bigg).
\end{aligned}
\label{eq:SM_bound_on_state_distance}
\end{equation}

The above bound scales quadratically with the evolution time $t$. When we consider the Trotterization procedure, which involves $s_P$ stages of forward- and backward-evolution within each time step, the evolution time extends from $t$ to $s_Pt$. Nonetheless, everything else follows the same analysis as in this subsection. Thus, it will not change the scaling behavior in Eq. (\ref{eq:SM_bound_on_state_distance}), which also holds for the quantity $\lVert U_\mathrm{tro}(t,T)\ket{\psi(0)}-\hat{U}_\mathrm{tro}(t,T)\ket{\psi(0)}\rVert$.

Lemma \ref{Lemma:EB_quditization_total} directly follows from Eq. (\ref{eq:SM_bound_on_state_distance})
and the fact that $\lVert\ket{\varphi_{1}}\bra{\varphi_{1}}-\ket{\varphi_{2}}\bra{\varphi_{2}}\rVert_{\mathrm{tr}}\leq2\lVert\ket{\varphi_{1}}-\ket{\varphi_{2}}\rVert$
for any state $\ket{\varphi_{1}},\ket{\varphi_{2}}$, as well as the
fact that the trace distance never increases under the partial trace.

\section{Extension to general Gaussian initial states\label{sec:Extension-to-general}}

In this section, we extend our algorithm to the general Gaussian initial
state $\rho_{E}(0)$, which is characterized by its one-point and
two-point correlators in the temporal region. Similarly, we use $\langle\cdot\rangle$
to denote $\mathrm{Tr}_E((\cdot)\rho_E(0))$. Since we can always absorb
the one-point correlator into the definition of $H_S$, without loss of generality, we can set
\begin{equation}
\langle a_{i,\tau}\rangle=\langle a_{i,\tau}^{\dagger}\rangle=0.
\end{equation}
We further assume that the two-point correlator is stationary and integrable with
a bounded integration, i.e.,
\begin{equation}
\begin{aligned}\langle a_{i,\tau}a_{i',\tau'}\rangle & =\delta_{i,i'}G_i(\tau,\tau'),\\
\langle a_{i,\tau}a_{i',\tau'}^{\dagger}\rangle & =\delta_{i,i'}B_i(\tau,\tau'),\\
\langle a_{i,\tau}^{\dagger}a_{i',\tau'}^{\dagger}\rangle & =\delta_{i,i'}G^{*}_i(\tau',\tau),\\
\langle a_{i,\tau}^{\dagger}a_{i',\tau'}\rangle & =\delta_{i,i'}\left(B_i(\tau',\tau)-\delta(\tau-\tau')\right),
\\
\max_{i,\sigma,\sigma'}|\langle a_{i,\tau}^\sigma a_{i,\tau'}^{\sigma'}\rangle|&\leq D(\tau,\tau'),
\end{aligned}
\label{eq:SM_notation_in_two_point_CF}
\end{equation}
with 
\begin{equation}
\max_{\tau}\int_{-\infty}^{\infty}D(\tau,\tau')d\tau'=\max_{\tau'}\int_{-\infty}^{\infty}D(\tau,\tau')d\tau\leq\mathcal{D}.\label{eq:SM_Bounding_on_the_two_point_correlator}
\end{equation}
Here $D(\tau,\tau')$ is a distribution which may include $\delta(\tau-\tau')$ and $\mathcal{D}$ is a $\mathcal{O}(1)$ constant. With a little abuse of notations, we define $f\leq g$ for two distributions $f,g$ in the sense that $g-f$ is a positive distribution. Our algorithm directly works through
if the assumption Eq. (\ref{eq:SM_Bounding_on_the_two_point_correlator})
is satisfied and the corresponding environmental initial state (after
the discretization and truncation steps discussed in the main text)
can be prepared with $\mathcal{O}(Nt)$ gates. Before proving this
claim, we first introduce the memory Kernels associated with a general
Gaussian initial state $\rho_{E}(0)$.

Using the notation in Eq. (\ref{eq:SM_notation_in_two_point_CF}),
the memory Kernels can be expressed as 
\begin{equation}
\begin{aligned}K_i^{(-,+)}(s,s') & =\langle A_{i}(s)A_{i}^{\dagger}(s')\rangle=\int_{-\infty}^{\infty}d\tau\int_{-\infty}^{\infty}d\tau'v_i(s-\tau)B_i(\tau,\tau')v_i^{*}(s'-\tau');\\
K^{(-,-)}_i(s,s') & =\langle A_{i}(s)A_{i}(s')\rangle=\int_{-\infty}^{\infty}d\tau\int_{-\infty}^{\infty}d\tau'v_i(s-\tau)G_i(\tau,\tau')v_i(s'-\tau');\\
K_i^{(+,+)}(s,s') & =\langle A_{i}^{\dagger}(s)A_{i}^{\dagger}(s')\rangle=\int_{-\infty}^{\infty}d\tau\int_{-\infty}^{\infty}d\tau'v_i^{*}(s-\tau)G_i^{*}(\tau',\tau)v_i^{*}(t'-\tau');\\
K_i^{(+,-)}(s,s') & =\langle A_{i}^{\dagger}(s)A_{i}(s')\rangle=\int_{-\infty}^{\infty}d\tau\int_{-\infty}^{\infty}d\tau'v_i^{*}(s-\tau)(B(\tau',\tau)-\delta(\tau-\tau'))v_i(s'-\tau'),\\
 & =\int_{-\infty}^{\infty}d\tau\int_{-\infty}^{\infty}d\tau'v_i^{*}(s-\tau)B(\tau',\tau)v_i(s'-\tau')-\int_{-\infty}^{\infty}d\tau v_i^{*}(t-\tau)v_i(t'-\tau).
\end{aligned}
\end{equation}
Those memory Kernels as well as their integration
is bounded. To see this, we take $K_i^{(-,+)}$ as an example.
We have
\begin{equation}
K_i^{(-,+)}(s,s')\leq C_0^{2}\int_{s- r }^{s+ r }d\tau\int_{s'- r }^{s'+ r }d\tau'B(\tau,\tau')\leq2C_0^{2} r \mathcal{D},
\end{equation}
where in the first inequality we use the assumption that the temporal
coupling function $v_i(t)$ is bounded and finitely supported, and in
the second inequality we use Eq. (\ref{eq:SM_Bounding_on_the_two_point_correlator}).
Also, we can derive
\begin{equation}
\begin{aligned}\int_{-\infty}^{\infty}ds'|K_i^{(-,+)}(s,s')| & \leq\int_{-\infty}^{\infty}d\tau\int_{-\infty}^{\infty}d\tau'|v_i(s-\tau)B(\tau,\tau')|\int_{-\infty}^{\infty}ds'|v_i^{*}(s'-\tau')|,\\
 & \leq2C_0 r \int_{-\infty}^{\infty}|v_i(s-\tau)|d\tau\int_{-\infty}^{\infty}d\tau'|B(\tau,\tau')|,\\
 & \leq2C_0 r \mathcal{D}\int_{-\infty}^{\infty}|v_i(s-\tau)|d\tau,\\
 & \leq\left(2C_0 r \right)^{2}\mathcal{D}.
\end{aligned}
\end{equation}
Here in the second and forth inequality, we use $\int_{-\infty}^{\infty}|v_i(t)|dt\leq2C_0 r $,
and in the third inequality use Eq. (\ref{eq:SM_Bounding_on_the_two_point_correlator}).
The upper bound for the other three memory Kernels can be proved in
a similar way. Now, in order to show that our algorithm in the main
text applies, we only need to prove the Lemma \ref{lemma:EB_Kernel} and Lemma \ref{Lemma4} with $\rho_{E}(0)$
under the assumption Eq. (\ref{eq:SM_Bounding_on_the_two_point_correlator}).

\subsection{Proof of Lemma \ref{lemma:EB_Kernel} for general initial states}

The new sets of bosonic modes $b_{i,j}^{n}$ is defined as 
\begin{equation}
b_{i,j}^{n}=\int_{-\infty}^{\infty}P_{j}^{n}a_{i,\tau}d\tau,
\end{equation}
with $P_{j}^{n}$ only supports on the interval $[n\eta,(n+1)\eta)$.
According to Eq. (\ref{eq:SM_notation_in_two_point_CF}), we have
\begin{equation}
\begin{aligned}\langle b_{i,j}^{n}\rangle & =\langle b_{i,j}^{n\dagger}\rangle=0;\\
\langle b_{i,j}^{n}b_{i',j'}^{n'}\rangle & =\delta_{i,i'}\langle\int_{-\infty}^{\infty}P_{j}^{n}a_{i,\tau}d\tau\int_{-\infty}^{\infty}P_{j'}^{n'}a_{i',\tau'}d\tau\rangle\\
 & =\delta_{i,i'}\int_{-\infty}^{\infty}d\tau\int_{-\infty}^{\infty}d\tau'P_{j}^{n}(\tau)P_{j'}^{n'}(\tau')G_i(\tau,\tau');\\
\langle b_{i,j}^{n}b_{i',j'}^{\dagger n'}\rangle & =\delta_{i,i'}\int_{-\infty}^{\infty}d\tau\int_{-\infty}^{\infty}d\tau'P_{j}^{n}(\tau)P_{j'}^{n'}(\tau')B_i(\tau,\tau').
\end{aligned}
\end{equation}
Thus, the new memory Kernel can be expressed as 
\begin{equation}
\begin{aligned}\tilde{K}_i^{(-,+)}(s,s') & =\sum_{j,n,j',n'}\langle C_{i,j}^{n}(s)b_{i,j}^{n}C_{i,j'}^{*n'}(s')b_{i,j'}^{\dagger n'}\rangle,\\
 & =\int_{-\infty}^{\infty}d\tau\int_{-\infty}^{\infty}d\tau'\sum_{j,n,j',n'}C_{i,j}^{n}(s)C_{i,j'}^{*n'}(s')P_{j}^{n}(\tau)P_{j'}^{n'}(\tau')B_i(\tau,\tau');\\
\tilde{K}_i^{(-,-)}(s,s') & =\int_{-\infty}^{\infty}d\tau\int_{-\infty}^{\infty}d\tau'\sum_{j,n,j',n'}C_{i,j}^{n}(s)C_{i,j'}^{n'}(s')P_{j}^{n}(\tau)P_{j'}^{n'}(\tau')G_i(\tau,\tau');\\
\tilde{K}_i^{(+,+)}(s,s') & =\int_{-\infty}^{\infty}d\tau\int_{-\infty}^{\infty}d\tau'\sum_{j,n,j',n'}C_{i,j}^{*n}(s)C_{i,j'}^{*n'}(s')P_{j}^{n}(\tau)P_{j'}^{n'}(\tau')G_i^{*}(\tau',\tau);\\
\tilde{K}_i^{(+,-)}(s,s') & =\int_{-\infty}^{\infty}d\tau\int_{-\infty}^{\infty}d\tau'\sum_{j,n,j',n'}C_{i,j}^{*n}(s)C_{i,j'}^{n'}(s')P_{j}^{n}(\tau)P_{j'}^{n'}(\tau')\left(B_i(\tau',\tau)-\delta(\tau-\tau')\right).
\end{aligned}
\end{equation}
In order to show that the Lemma \ref{lemma:EB_Kernel} still holds,
we take $\tilde{K}_i^{(-,+)}(s,s')$ as an example. The other
three memory Kernels can be proved in the same way. We introduce
\begin{equation}
R_{s}(\tau)=\sum_{j,n}C_{i,j}^{n}(s)P_{j}^{n}(\tau)-v_i(s-\tau).
\end{equation}
 Since $\sum_{j,n}C_{i,j}^{n}(s)P_{j}^{n}(\tau)$ is a piecewise $j_{\mathrm{max}}^\text{th}$
order polynomial interpolation of $v_i(s-\tau)$, we have 
\begin{equation}
|R_{s}(\tau)|\leq\frac{D_{j_\mathrm{max}+1}}{(j_{\mathrm{max}}+1)!}\eta^{j_{\mathrm{max}}+1}\leq \Gamma\eta^{j_{\mathrm{max}}+1},\label{eq:SM_upper_bound_on_R}
\end{equation}
with $\Gamma$ a constant.
What's more, according to Eq. (\ref{eq:SM_region_of_n_fixed_t}),
for a fixed $s$, $C_{i,j}^{n}(s)$ is non-zero only if $\frac{s- r }{\eta}-1\leq n\leq\frac{s+ r }{\eta}$.
Therefore, we have 
\begin{equation}
R_{s}(\tau)=0\ \mathrm{if}\ |s-\tau|> r +\eta,\label{eq:SM_support_of_R}
\end{equation}
Consequently
\begin{equation}
\begin{aligned} & |\tilde{K}_i^{(-,+)}(s,s')-K_i^{(-,+)}(s,s')|\\
&=  |\int_{-\infty}^{\infty}d\tau\int_{-\infty}^{\infty}d\tau'B_i(\tau,\tau')\{\sum_{j,n}C_{i,j}^{n}(s)P_{j}^{n}(\tau)\sum_{j',n'}C_{i,j'}^{n'*}(s')P_{j'}^{n'}(\tau')-v_i(s-\tau)v^{*}_i(s'-\tau')\}|,\\
&\leq  \int_{-\infty}^{\infty}d\tau\int_{-\infty}^{\infty}d\tau'|B_i(\tau,\tau')|\left\{ \left(v_i(s-\tau)+R_{s}(\tau)\right)\left(v^{*}_i(s'-\tau')+R_{s'}^{*}(\tau')\right)-v_i(s-\tau)v^{*}_i(s'-\tau')\right\}, \\
&\leq  \int_{-\infty}^{\infty}d\tau\int_{-\infty}^{\infty}d\tau'|B_i(\tau,\tau')|\left\{ |R_{s}(\tau)v^{*}_i(s'-\tau')|+|R_{s'}^{*}(\tau')v_i(s-\tau)|+|R_{s}(\tau)R_{s'}^{*}(\tau')|\right\}, \\
&\leq  \int_{s- r -\eta}^{s+ r +\eta}d\tau\int_{s'- r -\eta}^{s'+ r +\eta}d\tau'|B_i(\tau,\tau')|(2C_0\Gamma\eta^{j_{\mathrm{max}}+1}+\Gamma^2\eta^{2j_{\mathrm{max}}+2}),\\
&\leq  \int_{s- r -\eta}^{s+ r +\eta}d\tau\mathcal{D}(2C_0\Gamma\eta^{j_{\mathrm{max}}+1}+\Gamma^2\eta^{2j_{\mathrm{max}}+2}),\\
&=  \mathcal{O}\left(\eta^{j_{\mathrm{max}}+1} r \right).
\label{eq:SM_bounding_on_general_kernel}
\end{aligned}
\end{equation}
Here in the third inequality, we use Eqs. (\ref{eq:SM_support_of_R}),
(\ref{eq:SM_upper_bound_on_R}), and the assumption that $v_i(t)$ is
finitely supported and bounded. Similarly,
\begin{equation}
\begin{aligned} & \int_{0}^{t}ds\int_{0}^{t}ds'|\tilde{K}_i^{(-,+)}(s,s')-K_i^{(-,+)}(s,s')|\\
&\leq  \int_{0}^{t}ds\int_{0}^{t}ds'\int_{-\infty}^{\infty}d\tau\int_{-\infty}^{\infty}d\tau'|B_i(\tau,\tau')|\left\{ |R_{s}(\tau)v_i^{*}(s'-\tau')|+|R_{s'}^{*}(\tau')v_i(s-\tau)|+|R_{s}(\tau)R_{s'}^{*}(\tau')|\right\},\\
&\leq  \int_{0}^{t}ds'\int_{-\infty}^{\infty}d\tau'|v^{*}_i(s'-\tau')|\int_{-\infty}^{\infty}d\tau|B_i(\tau,\tau')|\int_{0}^{t}ds|R_{s}(\tau)|+\\
 & \int_{0}^{t}ds\int_{-\infty}^{\infty}d\tau|v_i(s-\tau)|\int_{-\infty}^{\infty}d\tau'|B_i(\tau,\tau')|\int_{0}^{t}ds'|R_{s'}^{*}(\tau')|+\\
 & \int_{0}^{t}ds\int_{-\infty}^{\infty}d\tau|R_{s}(\tau)|\int_{-\infty}^{\infty}d\tau'|B_i(\tau,\tau')|\int_{0}^{T}ds'|R_{s'}^{*}(\tau')|,\\
&\leq  2\left( r +\eta\right)\Gamma\eta^{j_{\mathrm{max}}+1}\mathcal{D}\left(2C_0 r \right)t+
 \left( r +\eta\right)^{2}\Gamma^{2}\eta^{2j_{\mathrm{max}}+2}\mathcal{D}t,\\
&=  \mathcal{O}\left(t\eta^{j_{\mathrm{max}}+1}r^2\right).
\end{aligned}\label{eq:SM_bounding_on_general_kernel_integrals}
\end{equation}

Last, we show that $\tilde{K}_{i}^{(-,+)}(s,s')$ has a finite $L^{1}$ norm bound. This follows from 
\begin{equation}
\begin{aligned}
\int_{-\infty}^{\infty}|\tilde{K}_{i}(s,s')|ds' & =\int_{-\infty}^{\infty}ds'|\int_{-\infty}^{\infty}d\tau\int_{-\infty}^{\infty}d\tau'\sum_{j,n,j',n'}C_{i,j}^{n}(s)C_{i,j'}^{n'*}(s')P_{j}^{n}(\tau)P_{j'}^{n'}(\tau')B_{i}(\tau,\tau')|,\\
 & \leq\int_{-\infty}^{\infty}d\tau|v_{i}(s-\tau)+R_{s}(\tau)|\int_{-\infty}^{\infty}d\tau'|B_{i}(\tau,\tau')|\int_{-\infty}^{\infty}ds'|v_{i}^{*}(s'-\tau')+R_{s'}^{*}(\tau')|,\\
 & \leq(2 r +2\eta)(C_{0}+\Gamma\eta^{j_{\mathrm{max}}+1})\int_{-\infty}^{\infty}d\tau|v_{i}(s-\tau)+R_{s}(\tau)|\int_{-\infty}^{\infty}d\tau'|B_{i}(\tau,\tau')|,\\
 & \leq\mathcal{D}(2 r +2\eta)(C_{0}+\Gamma\eta^{j_{\mathrm{max}}+1})\int_{-\infty}^{\infty}d\tau|v_{i}(s-\tau)+R_{s}(\tau)|,\\
 & \leq4\mathcal{D}( r +\eta)^{2}(C_{0}+\Gamma\eta^{j_{\mathrm{max}}+1})^{2},
\end{aligned}
\end{equation}
where in the second and last inequality, we used Eqs. (\ref{eq:SM_upper_bound_on_R}) and (\ref{eq:SM_support_of_R}), and in the third inequality we used Eq. (\ref{eq:SM_Bounding_on_the_two_point_correlator}).
Thus, we have proved the Lemma \ref{lemma:EB_Kernel}.

\subsection{Proof of Lemma \ref{Lemma4} for general initial states}

Since each dynamics of each bosonic mode $b_{i,j}^{n}$ is decoupled,
we can consider the fixed $n,i,j$ in this subsection. Let's first
bound $2k$-point correlators for $b_{i,j}^{n},b_{i,j}^{n\dagger}$
at initial time $t=0$. Namely, we want to bound
\begin{equation}
\langle O_{1}^{\sigma_1}O_{2}^{\sigma_2}\cdots O_{2k}^{\sigma_{2k}}\rangle
\end{equation}
with $O_{i}^{\sigma_i}=b_{i,j}^{n,\sigma_i}(0)$.
Since the orthonormal polynomial satisfies
\begin{equation}
|P_{j}^{n}(\tau)|\leq\frac{g}{\sqrt{\eta}},\label{eq:SM_maximum_bounding_on_orthonormal_poly}
\end{equation}
with $g$ only depends on the maximum degree $j_{\mathrm{max}}$,
we obtain
\begin{equation}
\begin{aligned} & \langle O_{1}^{\sigma_1}O_{2}^{\sigma_2}\cdots O_{2k}^{\sigma_{2k}}\rangle\\
&= \left\langle \prod_{l=1}^{2k}\int_{n\eta}^{(n+1)\eta}d\tau_{l}P_{j}^{n}(\tau_{l})a_{i,\tau_{l}}^{\sigma_l}\right\rangle, \\
&\leq  (2k-1)!!\prod_{l=1}^{k}\int_{n\eta}^{(n+1)\eta}d\tau_{(2l-1)}P_{j}^{n}(\tau_{2l-1})\int_{n\eta}^{(n+1)\eta}d\tau_{2l}P_{j}^{n}(\tau_{2l})D(\tau_{2l-1},\tau_{2l}),\\
&=  (2k-1)!!\left(\int_{n\eta}^{(n+1)\eta}d\tau_{1}P_{j}^{n}(\tau_{1})\int_{n\eta}^{(n+1)\eta}d\tau_{2}P_{j}^{n}(\tau_{2})D(\tau_{1},\tau_{2})\right)^{k},\\
&\leq  (2k)!!g^{2k}\left(\frac{1}{\eta}\int_{n\eta}^{(n+1)\eta}d\tau_{1}\int_{n\eta}^{(n+1)\eta}d\tau_{2}D(\tau_{1},\tau_{2})\right)^{k},\\
&\leq  2^{k}k!g^{2k}\left(\frac{1}{\eta}\int_{n\eta}^{(n+1)\eta}d\tau_{1}\mathcal{D}\right)^{k},\\
&\leq  k!(2g^{2}\mathcal{D})^{k}.
\end{aligned}
\label{eq:SM_k_correlator_bounding_on_initial}
\end{equation}
Here, in the first inequality, we apply Wick's theorem to decompose
the $2k$-point correlator into a sum over all possible pairwise contractions
of 2-point correlators. In the second inequality, we employ Eq. (\ref{eq:SM_maximum_bounding_on_orthonormal_poly})
and in the third inequality we use Eq. (\ref{eq:SM_Bounding_on_the_two_point_correlator}).

Specifically, the above inequality indicates that 
\begin{equation}
F_{i,j}^{n,k}(0)\leq k!(2g^{2}\mathcal{D})^{k}\ \mathrm{for}\ k\geq1,\quad F_{i,j}^{n,0}(\tau)=1,
\end{equation}
which serves as the initial condition for the recursive inequality
Eq. (\ref{eq:SM_recursive_inequality_for_F_nkij}). We remark that the Cauchy-Schwarz inequality we used in Eq. (\ref{eq:SM_CS_inequality_used_in_prove}) also holds for the mixed state $\rho_{SE}(\tau)$, where $\rho_{SE}(\tau)$ is the system-environment density matrix at time $\tau$. Furthermore, the last inequality in Eq. (\ref{eq:SM_CS_inequality_used_in_prove}) can be generalized to the mixed state as 
\[
\begin{aligned} & \mathrm{Tr}\left(\sqrt{\underline{(N_{i,j}^{n})^{k-1}}}J_{i,i+1}^{\dagger}J_{i,i+1}\sqrt{\underline{(N_{i,j}^{n})^{k-1}}}\rho_{SE}(\tau)\right)\\
&=  \mathrm{Tr}\left(J_{i,i+1}^{\dagger}J_{i,i+1}\sqrt{\underline{(N_{i,j}^{n})^{k-1}}}\rho_{SE}(\tau)\sqrt{\underline{(N_{i,j}^{n})^{k-1}}}\right),\\
&\leq  \lVert J_{i,i+1}^{\dagger}J_{i,i+1}\rVert\mathrm{Tr}\left(\sqrt{\underline{(N_{i,j}^{n})^{k-1}}}\rho_{SE}(\tau)\sqrt{\underline{(N_{i,j}^{n})^{k-1}}}\right),\\
&\leq  F_{i,j}^{n,k-1}(\tau).
\end{aligned}
\]
Therefore, Eq. (\ref{eq:SM_recursive_inequality_for_F_nkij}) also holds for the mixed state.
To solve this
recursive inequalities, we also first assume purely forward evolution
(i.e., neglecting the Trotterization procedure at this stage), and
later account for the effect of Trotterization. Setting $k=1$ in
Eq. (\ref{eq:SM_recursive_inequality_for_F_nkij}) leads to

\begin{equation}
\sqrt{F_{i,j}^{n,1}(\tau)}\leq\sqrt{1!(2g^{2}\mathcal{D})^{1}}+C_{0}\sqrt{\eta}(\tau-\tau_{n,l})\ \mathrm{for}\ \tau_{n,l}\leq\tau\leq\tau_{n,u},
\end{equation}
where $\tau_{n,l},\tau_{n,u}$ is defined below Eq. (\ref{eq:SM_F_n1ij_vacuum}).
We next set $k=2$, which yields
\begin{equation}
\begin{aligned}\sqrt{F_{i,j}^{n,2}(\tau)} & \leq\sqrt{2!(2g^{2}\mathcal{D})^{2}}+2C_{0}\sqrt{\eta}\int_{\tau_{n,l}}^{\tau}d\tau'\sqrt{F_{i,j}^{n,1}(\tau')},\\
 & \leq\sqrt{2!}(2g^{2}\mathcal{D})+\sqrt{1!}2C_{0}\sqrt{\eta}(\tau-\tau_{n,l})\sqrt{2g^{2}\mathcal{D}}+(C_{0}\sqrt{\eta})^{2}(\tau-\tau_{n,l})^{2},\\
 & \leq\sqrt{2!}(C_{0}\sqrt{\eta}(\tau-\tau_{n,l})+\sqrt{2g^{2}\mathcal{D}})^{2},
\end{aligned}
\end{equation}
for $\tau_{n,l}\leq\tau\leq\tau_{n,u}$. Inspired by the above formula,
we assume that 
\begin{equation}
\sqrt{F_{i,j}^{n,k}(\tau)}\leq\sqrt{k!}(C_{0}\sqrt{\eta}(\tau-\tau_{n,l})+\sqrt{2g^{2}\mathcal{D}})^{k}\ \mathrm{for}\ \tau_{n,l}\leq\tau\leq\tau_{n,u},\label{eq:SM_general_initial_bound_on_Fijnk}
\end{equation}
which will be proved by induction.
Indeed, substituting Eq. (\ref{eq:SM_general_initial_bound_on_Fijnk}) into Eq. (\ref{eq:SM_recursive_inequality_for_F_nkij}) gives an upper
bound on $\sqrt{F_{i,j}^{n,k+1}(\tau)}$ as 
\begin{equation}
\begin{aligned}\sqrt{F_{i,j}^{n,k+1}(\tau)} & \leq\sqrt{(k+1)!}(2g^{2}\mathcal{D})^{\frac{k+1}{2}}+(k+1)C_{0}\sqrt{\eta}\int_{\tau_{n,l}}^{\tau}\sqrt{F_{i,j}^{n,k}(\tau')}d\tau',\\
 & \leq\sqrt{(k+1)!}(2g^{2}\mathcal{D})^{\frac{k+1}{2}}+\sqrt{k!}\int_{\tau_{n,l}}^{\tau}(k+1)C_{0}\sqrt{\eta}d\tau'(C_{0}\sqrt{\eta}(\tau'-\tau_{n,l})+\sqrt{2g^{2}\mathcal{D}})^{k},\\
 & \leq\sqrt{(k+1)!}(2g^{2}\mathcal{D})^{\frac{k+1}{2}}+\sqrt{k!}\left((C_{0}\sqrt{\eta}(\tau-\tau_{n,l})+\sqrt{2g^{2}\mathcal{D}})^{k+1}-(2g^{2}\mathcal{D})^{\frac{k+1}{2}}\right),\\
 & \leq\sqrt{(k+1)!}(2g^{2}\mathcal{D})^{\frac{k+1}{2}}+\sqrt{(k+1)!}\left((C_{0}\sqrt{\eta}(\tau-\tau_{n,l})+\sqrt{2g^{2}\mathcal{D}})^{k+1}-(2g^{2}\mathcal{D})^{\frac{k+1}{2}}\right)\\
 & =\sqrt{(k+1)!}(C_{0}\sqrt{\eta}(\tau-\tau_{n,l})+\sqrt{2g^{2}\mathcal{D}})^{k+1}
\end{aligned}
\end{equation}
for $\tau_{n,l}\leq\tau\leq\tau_{n,u}$. According to the Induction
principle, Eq. (\ref{eq:SM_general_initial_bound_on_Fijnk}) holds
for all $k$. Alternatively, we can express this upper bound as 
\begin{equation}
F_{i,j}^{n,k}(\tau)\leq k!(C_{0}\sqrt{\eta}(\tau_{n,u}-\tau_{n,l})+\sqrt{2g^{2}\mathcal{D}})^{2k}\ \mathrm{for}\ \forall\tau.
\end{equation}

In the above expression, $(\tau_{n,u}-\tau_{n,l})$ is the total interaction
time associated with the driving field $C_{i,j}^{n}(\tau)$. Including
the Trotterization procedure, which involves both forward- and backward-evolution,
is equivalent to extending the interaction time up to a factor $s_{P}$.
Therefore, the final upper bound for $F_{i,j}^{n,k}(\tau)$ after
the Trotterization is given by
\begin{equation}
\begin{aligned}F_{i,j}^{n,k}(\tau) & \leq k!(C_{0}s_{P}\sqrt{\eta}(\tau_{n,u}-\tau_{n,l})+\sqrt{2g^{2}\mathcal{D}})^{2k}\ \mathrm{for}\ \forall\tau,\\
 & \leq k!\left(C_{0}s_{P}\sqrt{\eta}(2 r +\eta)+\sqrt{2g^{2}\mathcal{D}}\right)^{2k}\ \mathrm{for}\ \forall\tau.
\end{aligned}
\end{equation}
For convenience, we introduce the constant 
\begin{equation}
D_{\mathrm{trun}}=(C_{0}s_{P}\sqrt{\eta}(2 r +\eta)+\sqrt{2g^{2}\mathcal{D}})^{2}.
\end{equation}
We can therefore derive that for $\forall\tau$,
\begin{equation}
\begin{aligned} & \mathrm{Tr}\left(e^{\theta N_{i,j}^{n}}\rho_{SE}(\tau)\right)\\
&=  \sum_{k=0}^{\infty}(e^{\theta}-1)^{k}\frac{F_{i,j}^{n,k}(\tau)}{k!},\\
&\leq  \sum_{k=0}^{\infty}\left((e^{\theta}-1)D_{\mathrm{trum}}\right)^{k}.
\end{aligned}
\end{equation}
We can choose $e^{\theta_{\mathrm{trun}}}-1=1/(2D_{\mathrm{trun}})$
or $\theta_{\mathrm{trun}}=\ln(1+1/(2D_{\mathrm{trun}}))$ such that
\begin{equation}
\mathrm{Tr}\left(e^{\theta_{\mathrm{trun}}N_{i,j}^{n}}\rho_{SE}(\tau)\right)\leq2.
\end{equation}
Similar as Lemma \ref{Lemma:bounding_from_expectation_value}, we can obtain
the probability bound on high occupation state as 
\begin{equation}
e^{d\theta_\mathrm{trun}}\mathrm{Tr}\left((I-P_{i,j}^{n}(d))\rho_{SE}(\tau)\right)\rVert\leq\mathrm{Tr}\left(e^{\theta_\mathrm{trun} N_{i,j}^n}\rho_{SE}(\tau)\right)\leq 2,
\end{equation}
such that
\begin{equation}
\mathrm{Tr((I-P^n_{i,j}(d)\rho_{SE}(\tau))}\leq2e^{-d\theta_\mathrm{trun}},\label{eq:SM_bounding_on_high_occupation_number_for_mixed_state}
\end{equation}
with $\theta_{\mathrm{trun}}$ a constant. By 
using $\sqrt{\eta} r \leq\mathcal{O}(1)$, we can upper bound $D_\mathrm{trun}$ by $(C_0s_P+\sqrt{2g^2\mathcal{D}})^2$. As a result, $\theta_\mathrm{trun}$ is lower bounded by
\[
\theta_\mathrm{trun}\geq\ln\left(1+\frac{1}{2(C_0s_P+\sqrt{2g^2\mathcal{D}})^2}\right).
\]
Therefore, for the general
initial state we also obtain an exponential decay occupation probability.
\subsection{Proof of Lemma \ref{Lemma:EB_quditization_total} for general initial states}
Here we briefly mention how to prove the Lemma \ref{Lemma:EB_quditization_total} for general initial states. Most of the proof strategies directly parallel those in Subsec. \ref{subsec:Proof-of-Lemma5}. In this subsection, we just outline the difference. 

By denoting $\rho_{SE}(\tau)=\tilde{U}_{SE}(\tau,0)\rho_{SE}(0)\tilde{U}_{SE}(0,\tau)$,  Eq. (\ref{eq:SM_bounding_on_high_occupation_number_for_mixed_state}) leads to 
\[
\mathrm{Tr}\left(Q_d\rho_{SE}(\tau)\right)=\mathrm{Tr}\left(Q_d\rho_{SE}(\tau)Q_d\right)\leq \mathcal{O}\left(\frac{Nt}{\eta}re^{-d\theta_\mathrm{trun}}\right),
\]
therefore, by Holder inequality, we can upper bound $\lVert Q_d\rho_{SE}(\tau)\rVert_\mathrm{tr}$ as
\[
\begin{aligned}\lVert \rho_{SE}(\tau)Q_{d}\rVert_{\mathrm{tr}}=\lVert Q_{d}\rho_{SE}(\tau)\rVert_{\mathrm{tr}} & =\lVert(Q_{d}\rho_{SE}^{\frac{1}{2}}(\tau))\rho_{SE}^{\frac{1}{2}}(\tau)\rVert_{\mathrm{tr}},\\
 & \leq\lVert Q_{d}\rho_{SE}^{\frac{1}{2}}(\tau)\rVert_{2}\lVert\rho_{SE}^{\frac{1}{2}}(\tau)\rVert_{2},\\
 & \leq\sqrt{\mathrm{Tr}Q_{d}\rho_{SE}(\tau)Q_{d}},\\
 & \leq\mathcal{O}\left(\sqrt{\frac{Nt}{\eta}r}e^{-\frac{d\theta_{\mathrm{trun}}}{2}}\right),
\end{aligned}
\]
where in the first inequality, we used the Holder inequality and $\lVert \cdot \rVert_2$ is the $2$-norm (or Hilbert-Schmidt norm), and in the second inequality we used the fact that $\lVert \rho_{SE}^{1/2}(\tau)\rVert_2=1$. 

Now, for the mixed  state, we have the inequality
\begin{equation}
\begin{aligned} & \left\lVert\tilde{U}_{SE}(t,0)\rho_{SE}(0)\tilde{U}_{SE}^{\dagger}(t,0)-\hat{U}_{SE}(t,0)\rho_{SE}(0)\hat{U}_{SE}^{\dagger}(t,0)\right\rVert_{\mathrm{tr}},\\
 & \leq\lVert P_{d}\rho_{SE}(t)Q_{d}\rVert_{\mathrm{tr}}+\lVert Q_{d}\rho_{SE}(t)Q_{d}\rVert_{\mathrm{tr}}+\lVert Q_{d}\rho_{SE}(t)P_{d}\rVert_{\mathrm{tr}}+\\
 & \int_{0}^{t}\left\lVert (P_{d}\tilde{H}_{SE}(\tau)Q_{d})(Q_{d}\rho_{SE}(\tau)P_{d})\right\rVert _{\mathrm{tr}}d\tau+\int_{0}^{t}\left\lVert (P_{d}\rho_{SE}(\tau)Q_{d})(Q_{d}\tilde{H}_{SE}(\tau)P_{d})\right\rVert _{\mathrm{tr}}d\tau,\\
 & \leq\mathcal{O}\left(\frac{(Ntr)^{\frac{3}{2}}}{\eta}\sqrt{d}e^{-\frac{d\theta_{\mathrm{trun}}}{2}}\right),
\end{aligned}
\end{equation}
where in the last inequality, we used the fact that 
\[\lVert P_d\rho_{SE}(\tau)Q_d\rVert_\mathrm{tr}\leq \lVert P_d\rVert\lVert\rho_{SE}(\tau)Q_d\rVert_\mathrm{tr}\leq\lVert\rho_{SE}(\tau)Q_d\rVert_\mathrm{tr},\] and
\[
\left\lVert (P_{d}\rho_{SE}(\tau)Q_{d})(Q_{d}\tilde{H}_{SE}(\tau)P_{d})\right\rVert _{\mathrm{tr}}\leq \left\lVert P_{d}\rho_{SE}(\tau)Q_{d}\right\rVert_\mathrm{tr}\left\lVert Q_{d}\tilde{H}_{SE}(\tau)P_{d}\right\rVert.
\]
Therefore, we can derive a similar exponential decay error bound with respect to $d$ as in Lemma \ref{Lemma:EB_quditization_total}.
\section{Error in simulating local observables\label{sec:noise_robustness}}
In this section, we examine the error in simulating local observables
which resides on the region $X$ whose diameter is $\mathcal{O}(1)$. Denoting the simulating channel
by $\mathcal{E}$, we will show that if we only require the local
error 
\begin{equation}
\begin{aligned}|O_{\mathrm{tar}}-O_{\mathrm{sim}}| & :=\max_{\lVert\rho_S(0)\rVert_{\mathrm{tr}}=1}|\mathrm{Tr}\left(O_{X}U_{SE}(t,0)\rho_S(0)\otimes \rho_E(0)U_{SE}^{\dagger}(t,0)\right)-\mathrm{Tr}\left(O_{X}\mathcal{E}(\rho_S(0))\right)|\end{aligned}
\label{eq:definition_local_error}
\end{equation}
is below $\delta_\mathrm{loc}$ for any local operator $\lVert O_{X}\rVert\leq1$,
the circuit depth and the number of ancillary qubits per site can
be chosen to be independent of the total system size $N$. As a consequence,
according to Ref. \cite{RahulUnpunlishedNoiseRobust}, our simulation
procedure is noise robustness.

Our proof strategy mainly follows \cite{RahulUnpunlishedNoiseRobust,trivedi2024liebrobinsonboundopenquantum}.
We first combine the Lieb-Robinson bound proved in Ref. \cite{trivedi2024liebrobinsonboundopenquantum}
and the technics in \cite{RahulUnpunlishedNoiseRobust}
to examine the local error of the Trotterization formula in Lemma \ref{lemma1}. This directly gives us the expression of local error and noise
robustness in simulating non-dissipative non-Markovian dynamics. We
then perform the local error analysis for the discretization and quditization
procedure in simulating dissipative non-Markovian dynamics.

Throughout this section, we consider a local non-Markovian dynamics
in $D$-dimensional lattice. Following the same notations in Sec.
\ref{sec:Detailed-Proof-ofLemma1}, we express the time dependent
Hamiltonian in doubled Hilbert space as 
\begin{equation}
\mathbb{H}_{SE}(t)=\sum_{\nu\subseteq\mathbb{R}^{D}}\mathbb{H}_{\nu}+\sum_{\nu\subseteq\mathbb{R}^{D}}\sum_{\alpha=1}^{4}\mathbb{J}_{\nu}^{\alpha}\mathbb{A}_{\nu}^{\alpha}(t),\label{eq:Hamiltonian_in_higher_dimension_interaction}
\end{equation}
where $\mathbb{H}_{\nu},\mathbb{J}_{\nu}^{\alpha}$ are system superoperators
and $\nu$ denotes the support of the superoperator. Without loss
of generality, we assume that each $\nu$ only covers two nearest
neighbor sites. $\mathbb{A}_{\nu}^{\alpha}(t)$ follows the same memory
kernel function as in Eq. (\ref{eq:doubled_space_Kernel_function}).
We will use the subscript to denote the support of the operator in
this section. For example, if we choose $\Lambda\subseteq\mathbb{R}^{D}$,
$\mathbb{H}_{SE,\Lambda}(t)=\sum_{\nu\subseteq\Lambda}\mathbb{H}_{\nu}+\sum_{\nu\subseteq\Lambda}\sum_{\alpha=1}^{4}\mathbb{J}_{\nu}^{\alpha}\mathbb{A}_{\nu}^{\alpha}(t)$
and 
\begin{equation}
\mathbb{U}_{\Lambda}(t_{2},t_{1})=\mathcal{T}\exp\left(-i\int_{t_{1}}^{t_{2}}\mathbb{H}_{SE, \Lambda}(t)dt\right).\label{eq:SM_definition_restricted_evolution_operator}
\end{equation}
With the notations of vectorization in the doubled Hilbert space, Eq. (\ref{eq:definition_local_error}) can be
rewritten as 
\begin{equation}
|O_{\mathrm{tar}}-O_{\mathrm{sim}}|=\max_{\lVert\vecket{\rho_S(0)}\rVert=1}|\vecbra{{O_{X},I_{E}}|\mathbb{U}_{SE}(t,0)\vecket{\rho_S(0),\rho_E(0)}}-\vecbraket{O_{X}|\mathcal{E}(\rho_S(0))}|,
\end{equation}
with the abbreviation $\vecbra{O_{X},I_{E}}=\vecbra{O_{X}}_{S}\otimes\vecbra{I_{E}}_{E}$
and $\vecket{\rho_S(0),\rho_E(0)}=\vecket{\rho_S(0)}_{S}\otimes\vecket{\rho_{E}(0)}_E$.

\subsection{Proof of Lemma \ref{Lemma:local_trotter}}

In this subsection, we analyze the local error in Trotterization formula
defined as 
\begin{equation}
\begin{aligned}\delta_{\mathrm{loc,tro}} & :=\max_{\lVert O_{X}\rVert\leq1,\lVert \rho_S \rVert_\mathrm{tr}=1}|O_{\mathrm{tar}}-O_{\mathrm{tro}}|\\
 & =\max_{\Vert O_{X}\rVert\leq1,\Vert\vecket{\rho_S}\rVert=1}|\vecbraket{O_{X},I_{E}|\mathbb{U}_{SE}(t,0)-\prod_{i=1}^{T}\mathbb{V}\left(i\Delta t,(i-1)\Delta t\right)|\rho_S(0),\rho_E(0)}|,
\end{aligned}
\label{eq:definition_local_error_trotter}
\end{equation}
with $\mathbb{V}\left(i\Delta t,(i-1)\Delta t\right)$ the Trotterization
formula similar as Eq. (\ref{eq:SM_definition_1D_trotter_stage}) but now for a $D$-dimensional lattice.
We also assume each $\mathbb{V}\left(i\Delta t,(i-1)\Delta t\right)$
has $s_P$ stages as 
\[
\mathbb{V}\left(i\Delta t,(i-1)\Delta t\right)=\prod_{j=1}^{s_P}\mathbb{V}_{j}(i\Delta t,(i-1)\Delta t).
\]
For each stage, $\mathbb{V}_{j}$ consists of $l$ unitaries 
\[
\mathbb{V}_{j}(i\Delta t,(i-1)\Delta t)=\prod_{n=1}^{l}\mathbb{V}_{B_{n}}((i-1+f_{j}^{n})\Delta t,(i-1+f_{j-1}^{n})\Delta t)
\]
with $B_{n}\subseteq\mathbb{R}^{D}$ denoting the support of the evolution operator as in Eq. (\ref{eq:SM_definition_restricted_evolution_operator}), and $\{f_{j}^{n}\}$ the intermediate
time points, $0\leq f_{j}^{n}\leq 1$. The $l$-partition
of the lattice $\{B_{n}\}$ is chosen such that $\nu\cap\nu'=\emptyset$
for $\nu,\nu'\subseteq B_{n},n\in[1:l]$. The main technics
used here is the error remainder derived in Sec. \ref{sec:Detailed-Proof-ofLemma1}
(Eq. (\ref{eq:Chunk_error_p_order}), Eqs. (\ref{eq:rewritten_each_chunk_term})-(\ref{eq:function_to_estimate}))
and the Lieb-Robinson bound for non-Markovian systems which is well
established in Ref. \cite{trivedi2024liebrobinsonboundopenquantum}.
For convenience, we restate the latter here with our notations.

\begin{Lemma}
[Lieb-Robinson bound for non-Markovian dynamics (from Lemma 16 in Ref. (\cite{trivedi2024liebrobinsonboundopenquantum}))].

Define
\begin{align}
\gamma_{\zeta}^{X,Y}(t,t') & =\sup_{\substack{\phi\in S_{\zeta}(\rho_{E})\\
\lVert O_{X}\rVert,\lVert\mathbb{L}_{Y}\rVert_{\diamond}\leq1
}
}|\vecbraket{O_{X},I_{E}|\mathbb{U}_{SE}(t,t')\mathbb{L}_{Y}|\phi}|,\label{eq:SM_LR_definition_gamma}
\end{align}
where $O_{X}$, $\mathbb{L}_{Y}$ are the system operator and superoperator,
which support at $X,Y\subseteq\mathbb{R}^{D}$, respectively. $\mathbb{L}_{Y}$
further satisfies that $\vecbra{I_{S}}\mathbb{L}_{Y}=0$. The operator
space $S_{\zeta}(\rho_{E})$ with $\zeta:\mathbb{R}\to[0,\infty)$
a scalar function is defined in the following way. An operator $\phi\in S_{\zeta}(\rho_{E})$
if there exists an integer $n$, a set of system superoperators $\{\Omega_{i}:\lVert\Omega_{i}\rVert_{\diamond}\leq1\}_{i\in[1:n]}$
and the sets $\{A_{i}\subseteq\mathbb{R}^{D}\}_{i\in[1:n]},$$\{s_{i},t_{i}\in\mathbb{R}\}_{i\in[1:n]}$
such that
\begin{equation}
\vecket{\phi}=\prod_{i=1}^{n}\Omega_{i}\mathbb{U}_{A_{i}}(t_{i},s_{i})\vecket{\rho_S(0),\rho_E(0)},
\end{equation}
and
\begin{equation}
\max_{\nu}\sum_{i=1}^{n}\Theta_{A_{i}}(\nu)|\int_{s_{i}}^{t_{i}}K(\tau,\tau')d\tau'|\leq\zeta(\tau)\ \mathrm{for}\ \forall\tau,\label{eq:SM_definition_operator_space}
\end{equation}
where $K(\tau,\tau')=\max_{\nu,\sigma,\sigma'}|K_\nu^{(\sigma,\sigma')}(\tau,\tau')|$ and $\Theta_{A_{i}}(\nu)=1$ if $\nu\subseteq A_{i}$ and $0$ otherwise,
$\nu$ is defined by the form of the Hamiltonian Eq. (\ref{eq:Hamiltonian_in_higher_dimension_interaction}).
It can be proved that 
\begin{equation}
\gamma_{\zeta}^{X,Y}(t,t')\leq|Y|\bigg(\exp\left(v(\zeta)|t-t'|\right)-1\bigg)\exp\left(-\frac{\mathrm{dist}(X,Y)}{a_{0}}\right),
\end{equation}
where $v(\zeta)=a_{1}(1+16\lVert\zeta\rVert_{\infty}+40M)$, $\mathrm{dist}(X,Y)$ is the distance between the two regions $X$ and $Y$ defined via $\mathrm{dist}(X,Y)=\min_{\bm{v}_i\in X, \bm{v}_j\in Y}|\bm{v}_i-\bm{v}_j|$, $|Y|$ is the number of sites in the region $Y$, and $a_{0},a_{1}$
are two constants which only depend on the dimension $D$.
\label{lamma:LR_bound_non_mark1}
\end{Lemma}

Now, we start from one term in the error remainders in Eqs. (\ref{eq:rewritten_each_chunk_term})-(\ref{eq:function_to_estimate}),
which can be expressed as 
\begin{equation}
|P\int_{0}^{\Delta t}d\tau\int_{0}^{1}dx(1-x)^{P}\frac{\tau^{P}}{P!}\vecbraket{O_{X},I_{E}|\mathbb{U}_{SE}(t,t_{1}+\tau)\mathbb{V}(t_{1}+\tau,t_{1})\mathcal{J}_{P}^{1}(x\tau)\tilde{\mathbb{V}}(t_{1},0)|\rho_S(0),\rho_{E}(0)}|.\label{eq:SM_chunk_error_local_ob}
\end{equation}
Recall that $t_1=(i-1)\Delta t$. Here in higher dimensions, $\mathcal{J}_{P}^{1}(\tau)$ can be written
as
\begin{equation}
\mathcal{J}_{P}^{1}(\tau)=\sum_{k=1}^{s_Pl}\sum_{w_{1}+\cdots+w_{k}=P}\left(\prod_{\sigma=k-1}^{1}\mathrm{Ad}_{\mathbb{V}_{B_{n_{\sigma}}}(\alpha_{\sigma},\beta_{\sigma})}\mathrm{ad}_{\mathbb{H}_{B_{n_{\sigma}}}^{SE}(\beta_{\sigma})}^{w_{\sigma}}\right)\left(\sum_{\nu\subseteq\mathbb{R}^{D}}C_{\nu}^{k}(\mathbb{H}_{\nu}+\sum_{\alpha=1}^{4}\mathbb{J}_{\nu}^{\alpha}\mathbb{A}_{\nu}^{\alpha}(t_{k,\nu}))\right).\label{eq:SM_higher_dimension_J1}
\end{equation}
Here $C_{\nu}^{k}\leq s_P$, $t_1\leq\alpha_{\sigma},\beta_{\sigma},t_{k,\nu}\leq t_1+\tau$
are intermediate time points, $n_{\sigma}\in[1:n]$. Specifically,
$n_{\sigma}=(\sigma\ \mathrm{mod}\ l)+1$, $\alpha_{\sigma}=t_1+\tau f_{\lceil\frac{\sigma}{l}\rceil-1}^{n_{\sigma}}$,
$\beta_{\sigma}=t_1+\tau f_{\lceil\frac{\sigma}{l}\rceil}^{n_{\sigma}}$, and
$t_{k,\nu}=t_1+\tau f_{\lceil\frac{k-n}{l}\rceil}^{n}\ \mathrm{for}\ \nu\in B_{n}$,
though their exact form will not be used in the analysis. For fixed
$k,\nu$ and $w_{1},\cdots,w_{k}$, following the similar analysis
in Sec. \ref{sec:Detailed-Proof-ofLemma1}, the support of each term
in Eq. (\ref{eq:SM_higher_dimension_J1}) can be restricted to a region
whose diameter is at most $2k-1$. Consequently, the number of monomials
as well as the diamond norm of each monomial after contracting all
the bosonic operators $\mathbb{A}_{\nu'}^{\alpha}$ can be upper bounded
by $2(2^{2P+1}(2k)^{(2P+1)D})$, $2^{P+1}s_P[16Pm+16M(1+2s_P+k+P)]^{P+1}$,
respectively.

Now let us analyze Eq. (\ref{eq:SM_chunk_error_local_ob}) after fixing
$k,\nu,w_{1},\cdots,w_{k}$, which reads as 

\begin{equation}
\begin{aligned} & |P\int_{0}^{\Delta t}d\tau\int_{0}^{1}dx(1-x)^{P}\frac{\tau^{P}}{P!}\times\\&\llangle O_{X},I_{E}|\mathbb{U}_{SE}(t,t_{1}+\tau)\mathbb{V}(t_{1}+\tau,t_{1})\\
 & \left(\prod_{\sigma=k-1}^{1}\mathrm{Ad}_{\mathbb{V}_{B_{n_{\sigma}}}(\alpha_{\sigma},\beta_{\sigma})}\mathrm{ad}_{\mathbb{H}_{B_{n_{\sigma}}}^{SE}(\beta_{\sigma})}^{w_{\sigma}}\right)\left(C_{\nu}^{k}(\mathbb{H}_{\nu}+\sum_{\alpha=1}^{4}\mathbb{J}_{\nu}^{\alpha}\mathbb{A}_{\nu}^{\alpha}(t_{k,\nu}))\right)\tilde{\mathbb{V}}(t_{1},0)|\rho_S(0),\rho_{E}(0)\rrangle|.
\end{aligned}
\label{eq:SM_analyze_one_term_in_error}
\end{equation}
For each monimial in Eq. (\ref{eq:SM_analyze_one_term_in_error}),
after we contract all the bosonic operators we find that it has the
same form as the right hand side of Eq. (\ref{eq:SM_LR_definition_gamma}).
Actually, we can track the most left superoperator in Eq. (\ref{eq:SM_analyze_one_term_in_error})
which is not an evolution operator, and we find that it can only be
either $\mathbb{H}_{\nu'}$ or $\mathbb{J}_{\nu'}^{\alpha}\mathbb{K}_{\alpha\beta}$
for some $\nu',\beta$. By definition, both of these superoperator
annihilates the state $\vecbra{I_{S}}$ as $\vecbra{I_{S}}\mathbb{H}_{\nu'}=\vecbra{I_{S}}\mathbb{J}_{\nu'}^{\alpha}\mathbb{K}_{\alpha\beta}=0$.
Thus, we can identify this superoperator as the $\mathbb{L}_{Y}$
in Eq. (\ref{eq:SM_LR_definition_gamma}) with $|Y|=|\nu'|=2$. The evolution superoperator
on the left side of this one is identified as $\mathbb{U}_{SE}(t,t')$
in Eq. (\ref{eq:SM_LR_definition_gamma}) and all the superoperators
on the right side can be absorbed into $\vecket{\phi}$ in Eq. (\ref{eq:SM_LR_definition_gamma})
with a proper normalization factor $\mathcal{N}$. Consequently, we
can upper bound each monomial in the integrand of Eq. (\ref{eq:SM_analyze_one_term_in_error}) by
\begin{equation}
\mathcal{N}\gamma_{\zeta}^{X,Y(\nu)}(t,t'),
\end{equation}
with $\zeta,t'$ being analyzed more carefully in the following. The
normalization factor $\mathcal{N}$ can be upper bounded by the diamond
norm of each monomial discussed in the previous paragraph, i.e., $2^{P+1}s_P[16Pm+16M(1+2s_P+k+P)]^{P+1}$.
Here the argument $\nu$ of $Y$ indicates that the support of $\mathbb{L}_{Y}$
may depend on the $\nu$ chosen in Eq. (\ref{eq:SM_higher_dimension_J1}).

The value of $\zeta,t'$ can be determined by the position of $\mathbb{L}_{Y}$
in the expression Eq. (\ref{eq:SM_analyze_one_term_in_error}). Since
$\mathbb{L}_{Y}$ can not reside on the right hand of the first ad
in Eq. (\ref{eq:SM_analyze_one_term_in_error}) counted from right
to left, i.e., the ad associated with the largest $\sigma$ such that
$w_{\sigma}\neq0$, $t'\geq t_{1}-k\Delta t\geq t_{1}-s_Pl\Delta t$.
Here, $k$ comes from the maximum number of Ad before the first ad in Eq.
(\ref{eq:SM_analyze_one_term_in_error}). Since $t_{1}\geq0$, we
obtain $t'\geq-s_Pl\Delta t$. Similarly, as we absorb all the superoperator
on the right side of $\mathbb{L}_{Y}$ to $\vecket{\phi}$, we should
count their contribution to $\zeta$ as required by Eq. (\ref{eq:SM_definition_operator_space}).
Here, we can choose $\zeta$ as
\begin{equation}
\zeta(\tau)\geq s_P\int_{0}^{t_{1}}|K(\tau,\tau')|d\tau'+2s_P\int_{t_{1}}^{t_{1}+\Delta t}|K(\tau,\tau')|d\tau'+s_P\int_{t_{1}}^{t_{1}+\tau}|K(\tau,\tau')|d\tau'+\int_{t_{1}+\tau}^{t}|K(\tau,\tau')|d\tau'\label{eq:SM_zeta_tau_value}
\end{equation}
such that Eq. (\ref{eq:SM_definition_operator_space}) is satisfied.
Here the first term in Eq. (\ref{eq:SM_zeta_tau_value}) is from the
contribution of $\tilde{\mathbb{V}}(t_{1},0)$ in $\vecket{\phi}$, and
the $s_P$ factor comes from the assumption that each $\mathbb{V}(i\Delta t,(i-1)\Delta t)$
consists of $s_P$ stages. The second term is from the possible contribution
of all the Ad in $\mathcal{J}_{P}^{1}$, with each Ad contributes
a factor of $2$. The third term is from the possible contribution
of $\mathbb{V}(t_{1}+\tau,t_{1})$ in Eq. (\ref{eq:SM_analyze_one_term_in_error})
where $0\leq\tau\leq\Delta t$ and the last term is from $\mathbb{U}_{SE}(t,t_{1}+\tau)$.
For convenience, we can choose 
\begin{equation}
\zeta(\tau)=3s_PM.
\end{equation}

Since $\gamma_{\zeta}^{X,Y(\nu)}(t,t')$ is a decreasing function
of $t'$, with the Lemma \ref{lamma:LR_bound_non_mark1} , we can upper bound the integrand (the last two lines) in Eq. (\ref{eq:SM_analyze_one_term_in_error})
by
\begin{equation}
a_{2}\max\left\{ \bigg(\exp\left(v_{LR}|t+s_Pl\Delta t|\right)-1\bigg)\exp\left(-\frac{\mathrm{dist}(X,Y(\nu))}{a_{0}}\right),1\right\} .\label{eq:SM_LR_applied_to_trotter}
\end{equation}
Here $a_{2}=2(2^{2P+1}(2k)^{(2P+1)D})2^{P+1}s_P[16Pm+16M(1+2s_P+k+P)]^{P+1}$,
which is obtained by multiplying the number of the monomials and the
normalization factor $\mathcal{N}$, $v_{LR}=a_{1}(1+16(3s_PM)+40M)$.
We also use the fact that $\gamma_{\zeta}^{X,Y(\nu)}(T,t')\leq1$.
Since $\mathrm{dist}(Y(\nu),\nu)$ can only be enlarged by the adjoint
operation in Eq. (\ref{eq:SM_higher_dimension_J1}), following a similar
analysis in Sec. \ref{sec:Detailed-Proof-ofLemma1} we obtain that
$\mathrm{dist}(Y(\nu),\nu)\leq k$. Consequently, 
\begin{equation}
\mathrm{dist}(X,Y(\nu))\geq\mathrm{dist}(X,\nu)-\mathrm{dist}(Y(\nu),\nu)\geq\mathrm{dist}(X,\nu)-k.
\end{equation}
This allows us to sum over $\nu$ in Eq. (\ref{eq:SM_LR_applied_to_trotter}),
leading to 
\begin{equation}
\begin{aligned} & a_{2}\sum_{\nu}\max\left\{ \bigg(\exp\left(v_{LR}|t+s_Pl\Delta t|\right)-1\bigg)\exp\left(-\frac{\mathrm{dist}(X,Y(\nu))}{a_{0}}\right),1\right\}, \\
&\leq  e^{\frac{k}{a_{0}}}a_{2}\sum_{\nu}\max\left\{ \bigg(\exp\left(v_{LR}|t+s_Pl\Delta t|\right)\bigg)\exp\left(-\frac{\mathrm{dist}(X,\nu)}{a_{0}}\right),1\right\}, \\
&\leq  e^{\frac{k}{a_{0}}}a_{2}\sum_{n=0}^{\infty}\sum_{\nu:\mathrm{dist}(X,\nu)=n}\max\left\{ \bigg(\exp\left(v_{LR}|t+s_Pl\Delta t|\right)\bigg)\exp\left(-\frac{n}{a_{0}}\right),1\right\}, \\
&\leq  2De^{\frac{k}{a_{0}}}\mathrm{diam}(X)a_{2}\frac{2^{D}}{(D-1)!}\sum_{n=0}^{\infty}(n+D-1)^{D-1}\max\left\{ \bigg(\exp\left(v_{LR}|t+s_Pl\Delta t|\right)\bigg)\exp\left(-\frac{n}{a_{0}}\right),1\right\} .
\end{aligned}
\end{equation}
Here, in the final inequality, we used the fact that $\sum_{\nu:\mathrm{dist}(X,\nu)=n}\leq2D\frac{2^{D}}{(D-1)!}(n+D-1)^{D-1}\mathrm{diam}(X)$,
see Ref. \cite{RahulUnpunlishedNoiseRobust} for detailed analysis.
Ref. \cite{RahulUnpunlishedNoiseRobust} also proved that the sum
in $n$ can be bounded by $\mathcal{O}((t+s_Pl\Delta t)^{D})$, therefore,
we obtain 
\begin{equation}
a_{2}\sum_{\nu}\max\left\{ \bigg(\exp\left(v_{LR}|t+s_Pl\Delta t|\right)-1\bigg)\exp\left(-\frac{\mathrm{dist}(X,Y(\nu))}{a_{0}}\right),1\right\} \leq\mathcal{O}\left((t+s_Pl\Delta t)^{D}\right).
\end{equation}

Integrating $x,\tau$ in Eq. (\ref{eq:SM_analyze_one_term_in_error})
which gives us an additional factor $(\Delta t)^{P+1}$. We notice
that summing over $k,w_{1},\cdots,w_{k}$ as well as considering the
contribution from $\mathcal{J}_{P}^{2}$ will not change the asymptotical
scaling of the Trotterization error. Because Eq. (\ref{eq:definition_local_error_trotter})
contains $T=t/\Delta t$ terms of the form of Eq. (\ref{eq:SM_chunk_error_local_ob}),
the error of Trotterization formula for local observables can be bounded
by 
\begin{equation}
\delta_{\mathrm{loc,tro}}\leq\mathcal{O}\left(t(\Delta t)^{P}(t+sl\Delta t)^{D}\right)=\mathcal{O}\left(t^{D+1}(\Delta t)^{P}\right).\label{eq:SM_result_local_trotter_error}
\end{equation}
This expression is independent of the total system size $N$ as expected. 

Following the analysis in Ref. \cite{RahulUnpunlishedNoiseRobust},
the expression Eq. (\ref{eq:SM_result_local_trotter_error}) directly
indicates that our non-dissipative, non-Markovian simulation algorithm
is robust against local noise. Indeed, if we assume that the local
error rate for two-qubit gates during the simulation is $\gamma$,
the total simulation error for local observables defined in Eq. (\ref{eq:definition_local_error})
can be bounded by 
\begin{equation}
|O_{\mathrm{tar}}-O_{\mathrm{sim}}|\leq\lVert O_{X}\rVert\times\left(\mathcal{O}(t^{D+1}(\Delta t)^{P})+\mathcal{O}(\gamma t^{D+1}/(\Delta t)^{D+1})\right),
\end{equation}
see Ref. \cite{RahulUnpunlishedNoiseRobust} for a detailed derivation of the second term. One can thus choose $\Delta t=\mathcal{O}(\gamma^{1/(P+D+1)})$ such
that 
\begin{equation}
|O_{\mathrm{tar}}-O_{\mathrm{sim}}|\leq\mathcal{O}\left(t^{D+1}\gamma^{P/(P+D+1)}\right),
\end{equation}
which is independent of the total system size $N$.

\subsection{Proof of Lemma \ref{Lemma:local_dis_and_qud}}

In the dissipative, non-Markovian simulation algorithm, we also need
to count the error for local observables in the discretization and
quditization procedures. In this subsection, we assume the observable
$\lVert O_{X}\rVert\leq1$ without loss of generality. The symbol $X_{[l]}$
denotes the set of all points within a distance
no greater than $l$ from the region $X$, as 
\begin{equation}
X_{[l]}=\{\bm{v}\in\mathbb{R}^D:\mathrm{dist}(X,\bm{v})\leq l\}.
\end{equation}
From Lieb-Robinson bound, the information cannot propagate
too fast. Therefore, $U_{SE}(0,t)O_{X}U_{SE}(t,0)$ should be close to
$U_{SE, X_{[l]}}(0,t)O_{X}U_{SE, X_{[l]}}(t,0)$, where $U_{SE, X_{[l]}}$ is defined in Eq. (\ref{eq:SM_definition_restricted_evolution_operator}) as
\begin{subequations}
\label{eq:SM_detailed_definition_of_U_SE_X}
\begin{align}
U_{SE,X_{[l]}}(t_2,t_1)=\mathcal{T}\exp\left(-i\int_{t_1}^{t_2}H_{SE, X_{[l]}}(t)dt\right)
\end{align}
with
\begin{align}
H_{SE,X_{[l]}}(t)=\sum_{\nu \subseteq X_{[l]}}\left(H_\nu+(J_\nu A_\nu^\dagger(t)+\text{h.c.})\right).
\end{align}
\end{subequations}
This is strictly proved in Ref. \cite{trivedi2024liebrobinsonboundopenquantum}
even for non-Markovian dynamics, which we summarize below

\begin{Lemma}[Local observable in restricted dynamics, from Proposition 1 in Ref. \cite{trivedi2024liebrobinsonboundopenquantum}].

Define 
\begin{equation}
\Delta_{O_{X}}(t,t';l)=\max_{\lVert\rho\rVert_{\mathrm{tr}}=1}\left|\mathrm{Tr}\left(U(t',t)O_{X}U(t,t')-U_{X_{[l]}}(t',t)O_{X}U_{X_{[l]}}(t,t')\right)\rho(t')\right|,
\end{equation}
where $U(t,t')$ can correspond to either a usual closed unitary dynamics
or a non-Markovian dynamics $U_{SE}(t,t')$ and $\rho(t')=U(t',0)\rho_S(0)\otimes\rho_E(0)U(0,t')$. There exists a Lieb-Robinson velocity
$v_{LR}^{u(m)}$ such that 
\begin{equation}
\Delta_{O_{X}}(t,t';l)\leq f(l)\exp\left(-\frac{l}{a_{0}}\right)\left(\exp\left(\frac{v_{LR}^{u(m)}|t-t'|}{a_{0}}\right)-1\right).
\end{equation}
Here $v_{LR}^{u},v_{LR}^{m}$ corresponds to usual closed unitary
dynamics or non-Markovian dynamics, respectively. $v_{LR}^{u}$ and
$a_{0}$ only depends on the dimensions and $v_{LR}^{m}=v_{LR}^{u}(1+56M)$.
$f(l)\leq\mathcal{O}(l^{D-1})$. \label{lemma:restricted_dy_LR}
\end{Lemma}

Let us first analyze the discretization error. We view the segment width $\eta$ as a tunable
parameter which will depend on $t, \delta$.

We now introduce a length $l$ whose value will be determined later.
We further define 
\[
\begin{aligned}
O(t)=&\mathrm{Tr}(U_{SE}(0,t)O_{X}U_{SE}(t,0)\rho_{SE}(0)),
\\O_{X_{[l]}}(t)=&\mathrm{Tr}(U_{SE,X_{[l]}}(0,t)O_{X}U_{SE,X_{[l]}}(t,0)\rho_{SE}(0)),
\\O^{\eta}(t)=&\mathrm{Tr}(\tilde{U}_{SE}(0,t)O_{X}\tilde{U}_{SE}(t,0)\rho_{SE}(0)),
\\O_{X_{[l]}}^{\eta}(t)=&\mathrm{Tr}(\tilde{U}_{SE, X_{[l]}}(0,t)O_{X}\tilde{U}_{SE,X_{[l]}}(t,0)\rho_{SE}(0)),
\end{aligned}
\]
where $\rho_{SE}(0)=\rho_S(0)\otimes\rho_E(0)$ is the initial system-environment state, $\tilde{U}_{SE}$ is defined above Lemma \ref{Lemma: lemma_citation_from_functional_integral} in the main text, and $\tilde{U}_{SE,X_{[l]}}$ is the evolution operator restricted to the set $X_{[l]}$ after discretization, defined similarly as Eq. (\ref{eq:SM_detailed_definition_of_U_SE_X}) with $A_\nu(t)$ replaced by $\tilde{A}_\nu(t)$.
We
can write the discretization error for local observable as 
\begin{equation}
\begin{aligned}\delta_{\mathrm{loc,dis}} & =\max_{\lVert O_{X}\rVert\leq1,\lVert\rho_S\rVert_\mathrm{tr}=1}|O(t)-O^{\eta}(t)|,\\
 & \leq\max_{\lVert O_{X}\rVert\leq1,\lVert\rho_S\rVert_\mathrm{tr}=1}\left(|O(t)-O_{X_{[l]}}(t)|+|O_{X_{[l]}}(t)-O_{X_{[l]}}^{\eta}(t)|+|O_{X_{[l]}}^{\eta}(t)-O^{\eta}(t)|\right).
\end{aligned}
\label{eq:SM_discretization_error_for_local_ob}
\end{equation}
The first and the third term can be bounded by Lemma \ref{lemma:restricted_dy_LR}. Since the
first term corresponds to the original non-Markovian dynamics, its
Lieb-Robinson velocity $v_{LR}^{m}$ is directly given by $v_{LR}^{m}=v_{LR}^{u}(1+56M)$.
The third term corresponds to the non-Markovian dynamics after discretization,
therefore, its Lieb-Robinson velocity is given by 
\begin{equation}
v_{LR}^{m}=v_{LR}^{u}(1+56\max_{t,i}\int_{-\infty}^{\infty}|\tilde{K}_i(t,t')|dt').
\end{equation}
Here $\tilde{K}_i(t,t')=\max_{\sigma,\sigma'}|\tilde{K}_i^{(\sigma,\sigma')}(t,t')|$ with $\tilde{K}_i^{(\sigma,\sigma)}$ defined in Sec. \ref{sec:Extension-to-general}. It is obvious that
\begin{equation}
\tilde{K}_i(t,t')\leq K_i(t,t')+\sum_{\sigma,\sigma'}|\tilde{K}_i^{(\sigma,\sigma')}(t,t')-K_i^{(\sigma,\sigma')}(t,t')|,
\end{equation}
therefore
\begin{equation}
\int_{-\infty}^\infty|\tilde{K}_i(t,t')|dt'\leq \int_{-\infty}^\infty|K_i(t,t')|dt'+\sum_{\sigma,\sigma'}\int_{-\infty}^\infty|\tilde{K}_i^{(\sigma,\sigma')}(t,t')-K_i^{(\sigma,\sigma')}(t,t')|dt'.
\end{equation}
The last term on the right hand side of the above equation can be bounded by repeating the analysis in Eqs. (\ref{eq:SM_bounding_on_general_kernel}) and (\ref{eq:SM_bounding_on_general_kernel_integrals}). We take $\sigma=-,\sigma'=+$ as an example. We have
\begin{equation}
\begin{aligned} & \int_{-\infty}^{\infty}|\tilde{K}_{i}^{(-,+)}(t,t')-K_{i}^{(-,+)}(t,t')|dt'\\
&\leq  \int_{-\infty}^{\infty}dt'\int_{-\infty}^{\infty}d\tau\int_{-\infty}^{\infty}d\tau'|B(\tau,\tau')|\{|R_{t}(\tau)v^{*}(t'-\tau')|+|R_{t'}^{*}(\tau')v(t-\tau)|+|R_{t}(\tau)R_{t'}^{*}(\tau')|,\\
&\leq  \int_{-\infty}^{\infty}d\tau'|B(\tau,\tau')|\int_{t- r -\eta}^{t+ r +\eta}d\tau\int_{\tau'- r -\eta}^{\tau'+ r +\eta}dt'(2C_0\Gamma\eta^{j_{\mathrm{max}}+1}+\Gamma^{2}\eta^{2j_{\mathrm{max}}+2}),\\
&\leq  4( r +\eta)^{2}(2C_0\Gamma\eta^{j_{\mathrm{max}}+1}+\Gamma^{2}\eta^{2j_{\mathrm{max}}+2})\int_{-\infty}^{\infty}d\tau'|B(\tau,\tau'),\\
&\leq  \mathcal{O}\left(r^2\eta^{j_{\mathrm{max}}+1}\right).
\end{aligned}
\end{equation}
Other choices of $(\sigma,\sigma')$ can be upper bounded in a similar way and result in the same scaling. 
Thus, we obtain 
\begin{equation}
\max_{t,i}\int_{-\infty}^{\infty}\tilde{K}_i(t,t')dt' \leq M+\mathcal{O}(\eta^{j_{\mathrm{max}}+1}r^2).
\end{equation}
Since $\eta<1$, the Lieb-Robinson velocity for the third term is
$v_{LR}^{m}=v_{LR}^{u}(1+56M+\mathcal{O}(1))$. Therefore, the first
and the third term in Eq. (\ref{eq:SM_discretization_error_for_local_ob})
can be both bounded by 
$
\mathcal{O}\left(l^{D-1}\exp\left(\mathcal{O}\left((t-l)/a_{0}\right)\right)\right).
$

The second term in Eq. (\ref{eq:SM_discretization_error_for_local_ob})
is the discretization error investigated in Lemma \ref{Lemma:EB_discretization_total}, but restricted
to a system which only includes $\mathcal{O}(l^{D})$ sites. It can
be directly bounded by Lemma \ref{Lemma:EB_discretization_total} with $N$ replaced by $l^{D}$.
Therefore, we obtain 
\begin{equation}
\delta_{\mathrm{loc,dis}}\leq\mathcal{O}\left(l^{D-1}\exp\left(\mathcal{O}\left((t-l)/a_{0}\right)\right)\right)+\mathcal{O}\left(l^{D}t\eta^{j_{\mathrm{max}}+1}r^2\right).\label{eq:SM_final_result_on_discretized_local_op}
\end{equation}
We can choose $l$ to optimize the bound in Eq. (\ref{eq:SM_final_result_on_discretized_local_op}).
For example, if we choose $l=t+\mathcal{O}\log(1/\eta)$, we obtain
\begin{equation}
\begin{aligned}\delta_{\mathrm{loc,dis}} & \leq\mathcal{O}\left(t^{D+1}\eta^{j_{\mathrm{max}}+1}r^2\mathrm{polylog}(1/\eta)\right),\\
 & \leq\tilde{\mathcal{O}}\left(t^{D+1}\eta^{j_{\mathrm{max}+1}}r^2\right)=\tilde{\mathcal{O}}\left(t^{D+1}\eta^\mathrm{j_\mathrm{max}+1}\right).
\end{aligned}
\end{equation}
Here and throughout this section, we use the fact that $ r =\tilde{\mathcal{O}}(1)$. 

Similarly, we can analyze the quditization error for the local observables.
In this case, we will view the truncation level $d$ on the bosonic
spectrum as a tunable parameter and introduce another length scale
$l'$. By noting that the quditization is performed after Trotterzation, we further define 
\[
\begin{aligned}
O^{\eta}_\mathrm{tro}(t)=&\mathrm{Tr}(U_\mathrm{tro}^\dagger(t,T)O_{X}U_\mathrm{tro}(t,T)\rho_{SE}(0)),\\
O_{\mathrm{tro}, X_{[l']}}^{\eta}(t)=&\mathrm{Tr}(U_{\mathrm{tro},X_{[l']}}^\dagger(t,T)O_{X}U_{\mathrm{tro}, X_{[l']}}(t,T)\rho_{SE}(0)),\\
O^{\eta,d}_\mathrm{tro}(t)=&\mathrm{Tr}(\hat{U}_\mathrm{tro}^\dagger(t,T)O_{X}\hat{U}_\mathrm{tro}(t,T)\rho_{SE}(0)),\\
O_{\mathrm{tro},X_{[l']}}^{\eta,d}(t)=&\mathrm{Tr}(\hat{U}_{\mathrm{tro},X_{[l']}}^\dagger(t,T)O_{X}\hat{U}_{\mathrm{tro},X_{[l']}}(t,T)\rho_{SE}(0)),
\end{aligned}
\]
where $U_\mathrm{tro}(t,T)$ is obtained from the Trotterization formula in Lemma \ref{lemma1} applied to $\tilde{H}_{SE}$ as defined in the beginning of Subsec. \ref{subsec:Proof-of-Lemma5} and $\hat{U}_\mathrm{tro}(t,T)$ is obtained from the Trotterization formula in Lemma \ref{lemma1} applied to $P_d\tilde{H}_{SE}P_d$ as also defined in the beginning of Subsec. \ref{subsec:Proof-of-Lemma5}. $U_{\mathrm{tro},X_{[l']}}(t,T)$ is obtained from the Trotterization formula in Lemma \ref{lemma1} applied to Hamiltonian restricted to the region $X_{[l']}$
\[
\tilde{H}_{SE,X_{[l']}}=\sum_{\nu\in X_{[l']}}\left(H_\nu+(J_\nu A_\nu^\dagger(t)+\text{h.c.})\right),
\]
and $\hat{U}_{\mathrm{tro},X_{[l']}}$ is obtained from the Trotterization formmula in Lemma \ref{lemma1} applied to $P_d\tilde{H}_{SE,X_{[l']}}P_d$.

We can write the quditization error for local observable as 
\begin{equation}
\begin{aligned}\delta_{\mathrm{loc,qud}} & =\max_{\lVert O_{X}\rVert\leq1,\lVert\rho_S\rVert_\mathrm{tr}=1}|O^{\eta}_\mathrm{tro}(t)-O_\mathrm{tro}^{\eta,d}(t)|,\\
 & \leq\max_{\lVert O_{X}\rVert\leq1,\lVert\rho_S\rVert_\mathrm{tr}=1}\left(|O_\mathrm{tro}^{\eta}(t)-O_{\mathrm{tro},X_{[l']}}^{\eta}(t)|+|O_{\mathrm{tro},X_{[l']}}^{\eta}(t)-O_{\mathrm{tro},X_{[l']}}^{\eta,d}(t)|+|O_{\mathrm{tro},X_{[l']}}^{\eta,d}(t)-O_\mathrm{tro}^{\eta,d}(t)|\right).
\end{aligned}
\label{eq:SM_quditization_error_for_local_ob}
\end{equation}

We emphasize that the Trotterization procedure has the same effect as a purely forward evolution except that it extends the evolution time from $t$ to $s_Pt$.
In Eq. (\ref{eq:SM_quditization_error_for_local_ob}), the first term
and the third term can be bounded by Lemma \ref{lemma:restricted_dy_LR} with the non-Markovian
Lieb-Robinson velocity and usual unitary Lieb-Robinson velocity, respectively.
Since the Trotterization is performed before the quditization,
the non-Markovian Lieb-Robinson velocity corresponding to the first
term is thus given by $v_{LR}^{m}=v_{LR}^{u}(1+56s_PM)$, where the factor
$s_P$ is the number of stages in each time step.
Therefore, the first term can be bounded by 
$
\mathcal{O}\left(l'^{D-1}\exp\left(\mathcal{O}\left((t-l')/a_{0}\right)\right)\right).
$

For the third term, the Hamiltonian related to $P_d\tilde{H}_{SE}P_d$ has finite
dimension, thus, we can view it as a usual closed unitary dynamics
including both the system and ancillas. However, one should notice
the norm of each local term in this Hamiltonian is not $\mathcal{O}(1)$
but $\mathcal{O}(\sqrt{d/\eta})$. This corresponds
to rescale the evolution time. Thus, this term can be bounded by 
$
\mathcal{O}\left(l'^{D-1}\exp\left(\mathcal{O}\left((\sqrt{d/\eta}t-l')/a_{0}\right)\right)\right).
$

The second term in Eq. (\ref{eq:SM_quditization_error_for_local_ob})
is the quditization error investigated in Lemma \ref{Lemma:EB_quditization_total}, but restricted
to a system which only includes $\mathcal{O}(l'^{D})$ sites. It can
be directly bounded by Lemma \ref{Lemma:EB_quditization_total} with $N$ replaced by $l'{}^{D}$,
which reads as 
$
\mathcal{O}\left(\sqrt{d}e^{-d}\frac{l'{}^{2D}t^{2}}{\eta^\frac{3}{2}}\right).
$

$\delta_{\mathrm{loc,qud}}$ is obtained by summing
the above three error terms,
which gives us 
\begin{equation}
\delta_{\mathrm{loc,qud}}=\mathcal{O}\left(l'^{D-1}\left[\exp\left(\mathcal{O}\left((t-l')/a_{0}\right)\right)+\exp\left(\mathcal{O}\left((\sqrt{d/\eta}t-l')/a_{0}\right)\right)\right]\right)+\mathcal{O}\left(\sqrt{d}e^{-d}\frac{l'^{2D}t^{2}}{\eta^{\frac{3}{2}}}\right).
\end{equation}
If we choose $l'=\sqrt{d/\eta}t+da_{0}\leq\mathcal{O}(td/\sqrt{\eta})$, we
obtain 
\begin{equation}
\delta_{\mathrm{loc,qud}}\leq\mathcal{O}\left(d^{2D+1/2}e^{-d}\frac{t^{2+2D}}{\eta^{D+\frac{3}{2}}}\right).
\end{equation}

By counting the discretization, Trotterization, and quditization error
together, we obtain the final local error for the dissipative, non-Markovian
simulation algorithm as 
\begin{equation}
\delta_{\mathrm{loc}}\leq\tilde{\mathcal{O}}\left(t^{D+1}\eta^{j_{\mathrm{max}}+1}\right)+\mathcal{O}\left(t^{D+1}(\Delta t)^{P}\right)+\mathcal{O}\left(d^{2D+1/2}e^{-d}\frac{t^{2+2D}}{\eta^{D+\frac{3}{2}}}\right).
\end{equation}
The circuit depth is 
\begin{equation}
\text{Circuit Depth  }=\mathcal{O}\left(\frac{t}{\Delta t}\sqrt{\frac{d}{\eta}}\right).\label{Eq:SM_circuit_depth_local_di}
\end{equation}
Now let us fix $j_{\mathrm{max}}=P/2$, in the absence of experimental
noise, we can choose 
\begin{equation}
\eta=\tilde{\mathcal{O}}\left(\left(\frac{\delta_\mathrm{loc}}{t^{D+1}}\right)^{\frac{2}{P+2}}\right),\ \Delta t=\mathcal{O}\left(\left(\frac{\delta_\mathrm{loc}}{t^{D+1}}\right)^{\frac{1}{P}}\right),\ d=\mathcal{O}\left(\log\left(\frac{t^{2D+2}}{\delta_\mathrm{loc}}\times\left(\frac{t^{D+1}}{\delta_\mathrm{loc}}\right)^{\frac{2D+3}{P+2}}\right)\right),
\end{equation}
such that the total error for local observables is below $\delta_\mathrm{loc}$.
The circuit depth is 
\begin{equation}
\text{Circuit Depth  }=\tilde{\mathcal{O}}\left(t\left(\frac{t^{D+1}}{\delta_\mathrm{loc}}\right)^{\frac{2}{P}}\right),
\end{equation}
which is independent of the system size $N$. With recycling, the number of ancillary
qubits per system site is 
\begin{equation}
\text{Number of Ancillaries per System Site }=\mathcal{O}\left(\frac{\log(d)}{\eta}\right)=\tilde{\mathcal{O}}\left(\left(\frac{t^{D+1}}{\delta_\mathrm{loc}}\right)^{\frac{2}{P+2}}\right),
\end{equation}
which is also independent of the system size $N$.

In the presence of experimental noise, if the two-qubit gate error
rate is $\gamma$, the total error for local observable is given by
\begin{equation}
\delta_{\mathrm{loc}}\leq\tilde{\mathcal{O}}\left(t^{D+1}\eta^{P/2+1}\right)+\mathcal{O}\left(t^{D+1}(\Delta t)^{P}\right)+\mathcal{O}\left(d^{2D+1/2}e^{-d}\frac{t^{2+2D}}{\eta^{D+\frac{3}{2}}}\right)+\mathcal{O}\left(\gamma\left(\frac{t}{\Delta t}\sqrt{\frac{d}{\eta}}\right)^{D+1}\right),
\end{equation}
where we also take $j_{\mathrm{max}}=P/2$. The last term in $\delta_\mathrm{loc}$ comes from Ref. \citep{RahulUnpunlishedNoiseRobust} by noticing that the circuit depth is given by Eq. (\ref{Eq:SM_circuit_depth_local_di}). One may optimize the tunable
parameters $\eta$, $\Delta t$, $d$ to obtain the optimal performance.
Here, just as an example, we choose 
\begin{equation}
\eta=(\Delta t)^{2},\ d=\mathcal{O}\left(\log\left(\frac{t^{D+1}}{(\Delta t)^{P+2D+3}}\right)\right),
\end{equation}
such that 
\begin{equation}
\delta_{\mathrm{loc}}\leq\tilde{\mathcal{O}}(t^{D+1}(\Delta t)^{P})+\tilde{\mathcal{O}}\left(\gamma\left(\frac{t}{(\Delta t)^{2}}\right)^{D+1}\right).
\end{equation}
We can choose 
\begin{equation}
\Delta t=\tilde{\mathcal{O}}\left(\gamma^{1/(2D+P+2)}\right),
\end{equation}
which leads to 
\begin{equation}
\delta_{\mathrm{loc}}\leq\tilde{\mathcal{O}}\left(t^{D+1}\gamma^{P/(2D+P+2)}\right).
\end{equation}
Obviously, this local error is independent with the system size $N$. By choosing $P=2p$, we obtained the statement in the main text.

\section{Simulation of commuting, non-dissipative Lindbladian master equation\label{Sec:SM_commuting_nondissi}}
In this section, we provide proofs for several results we used in Sec. \ref{subsection:proof_markovian_non}.

\subsection{Generalizing Hamiltonian simulation algorithms (proof of Lemma \ref{Lemma:generalizing_time_dependent_simulation_closed})}
The only difference from Ref. \cite{Kieferov2019Simulating} is that in their case, the Hamiltonian is differentiable, but in our case, we only assume that the Hamiltonian is Holder-$\alpha$ continuous. In Ref. \cite{Kieferov2019Simulating}, the property that the Hamiltonian is differentiable is used to upper bound
the discretization error in the integration of time. Here we provide a similar error bound for Holder-$\alpha$ continuous Hamiltonian, which only increases the circuit depth by an $\mathcal{O}(\mathrm{polylog}(Nt/\delta))$ extra factor.  We assume that 
\[
\lVert H_{i,i+1}(s)-H_{i,i+1}(s')\rVert \leq D_\mathrm{Hol,\alpha}|s-s'|^\alpha 
\]
for some constant $D_\mathrm{Hol,\alpha}$. 
We can discretize the integration $\int_{s}^{s+\Delta t}H_{i,i+1}(\tau)d\tau$
with $\Delta t\leq\mathcal{O}(1)$ into $M$ pieces as 
\begin{equation}
 \int_{s}^{s+\Delta t}H_{i,i+1}(\tau)d\tau\approx\frac{\Delta t}{M}\sum_{j=1}^{M}H_{i,i+1}(s_{j}),\label{eq:discretization_time_integration}
\end{equation}
where $s_{j}=s+j\Delta t/M$.
The error can be bounded by
\begin{equation}
\begin{aligned} & \left\lVert\int_{s}^{s+\Delta t}H_{i,i+1}(\tau)d\tau-\frac{\Delta t}{M}\sum_{j=1}^{M}H_{i,i+1}\left(s_{j}\right)\right\rVert,\\
&\leq  \sum_{j=1}^{M}\int_{s_{j-1}}^{s_{j}}\lVert H_{i,i+1}(\tau)-H_{i,i+1}(s_{j})\rVert d\tau,\\
&\leq  D_\mathrm{Hol,\alpha}\sum_{j=1}^M \int_{s_{j-1}}^{s_j}|\tau - s_j|^\alpha d\tau, \\
&\leq  D_{\mathrm{Hol},\alpha}\frac{1}{1+\alpha} \sum_{j=1}^{M}\left(\frac{\Delta t}{M}\right)^{1+\alpha},\\
&\leq  \mathcal{O}(M^{-\alpha}).
\end{aligned}
\label{eq:CNM_error_bound_in_integration}
\end{equation}
 We can choose $M=\mathcal{O}(N^{1/\alpha}t^{1/\alpha}/\delta^{1/\alpha})=\mathrm{poly}(N,t,1/\delta)$ to ensure that this discretization error in time integration is less than
$\mathcal{O}(\delta/Nt)$.
This error bound is polynomially weaker than the one in Ref. \cite{Kieferov2019Simulating}, which is $\mathcal{O}(1/M)$ or equivalently $M=\mathcal{O}(Nt/\delta)$. 
However, the algorithm only scales logarithmically with $M$. Thus,
it only introduces an extra $\mathcal{O}(\mathrm{polylog}(Nt/\delta))$ factor into the gate complexity. With Eq. (\ref{eq:CNM_error_bound_in_integration}),
the algorithms in Ref. \cite{Kieferov2019Simulating}
can be directly applied. In summary, our algorithm is almost the same as the one in Ref. \cite{Kieferov2019Simulating} except that we use more numbers of discretization points $M$ in approximating integrations, which only costs an additional $\mathcal{O}(\mathrm{polylog}(Nt/\delta))$ factor in resource counting.
\subsection{Continuity of the Quantum Ito solver \label{sec:Continuity_of_Quantum_Ito_solver}}
In this subsection, we show the continuity of the quantum Ito solver which allows us to discretize the time integration.
For our purpose of simulating the Lindbladian dynamics, it is enough
to consider a discretized version of quantum Ito solver, which is
defined as 
\begin{equation}
V_{\mathrm{sto}}(\tau)=\frac{1}{\sqrt{\epsilon}}\sum_{n=0}^{\infty}\theta(\tau-n\epsilon)\theta\left((n+1)\epsilon-\tau\right)\sum_{i=1}^{N-1}(J_{i,i+1}a_{i,n}^{\dagger}+J_{i,i+1}^{\dagger}a_{i,n})
\end{equation}
with $[a_{i,n},a_{i',n'}^{\dagger}]=\delta_{i,i'}\delta_{n,n'}$,
$\theta(t)$ the step function, and $\epsilon\leq\mathcal{O}(\delta/(Nt))$ such that the total discretization error is below $\delta$.
In the following, we only consider the single body case, where the
subscript $i$ is ignored.

In the interaction picture, any intermediate system-environment state $\ket{\psi(\tau)}$
is given by 
\begin{equation}
\ket{\psi(\tau)}=\left(\mathcal{T}e^{-i\int_{0}^{\tau}V_{\mathrm{sto}}(s)ds}\right)^{\dagger}\mathcal{T}e^{-i\int_{0}^{\tau}(H_{S}+V_{\mathrm{sto}}(s))ds}\ket{\psi(0)},
\end{equation}
with the initial system-environment state $\ket{\psi(0)}$ a joint vacuum state of the bosonic bath. For
convenience, we define 
\begin{equation}
H_\mathrm{joint}(s)=\begin{cases}
H_{S}+V_{\mathrm{sto}}(\tau), & 0\leq s<\tau;\\
-V_{\mathrm{sto}}(2\tau -s), & \tau\leq s<2\tau,
\end{cases}
\end{equation}
such that 
\begin{equation}
\ket{\psi(\tau)}=\mathcal{T}e^{-i\int_{0}^{2\tau}H_\mathrm{joint}(s)ds}\ket{\psi(0)}.
\end{equation}
The corresponding equation of motion for each bosonic operator can
be integrated to 
\begin{equation}
\begin{aligned}a_{n}(\tau) & =a_{n}(0)+i\int_{0}^{2\tau}[H_\mathrm{joint}(s),a_{n}(s)]ds,\\
 & =a_{n}(0)-\frac{i}{\sqrt{\epsilon}}\int_{n\epsilon}^{(n+1)\epsilon}\theta(\tau-s)J(s)ds+\frac{i}{\sqrt{\epsilon}}\int_{2\tau-(n+1)\epsilon}^{2\tau-n\epsilon}\theta(2\tau-s)\theta(s-\tau)J(s)ds,
\end{aligned}
\end{equation}
with $J(s)=\left(\mathcal{T}e^{-i\int_{0}^{s}H_\mathrm{joint}(\tau')d\tau'}\right)^{\dagger}J\mathcal{T}e^{-i\int_{0}^{s}H_\mathrm{joint}(\tau')d\tau'}$.
Repeating the argument in Subsec. \ref{subsec:Proof-of-Lemma4}, we
can derive an exponential decay bound on the probability of the
occupation number for each bosonic mode $a_{n}$ at any time $\tau$,
which reads as 
\begin{equation}
\lVert (I-P_{n}(d))\ket{\psi(\tau)}\rVert \leq C_{1}e^{-d},\label{eq:probability_on_occupation_Quantum_Ito}
\end{equation}
where $P_{n}(d)=\sum_{i=0}^{d-1}\ket{i}_{n}\bra{i}$ is the projector
to the first $d$-level of the $n^\text{th}$ bosonic mode and $C_{1}$ is
a constant. Similar as in Subsec. \ref{subsec:Proof-of-Lemma5}, we can choose $d=\mathcal{O}(\mathrm{polylog}(N,t,\frac{1}{\delta}))=\mathcal{O}(\mathrm{polylog}(1/\epsilon))$
as a truncation parameter such that with high probability (almost
surely in the simulation procedure), no bosonic mode has occupation
number larger than $d$. Therefore, we can project $V_{\mathrm{sto}}(\tau)$ into
the subspace $\prod_{n}P_{n}(d)$ without losing accuracy. Consequently,
we can bound $\lVert W_{\mathrm{sto}}(t_2,0)-W_{\mathrm{sto}}(t_1,0)\rVert$ with
$n\epsilon\leq t_2-t_1<(n+1)\epsilon$ as 
\begin{equation}
\begin{aligned}\lVert W_{\mathrm{sto}}(t_2,0)-W_\mathrm{sto}(t_1,0)\rVert & \leq \int_{t_1}^{t_2}\left\lVert\prod_n P_n(d)V_\mathrm{sto}(\tau)\prod_n P_n(d)\right\rVert d\tau,\\
 & \leq\frac{2(t_2-t_1)}{\sqrt{\epsilon}}\sqrt{d+1},\\
 & \leq\mathcal{O}\left(\frac{t_2-t_1}{\sqrt{\epsilon}}\mathrm{polylog(1/\epsilon)}\right),\\
 & \leq\mathcal{O}\left(\frac{t_2-t_1}{\epsilon^{\frac{3}{4}}}\right),\\
 & \leq \mathcal{O}\left((t_2-t_1)^{\frac{1}{4}}\right),
\end{aligned}
\label{eq:continuity_bounding_on_the_quantum_ito}
\end{equation}
almost surely in the simulation procedure. Equation (\ref{eq:continuity_bounding_on_the_quantum_ito})
thus allows us to discretize the integration $\int_{s}^{s+\Delta s}W_{\mathrm{sto}}^{\dagger}(\tau,0)H_{S}W_{\mathrm{sto}}(\tau,0)d\tau$
into $M$ pieces with $M$ given below Eq. (\ref{eq:CNM_error_bound_in_integration}) and $\alpha=1/4$.

\section{Detailed proof of Theorem \ref{theorem3}\label{sec:Details-proof-ofObservation1}}

Here we consider an open quantum system of $N$ sites. Its Markovian
dynamics is described by 
\begin{equation}
\frac{d\rho_S}{dt}=\mathcal{L}(\rho_S)=-i[H_{S},\rho_S]+\sum_{i=1}^{N-1}J_{i,i+1}\rho_S J_{i,i+1}^{\dagger}-\frac{1}{2}\sum_{i=1}^{N-1}\{J_{i,i+1}^{\dagger}J_{i,i+1},\rho_S\}.\label{eq:SM_set_up}
\end{equation}
We further assume that each jump operator commutes with others, i.e.,
$[J_{i,i+1},J_{j,j+1}]=[J_{i,i+1},J_{j,j+1}^{\dagger}]=0$ for $i\neq j$.
Here $H_{S}=\sum_{i=1}^{N-1}H_{i,i+1}$ is a generic system Hamiltonian
which does not need to commute with $J_{i,i+1}$. Our target is to
construct a third-order dilated Hamiltonian $H_{\mathrm{dia}}^{(3)}$
acting on both the system and ancillas such that 

\begin{equation}
\mathcal{R}_\mathrm{dia}:=e^{\mathcal{L}\Delta t}-\mathcal{E}_\mathrm{dia}(\Delta t)=\mathcal{O}((\Delta t)^4).
\end{equation}
Throughout this section, the superscript on a Hamiltonian term denotes the order of dilation.
We also define the error remainders for the dilation at each order with respect to $\Delta t$, i.e. the superoperator $\mathcal{G}_\mathrm{dia}^{(p)}$ via
\begin{equation}
\mathcal{R}_\mathrm{dia}=\sum_{p=0}^{\infty}(\Delta t)^p\mathcal{G}_\mathrm{dia}^{(p)}.
\end{equation}
Our target is equivalent to require that $\mathcal{G}_\mathrm{dia}^{(p)}=0$ for $p\leq 3$.

\subsection{Single jump operator case}

Let's start from a simple case, where only one jump operator is involved,
i.e., 
\begin{equation}
\frac{d\rho_{S}}{dt}=-i[H_{S},\rho_{S}]+J\rho_{S} J^{\dagger}-\frac{1}{2}\{J^{\dagger}J,\rho_{S}\}.
\end{equation}
For this case, we only need to introduce a single ancilla. We assume that
the $p^\text{th}$ order dilated Hamiltonian $H_{\mathrm{dia}}^{(p)}$ has the following form:
\begin{equation}
H_{\mathrm{dia}}^{(p)}=H_{0}^{(p)}+\sum_{j}(S_{j}^{(p)}\otimes\ket{j}\bra{0}+S_{j}^{(p)}{}^{\dagger}\otimes\ket{0}\bra{j}),\label{eq:SM_form_of_single_dilated_Hamiltonian}
\end{equation}
where $\ket{j}$ denotes the $j^\text{th}$ excited level of the ancilla.
$S_{j}$ is a system operator.

At the first order $\mathcal{O}(\Delta t)$, we can choose 
\begin{equation}
H_{0}^{(1)}=H_{S};\ S_{1}^{(1)}=\frac{1}{\sqrt{\Delta t}}J;\ S_{j}^{(1)}=0\ \mathrm{for}\ j>1.\label{eq:SM_first_dilation_Hamiltonian}
\end{equation}
We can verify that the error remainder for the first order dilated
Hamiltonian is
\begin{equation}
\begin{aligned}
\mathcal{G}_\mathrm{dia}^{(0)}=\mathcal{G}_\mathrm{dia}^{(1)}= &0;\\
\mathcal{G}_{\mathrm{dia}}^{(2)}(\rho_S)= &-\frac{i\rho_{S} H_{S}J^{\dagger}J}{12}+\frac{i\rho_{S} J^{\dagger}H_{S}J}{6}-\frac{i\rho_{S} J^{\dagger}JH_{S}}{12}+\frac{\rho_{S} J^{\dagger}JJ^{\dagger}J}{12}+\\
 & \frac{iH_{S}J^{\dagger}J\rho_{S}}{12}-\frac{iJ^{\dagger}H_{S}J\rho_{S}}{6}-\frac{J^{\dagger}J^{2}\rho_{S} J^{\dagger}}{4}+\frac{iJ^{\dagger}JH_{S}\rho_{S}}{12}+\frac{J^{\dagger}JJ^{\dagger}J\rho_{S}}{12}+\\
 & \frac{J^{2}\rho_{S} J^{\dagger}{}^{2}}{2}-\frac{J\rho_{S} J^{\dagger}{}^{2}J}{4}-\frac{J\rho_{S} J^{\dagger}JJ^{\dagger}}{12}-\frac{JJ^{\dagger}J\rho_{S} J^{\dagger}}{12}.
\end{aligned}
\end{equation}
Now we are trying to introduce higher-order terms to cancel the error
component at order $\mathcal{O}((\Delta t)^{2})$ in $\mathcal{R}_{\mathrm{dia}}$, i.e. $\mathcal{G}_\mathrm{dia}^{(2)}$.
For this purpose, we modify Eq. (\ref{eq:SM_first_dilation_Hamiltonian})
to 
\begin{equation}
\begin{aligned}S_{1}^{(2)} & =\frac{1}{\sqrt{\Delta t}}J+\sqrt{\Delta t}\left(-\frac{J^{\dagger}J^{2}}{4}-\frac{JJ^{\dagger}J}{12}\right);\ \\
S_{2}^{(2)} & =\frac{J^{2}}{\sqrt{2}};\ S_{j}^{(2)}=0\ \mathrm{for}\ j>2.
\end{aligned}
\end{equation}

This can partially cancel the error component, leading to
\begin{equation}
\begin{aligned}\mathcal{G}_{\mathrm{dia}}^{(2)}(\rho_S)= & -\frac{i\rho_{S} H_{S}J^{\dagger}J}{12}+\frac{i\rho_{S} J^{\dagger}H_{S}J}{6}-\frac{i\rho_{S} J^{\dagger}JH_{S}}{12}+\frac{iH_{S}J^{\dagger}J\rho_{S}}{12}\\
 & -\frac{iJ^{\dagger}H_{S}J\rho_{S}}{6}+\frac{iJ^{\dagger}JH_{S}\rho_{S}}{12}.
\end{aligned}
\end{equation}
These remaining error components at order $\mathcal{O}((\Delta t)^{2})$
can be finally canceled by introducing additional term $\Delta t(-\frac{H_{S}J^{\dagger}J}{12}+\frac{J^{\dagger}H_{S}J}{6}-\frac{J^{\dagger}JH_{S}}{12})$
into $H_{0}$. Therefore, we obtain the second order dilated Hamiltonian
$H_{\mathrm{dia}}^{(2)}$ as

\begin{align}
\begin{aligned}H_{0}^{(2)} & =H_{S}+\Delta t\left(-\frac{H_{S}J^{\dagger}J}{12}+\frac{J^{\dagger}H_{S}J}{6}-\frac{J^{\dagger}JH_{S}}{12}\right);\\
S_{1}^{(2)} & =\frac{1}{\sqrt{\Delta t}}J+\sqrt{\Delta t}\left(-\frac{J^{\dagger}J^{2}}{4}-\frac{JJ^{\dagger}J}{12}\right);\ \\
 & S_{2}^{(2)}=\frac{J^{2}}{\sqrt{2}};\ S_{j}^{(2)}=0\ \mathrm{for}\ j>2.
\end{aligned}
\label{eq:SM_second_dilated_Hamiltonian_single}
\end{align}

With Eq. (\ref{eq:SM_second_dilated_Hamiltonian_single}), we can
compute the error remainders at the third order as 
\[
\begin{aligned}\mathcal{G}^{(3)}_{\mathrm{dia}}(\rho_S) & =\frac{\rho_{S} H_{S}J^{\dagger}H_{S}J}{24}-\frac{\rho_{S} H_{S}J^{\dagger}JH_{S}}{24}-\frac{i\rho_{S} H_{S}J^{\dagger}JJ^{\dagger}J}{120}+\frac{i\rho_{S} J^{\dagger}{}^{2}H_{S}J^{2}}{9}+\frac{\rho_{S} J^{\dagger}{}^{2}J^{2}J^{\dagger}J}{96}-\\
 & \frac{i\rho_{S} J^{\dagger}{}^{2}JH_{S}J}{18}+\frac{\rho_{S} J^{\dagger}{}^{2}JJ^{\dagger}J^{2}}{32}-\frac{\rho_{S} J^{\dagger}H_{S}^{2}J}{24}-\frac{i\rho_{S} J^{\dagger}H_{S}J^{\dagger}J^{2}}{18}+\frac{\rho_{S} J^{\dagger}H_{S}JH_{S}}{24}+\\
 & \frac{i\rho_{S} J^{\dagger}H_{S}JJ^{\dagger}J}{180}+\frac{i\rho_{S} J^{\dagger}JH_{S}J^{\dagger}J}{180}+\frac{\rho_{S} J^{\dagger}JJ^{\dagger}{}^{2}J^{2}}{96}+\frac{i\rho_{S} J^{\dagger}JJ^{\dagger}H_{S}J}{180}-\frac{i\rho_{S} J^{\dagger}JJ^{\dagger}JH_{S}}{120}-\\
 & \frac{\rho_{S} J^{\dagger}JJ^{\dagger}JJ^{\dagger}J}{480}+\frac{H_{S}J^{\dagger}H_{S}J\rho_{S}}{24}-\frac{iH_{S}J^{\dagger}J^{2}\rho_{S} J^{\dagger}}{12}-\frac{H_{S}J^{\dagger}JH_{S}\rho_{S}}{24}+\frac{iH_{S}J^{\dagger}JJ^{\dagger}J\rho_{S}}{120}+\\
 & \frac{iH_{S}J^{2}\rho_{S} J^{\dagger}{}^{2}}{12}-\frac{H_{S}J\rho_{S} H_{S}J^{\dagger}}{12}+\frac{iH_{S}J\rho_{S} J^{\dagger}{}^{2}J}{24}+\frac{H_{S}J\rho_{S} J^{\dagger}H_{S}}{12}-\frac{iH_{S}J\rho_{S} J^{\dagger}JJ^{\dagger}}{24}-\\
 & \frac{iJ^{\dagger}{}^{2}H_{S}J^{2}\rho_{S}}{9}+\frac{J^{\dagger}{}^{2}J^{2}J^{\dagger}J\rho_{S}}{96}+\frac{iJ^{\dagger}{}^{2}JH_{S}J\rho_{S}}{18}+\frac{J^{\dagger}{}^{2}JJ^{\dagger}J^{2}\rho_{S}}{32}-\frac{J^{\dagger}H_{S}^{2}J\rho_{S}}{24}+\\
 & \frac{iJ^{\dagger}H_{S}J^{\dagger}J^{2}\rho_{S}}{18}+\frac{iJ^{\dagger}H_{S}J^{2}\rho_{S} J^{\dagger}}{12}+\frac{J^{\dagger}H_{S}JH_{S}\rho_{S}}{24}-\frac{iJ^{\dagger}H_{S}JJ^{\dagger}J\rho_{S}}{180}-\frac{J^{\dagger}J^{3}\rho_{S} J^{\dagger}{}^{2}}{12}+\\
 & \frac{iJ^{\dagger}J^{2}\rho_{S} H_{S}J^{\dagger}}{24}+\frac{J^{\dagger}J^{2}\rho_{S} J^{\dagger}{}^{2}J}{48}-\frac{iJ^{\dagger}J^{2}\rho_{S} J^{\dagger}H_{S}}{24}-\frac{J^{\dagger}J^{2}\rho_{S} J^{\dagger}JJ^{\dagger}}{48}-\frac{iJ^{\dagger}J^{2}H_{S}\rho_{S} J^{\dagger}}{24}-\\
 & \frac{iJ^{\dagger}JH_{S}J^{\dagger}J\rho_{S}}{180}+\frac{iJ^{\dagger}JH_{S}J\rho_{S} J^{\dagger}}{24}+\frac{J^{\dagger}JJ^{\dagger}{}^{2}J^{2}\rho_{S}}{96}-\frac{iJ^{\dagger}JJ^{\dagger}H_{S}J\rho_{S}}{180}+\frac{J^{\dagger}JJ^{\dagger}J^{2}\rho_{S} J^{\dagger}}{24}+\\
 & \frac{iJ^{\dagger}JJ^{\dagger}JH_{S}\rho_{S}}{120}-\frac{J^{\dagger}JJ^{\dagger}JJ^{\dagger}J\rho_{S}}{480}+\frac{J^{3}\rho_{S} J^{\dagger}{}^{3}}{6}-\frac{iJ^{2}\rho_{S} H_{S}J^{\dagger}{}^{2}}{12}-\frac{J^{2}\rho_{S} J^{\dagger}{}^{3}J}{12}-\\
 & \frac{iJ^{2}\rho_{S} J^{\dagger}{}^{2}H_{S}}{12}-\frac{J^{2}\rho_{S} J^{\dagger}{}^{2}JJ^{\dagger}}{12}+\frac{iJ^{2}\rho_{S} J^{\dagger}H_{S}J^{\dagger}}{6}+\frac{iJ^{2}H_{S}\rho_{S} J^{\dagger}{}^{2}}{12}+\frac{iJ\rho_{S} H_{S}J^{\dagger}{}^{2}J}{24}+\\
 & \frac{iJ\rho_{S} H_{S}J^{\dagger}JJ^{\dagger}}{24}-\frac{iJ\rho_{S} J^{\dagger}{}^{2}H_{S}J}{12}+\frac{iJ\rho_{S} J^{\dagger}{}^{2}JH_{S}}{12}+\frac{J\rho_{S} J^{\dagger}{}^{2}JJ^{\dagger}J}{24}-\frac{iJ\rho_{S} J^{\dagger}H_{S}J^{\dagger}J}{24}-\\
 & \frac{iJ\rho_{S} J^{\dagger}H_{S}JJ^{\dagger}}{24}-\frac{J\rho_{S} J^{\dagger}JJ^{\dagger}JJ^{\dagger}}{120}+\frac{JH_{S}\rho_{S} H_{S}J^{\dagger}}{12}-\frac{iJH_{S}\rho_{S} J^{\dagger}{}^{2}J}{24}-\frac{JH_{S}\rho_{S} J^{\dagger}H_{S}}{12}+\\
 & \frac{iJH_{S}\rho_{S} J^{\dagger}JJ^{\dagger}}{24}-\frac{iJH_{S}J\rho_{S} J^{\dagger}{}^{2}}{6}+\frac{iJJ^{\dagger}H_{S}J\rho_{S} J^{\dagger}}{24}-\frac{JJ^{\dagger}J^{2}\rho_{S} J^{\dagger}{}^{2}}{12}-\frac{iJJ^{\dagger}J\rho_{S} H_{S}J^{\dagger}}{24}-\\
 & \frac{JJ^{\dagger}J\rho_{S} J^{\dagger}{}^{2}J}{48}+\frac{iJJ^{\dagger}J\rho_{S} J^{\dagger}H_{S}}{24}+\frac{JJ^{\dagger}J\rho_{S} J^{\dagger}JJ^{\dagger}}{48}-\frac{iJJ^{\dagger}JH_{S}\rho_{S} J^{\dagger}}{24}-\frac{JJ^{\dagger}JJ^{\dagger}J\rho_{S} J^{\dagger}}{120}.
\end{aligned}
\]
Similar as the previous procedure, we can first partially cancel the
error component by introducing higher order term to $S_{j}$. This
is done as follows

\begin{equation}
\begin{split}S_{1}^{(3)} & =\frac{1}{\sqrt{\Delta t}}\bigg\{ J-\frac{\Delta tJJ^{\dagger}J}{12}-\frac{\Delta tJ^{\dagger}J^{2}}{4}-\frac{(\Delta t)^{2}JJ^{\dagger}JJ^{\dagger}J}{120}+\frac{(\Delta t)^{2}J^{\dagger}JJ^{\dagger}J^{2}}{24}+\\
 & (\Delta t)^{2}\left(-\frac{iJJ^{\dagger}JH_{S}}{24}+\frac{iJJ^{\dagger}H_{S}J}{24}+\frac{iJ^{\dagger}JH_{S}J}{24}-\frac{iJ^{\dagger}J^{2}H_{S}}{24}+\frac{iJ^{\dagger}H_{S}J^{2}}{12}-\frac{iH_{S}J^{\dagger}J^{2}}{12}\right)\bigg\};\\
S_{2}^{(3)} & =\frac{\sqrt{2}\left(J^{2}-\frac{\Delta tJJ^{\dagger}J^{2}}{6}-\frac{\Delta tJ^{\dagger}J^{3}}{6}+2\Delta t\left(-\frac{iJH_{S}J}{6}+\frac{iJ^{2}H_{S}}{12}+\frac{iH_{S}J^{2}}{12}\right)\right)}{2};\\
S_{3}^{(3)} & =\frac{\sqrt{6}\sqrt{\Delta t}J^{3}}{6};\\
S_{4}^{(3)} & =-\frac{\sqrt{3}\sqrt{\Delta t}JJ^{\dagger}J}{12}+\frac{\sqrt{3}\sqrt{\Delta t}J^{\dagger}J^{2}}{12}+\frac{\sqrt{3}i\sqrt{\Delta t}\left(-JH_{S}+H_{S}J\right)}{6}.
\end{split}
\end{equation}

After this modification, the error remainders can be written as 
\begin{equation}
\begin{aligned}\mathcal{G}^{(3)}_{\mathrm{dia}}(\rho_S) & =-\frac{i\rho_{S} H_{S}J^{\dagger}JJ^{\dagger}J}{120}+\frac{i\rho_{S} J^{\dagger}{}^{2}H_{S}J^{2}}{9}-\frac{i\rho_{S} J^{\dagger}{}^{2}JH_{S}J}{18}-\\
 & \frac{i\rho_{S} J^{\dagger}H_{S}J^{\dagger}J^{2}}{18}+\frac{i\rho_{S} J^{\dagger}H_{S}JJ^{\dagger}J}{180}+\frac{i\rho_{S} J^{\dagger}JH_{S}J^{\dagger}J}{180}+\\
 & \frac{i\rho_{S} J^{\dagger}JJ^{\dagger}H_{S}J}{180}-\frac{i\rho_{S} J^{\dagger}JJ^{\dagger}JH_{S}}{120}+\frac{iH_{S}J^{\dagger}JJ^{\dagger}J\rho_{S}}{120}-\frac{iJ^{\dagger}{}^{2}H_{S}J^{2}\rho_{S}}{9}+\\
 & \frac{iJ^{\dagger}{}^{2}JH_{S}J\rho_{S}}{18}+\frac{iJ^{\dagger}H_{S}J^{\dagger}J^{2}\rho_{S}}{18}-\frac{iJ^{\dagger}H_{S}JJ^{\dagger}J\rho_{S}}{180}-\frac{iJ^{\dagger}JH_{S}J^{\dagger}J\rho_{S}}{180}-\\
 & \frac{iJ^{\dagger}JJ^{\dagger}H_{S}J\rho_{S}}{180}+\frac{iJ^{\dagger}JJ^{\dagger}JH_{S}\rho_{S}}{120}.
\end{aligned}
\end{equation}
This error component at order $\mathcal{O}((\Delta t)^{3})$
can be finally canceled by introducing new terms to $H_{0}$. Thus,
we obtain the third-order dilated Hamiltonian $H_{\mathrm{dia}}^{(3)}$
as 
\begin{equation}
\begin{split}H_{0}^{(3)} & =H_{S}+\Delta t\bigg(-\frac{[H_{S},J^{\dagger}]J}{12}-\frac{J^{\dagger}[J,H_{S}]}{12}\bigg)+(\Delta t)^{2}\bigg(-\frac{J^{\dagger}J[J^{\dagger}J,H_{S}]}{360}-\frac{[H_{S},J^{\dagger}J]J^{\dagger}J}{360}+\\
 & \frac{J^{\dagger}JJ^{\dagger}[H_{S},J]}{180}+\frac{[J^{\dagger},H_{S}]JJ^{\dagger}J}{180}-\frac{J^{\dagger}[H_{S},J^{\dagger}]J^{2}}{18}-\frac{J^{\dagger}{}^{2}[J,H_{S}]J}{18}\bigg);\\
S_{1}^{(3)} & =\frac{1}{\sqrt{\Delta t}}\bigg[J-\frac{\Delta tJJ^{\dagger}J}{12}-\frac{\Delta tJ^{\dagger}J^{2}}{4}-\frac{(\Delta t)^{2}JJ^{\dagger}JJ^{\dagger}J}{120}+\frac{(\Delta t)^{2}J^{\dagger}JJ^{\dagger}J^{2}}{24}+\\
 & (\Delta t)^{2}\left(-\frac{iJJ^{\dagger}[J,H_{S}]}{24}+\frac{iJ^{\dagger}J[H_{S},J]}{24}+\frac{i[J^{\dagger},H_{S}]J^{2}}{12}\right)\bigg];\\
S_{2}^{(3)} & =\sqrt{1/2}\left(J^{2}-\frac{\Delta tJJ^{\dagger}J^{2}}{6}-\frac{\Delta tJ^{\dagger}J^{3}}{6}+2\Delta t\left(\frac{iJ[J,H_{S}]}{12}+\frac{i[H_{S},J]J}{12}\right)\right);\\
S_{3}^{(3)} & =\frac{\sqrt{6}\sqrt{\Delta t}J^{3}}{6};\\
S_{4}^{(3)} & =-\frac{\sqrt{3}\sqrt{\Delta t}JJ^{\dagger}J}{12}+\frac{\sqrt{3}\sqrt{\Delta t}J^{\dagger}J^{2}}{12}+\frac{\sqrt{3}i\sqrt{\Delta t}\left(-[J,H_{S}]\right)}{6}.
\end{split}
\label{eq:SM_third_order_dilated_H_single}
\end{equation}
The dilation error $\mathcal{R}_{\mathrm{dia}}$ is at the order $\mathcal{O}((\Delta t)^{4})$, i.e., $\mathcal{G}_\mathrm{dia}^{(p)}=0$ for $p\leq 3$.

\subsection{Many-body case}

Now, we can go back to the dilated Hamiltonian for Eq.(\ref{eq:SM_set_up}),
i.e., the many-body case. Here, we introduce $N-1$ ancillas, i.e.,
we introduce a dedicated ancilla for each local jump operator $J_{i,i+1}$.
We assume that the dilated Hamiltonian $H_{\mathrm{dia}}$ has the
following form 
\begin{equation}
H_{\mathrm{dia}}=H_{0}+\sum_{i=1}^{N-1}(\sum_{j}S_{j,i}\otimes\ket{j}_{i}\bra{0}+\text{h.c.}),
\end{equation}
where we use $\ket{j}_{i}$ to denote the $j^\text{th}$ excited level of
the $i^\text{th}$ ancilla. $S_{j,i}$ is a system operator which should
be local and surrounds the physical site $i$. Inspired by Eq. (\ref{eq:SM_third_order_dilated_H_single}),
we first try this construction for the third-order dilated Hamiltonian
$H_{\mathrm{dia}}^{(3)}$ as
\begin{equation}
\begin{split}H_{0}^{(3)} & =H_{S}+\sum_{i=1}^{N-1}\bigg[\Delta t\bigg(-\frac{[H_{S},J_{i,i+1}^{\dagger}]J_{i,i+1}}{12}-\frac{J_{i,i+1}^{\dagger}[J_{i,i+1},H_{S}]}{12}\bigg)+\\
 & (\Delta t)^{2}\bigg(-\frac{J_{i,i+1}^{\dagger}J_{i,i+1}[J_{i,i+1}^{\dagger}J_{i,i+1},H_{S}]}{360}-\frac{[H_{S},J_{i,i+1}^{\dagger}J_{i,i+1}]J_{i,i+1}^{\dagger}J_{i,i+1}}{360}+\\
 & \frac{J_{i,i+1}^{\dagger}J_{i,i+1}J_{i,i+1}^{\dagger}[H_{S},J_{i,i+1}]}{180}+\frac{[J_{i,i+1}^{\dagger},H_{S}]J_{i,i+1}J_{i,i+1}^{\dagger}J_{i,i+1}}{180}-\\
 & \frac{J_{i,i+1}^{\dagger}[H_{S},J_{i,i+1}^{\dagger}]J_{i,i+1}^{2}}{18}-\frac{J_{i,i+1}^{\dagger}{}^{2}[J_{i,i+1},H_{S}]J_{i,i+1}}{18}\bigg)\bigg];\\
S_{j,i}^{(3)}= & S_{j}^{(3)}\ \mathrm{from\ Eq.(\ref{eq:SM_third_order_dilated_H_single})\ by\ replacing}\ J\to J_{i,i+1},J^{\dagger}\to J_{i,i+1}^{\dagger}.
\end{split}
\end{equation}
This construction leads to the following error remainder
\[
\begin{aligned}
\mathcal{G}^{(p)}_{\mathrm{dia}}= &0 \text{  for  }p\leq 2;\\
\mathcal{G}^{(3)}_{\mathrm{dia}}(\rho_S)= & i\sum_{i<j}\bigg(\frac{\rho_{S} H_{S}J_{i,i+1}^{\dagger}J_{i,i+1}J_{j,j+1}^{\dagger}J_{j,j+1}}{180}-\frac{\rho_{S} J_{j,j+1}^{\dagger}H_{S}J_{i,i+1}^{\dagger}J_{i,i+1}J_{j,j+1}}{90}+\\
 & \frac{\rho_{S} J_{j,j+1}^{\dagger}J_{j,j+1}H_{S}J_{i,i+1}^{\dagger}J_{i,i+1}}{180}-\frac{\rho_{S} J_{i,i+1}^{\dagger}H_{S}J_{i,i+1}J_{j,j+1}^{\dagger}J_{j,j+1}}{90}-\\
 & \frac{\rho_{S} J_{i,i+1}^{\dagger}J_{j,j+1}^{\dagger}J_{j,j+1}H_{S}J_{i,i+1}}{90}+\frac{\rho_{S} J_{i,i+1}^{\dagger}J_{i,i+1}H_{S}J_{j,j+1}^{\dagger}J_{j,j+1}}{180}-\\
 & \frac{\rho_{S} J_{i,i+1}^{\dagger}J_{i,i+1}J_{j,j+1}^{\dagger}H_{S}J_{j,j+1}}{90}+\frac{\rho_{S} J_{i,i+1}^{\dagger}J_{i,i+1}J_{j,j+1}^{\dagger}J_{j,j+1}H_{S}}{180}-\\
 & \frac{H_{S}J_{i,i+1}^{\dagger}J_{i,i+1}J_{j,j+1}^{\dagger}J_{j,j+1}\rho_{S}}{180}+\frac{J_{j,j+1}^{\dagger}H_{S}J_{i,i+1}^{\dagger}J_{i,i+1}J_{j,j+1}\rho_{S}}{90}-\\
 & \frac{J_{j,j+1}^{\dagger}J_{j,j+1}H_{S}J_{i,i+1}^{\dagger}J_{i,i+1}\rho_{S}}{180}+\frac{J_{i,i+1}^{\dagger}H_{i,i+1}J_{i,i+1}J_{j,j+1}^{\dagger}J_{j,j+1}\rho_{S}}{90}-\\
 & \frac{J_{i,i+1}^{\dagger}J_{j,j+1}^{\dagger}H_{S}J_{i,i+1}J_{j,j+1}\rho_{S}}{45}+\frac{J_{i,i+1}^{\dagger}J_{j,j+1}^{\dagger}J_{j,j+1}H_{S}J_{i,i+1}\rho_{S}}{90}-\\
 & \frac{J_{i,i+1}^{\dagger}J_{i,i+1}H_{S}J_{j,j+1}^{\dagger}J_{j,j+1}\rho_{S}}{180}+\frac{J_{i,i+1}^{\dagger}J_{i,i+1}J_{j,j+1}^{\dagger}H_{S}J_{j,j+1}\rho_{S}}{90}+\\
 & \frac{\rho_{S} J_{i,i+1}^{\dagger}J_{j,j+1}^{\dagger}H_{S}J_{i,i+1}J_{j,j+1}}{45}-\frac{J_{i,i+1}^{\dagger}J_{i,i+1}J_{j,j+1}^{\dagger}J_{j,j+1}H_{S}\rho_{S}}{180}\bigg),
\end{aligned}
\]
where the factor $i$ in front of $\sum$ is the imaginary
unity, which should not be confused with the summation index $i$.
These error components at order $\mathcal{O}((\Delta t)^{3})$
can be cancelled by introducing additional terms into $H_{0}$. Therefore,
we can obtain the final expression for the third-order dilated Hamiltonian
$H_{\mathrm{dia}}^{(3)}$ as
\begin{equation}
\begin{split}H_{0}^{(3)} & =H_{S}+\sum_{i=1}^{N-1}\bigg[\Delta t\bigg(-\frac{[H_{S},J_{i,i+1}^{\dagger}]J_{i,i+1}}{12}-\frac{J_{i,i+1}^{\dagger}[J_{i,i+1},H_{S}]}{12}\bigg)+\\
 & (\Delta t)^{2}\bigg(-\frac{J_{i,i+1}^{\dagger}J_{i,i+1}[J_{i,i+1}^{\dagger}J_{i,i+1},H_{S}]}{360}-\frac{[H_{S},J_{i,i+1}^{\dagger}J_{i,i+1}]J_{i,i+1}^{\dagger}J_{i,i+1}}{360}+\\
 & \frac{J_{i,i+1}^{\dagger}J_{i,i+1}J_{i,i+1}^{\dagger}[H_{S},J_{i,i+1}]}{180}+\frac{[J_{i,i+1}^{\dagger},H_{S}]J_{i,i+1}J_{i,i+1}^{\dagger}J_{i,i+1}}{180}-\\
 & \frac{J_{i,i+1}^{\dagger}[H_{S},J_{i,i+1}^{\dagger}]J_{i,i+1}^{2}}{18}-\frac{J_{i,i+1}^{\dagger}{}^{2}[J_{i,i+1},H_{S}]J_{i,i+1}}{18}\bigg)\bigg]+\\
 & (\Delta t)^{2}\sum_{i<j}\bigg(\frac{J_{i,i+1}^{\dagger}[[J_{i,i+1},H_{S}],J_{j,j+1}^{\dagger}]J_{j,j+1}}{180}+\frac{J_{j,j+1}^{\dagger}J_{i,i+1}^{\dagger}[J_{j,j+1},[J_{i,i+1},H_{S}]]}{180}+\\
 & \frac{J_{j,j+1}^{\dagger}[J_{j,j+1},[H_{S},J_{i,i+1}^{\dagger}]]J_{i,i+1}}{180}+\frac{[[H_{S},J_{i,i+1}^{\dagger}],J_{j,j+1}^{\dagger}]J_{i,i+1}J_{j,j+1}}{180}\bigg)\\
S_{j,i}^{(3)}= & S_{j}^{(3)}\ \mathrm{from\ Eq.(\ref{eq:SM_third_order_dilated_H_single})\ by\ replacing}\ J\to J_{i,i+1},J^{\dagger}\to J_{i,i+1}^{\dagger}.
\end{split}
\label{eq:SM_many_Body_third_order_dila}
\end{equation}
From the commutator form, it is obvious that all the terms in $H_{\mathrm{dia}}^{(3)}$
is geometrically local. At this stage, the dilation error satisfies
\begin{equation}
\mathcal{G}^{(p)}_\mathrm{dia}=0\text{  for  }p\leq 3,
\end{equation}
or equivalently
\[
\mathcal{R}_\mathrm{dia}=\mathcal{O}((\Delta t)^4).
\]

\subsection{Error scaling with $N$}

Now, we will further prove that $\mathcal{G}_\mathrm{dia}^{(4)}$ grows lineary in $N$, i.e., $\lVert\mathcal{G}_\mathrm{dia}^{(4)}(\rho_S)\rVert_\mathrm{tr}=\mathcal{O}(N)$.
To this end, we directly compute $\mathcal{G}_{\mathrm{dia}}^{(4)}$. We cannot
show the explicit form of $\mathcal{G}^{(4)}_{\mathrm{dia}}$ here because
there are thousands of terms appearing. Shortly speaking, we divide
the error components into five groups based on the number of $H_{S}$
they contain. We analyze each group separately by a brute-force calculation.

\textbf{Group 1: No $H_{0}$ contained.}

This group is equivalent to the setting where the simulated Lindbladian
only includes the jump operator. By the assumption that each jump
operator commute with others, the overall evolution factorizes into
a product of the individual channels generated by each jump operator.
This is true for both the Lindbladian dynamics itself or dynamics
from the dilated Hamiltonian. Thus, the total error is simply the
sum of the errors contributed by each jump operator, yielding an overall
scale $\mathcal{O}(N)$ in $\mathcal{G}_\mathrm{dia}^{(4)}$.

\textbf{Group 2: One $H_{0}$ contained.}

Each error component in this group is constructed from the product
of $\{J_{i,i+1},J_{j,j+1},J_{k,k+1},H_{m,m+1}\}$ as well as their
Hermitian conjugate. Here $H_{m,m+1}$ is the local Hamiltonian acting
on sites $m,m+1$. By a brute-force computation, we can show that
if any jump operator in $\{J_{i,i+1},J_{j,j+1},J_{k,k+1}\}$ commutes
with $H_{m,m+1}$, the contribution to $\mathcal{G}_{\mathrm{dia}}^{(4)}$
is $0$. Thus, the non-vanishing contribution is only from the pairs
of integers $\{i,j,k,m\}$ where $\{i,j,k\}$ are both close to $m$
such that $[J_{\alpha,\alpha+1},h_{m}]\neq0$ for $\forall\alpha\in\{i,j,k\}$.
Therefore, the error contribution to $\mathcal{G}_\mathrm{dia}^{(4)}$ from this group grows as $\mathcal{O}(N)$.

\textbf{Group 3: Two $H_{0}$ contained.}

Each error component in this group is constructed from the product
of $\{J_{i,i+1,},J_{j,j+1},H_{k,k+1},H_{m,m+1}\}$ as well as their
Hermitian conjugate. First, we can show that, by a brute-force calculation,
if any operator from the set $\{J_{i,i+1,},J_{j,j+1},H_{k,k+1},H_{m,m+1}\}$
commutes with the rest of the set except its Hermitian conjugate,
the contribution to $\mathcal{G}^{(4)}_{\mathrm{dia}}$ is $0$. This restricts
the form of pairs of integers $\{i,j,k,m\}$ for possibly non-vanishing
error contribution. Concretely, any integer from $\{i,j,k,m\}$ can
not be far away from the remaining simultaneously. Furthermore, if
$i$ is close to $j$ but far away from $k,m$, the contribution to $\mathcal{G}_\mathrm{dia}^{(4)}$ is
still $0$ as we always have $[J_{i,i+1},J_{j,j+1}]=0$. Thus, without
loss of generality, for non-vanishing contribution we can assume $i$
is close to $k$ such that $[J_{i,i+1},H_{k,k+1}]\neq0$. If $j$
is also close to $k$, the non-vanishing contribution requires $m$
also close to $k$. Therefore, the integers $\{i,j,k,m\}$ should
be all close to each other to get a non-vanishing contribution to $\mathcal{G}_\mathrm{dia}^{(4)}$, which
scales as $\mathcal{O}(N)$. On the other
hand, if $j$ is close to $m$ but far from $k$ such that $[J_{i,i+1},H_{m,m+1}]=[L_{j,j+1},H_{k,k+1}]=[H_{m,m+1},H_{k,k+1}]=0$,
we can show that the contribution to $\mathcal{G}_\mathrm{dia}^{(4)}$ is still $0$. Thus, $j$ must be
also close to $k$, so as $m$. In summary, the total contribution to $\mathcal{G}_\mathrm{dia}^{(4)}$
from this group grows as $\mathcal{O}(N)$ 

\textbf{Group 4: Three $H_{0}$ contained.}

For this group, we can directly write the error component into the
nested commutator as
\begin{equation}
\begin{aligned}\frac{i(\Delta t)^{4}}{720}\sum_{i=1}^{N-1}\bigg\{-[L_{i}^{\dagger},H_{S}][[L_{i},H_{S}],H_{S}]+[[L_{i}^{\dagger},H_{S}],H_{S}][L_{i},H_{S}]-\\
L_{i}^{\dagger}\bigg[[H_{S},[L_{i},H_{S}]],H_{S}\bigg]-\bigg[[H_{S},[H_{S},L_{i}^{\dagger}]],H_{S}\bigg]L_{i}\bigg\}\rho_{S}+\text{h.c.}.
\end{aligned}
\end{equation}
Due to the nested-commutator form, this error term grows as $\mathcal{O}(N(\Delta t)^{4})$. Its contribution to $\mathcal{G}_\mathrm{dia}^{(4)}$ is $\mathcal{O}(N)$.

\textbf{Group 5: Four $H_{0}$ contained.}

It is obvious this group does not contribute to error, as the dilated
Hamiltonian can exactly simulate the system dynamics if no jump operator
appears.

Notice that the above analysis also holds if $\rho_S$ is a joint density matrix. In conclusion, for the third-order dilated Hamiltonian $H_{\mathrm{dia}}^{(3)}$
in Eq. (\ref{eq:SM_many_Body_third_order_dila}), the dilation error
scales as 
\begin{equation}
\mathcal{G}_\mathrm{dia}^{(p)}=0\text{  for  }p\leq 3\text{  and  } \lVert\mathcal{G}_\mathrm{dia}^{(4)}\rVert_{\diamond}=\mathcal{O}(N).
\end{equation}
\end{document}